%% file: paper.tex
\documentclass[a4paper,onecolumn,superscriptaddress,11pt,dvipsnames,svgnames,noarxiv,accepted=2025-12-15]{quantumarticle}

\pdfoutput=1
\usepackage{iftex}
\usepackage{underscore}         
\usepackage[T1]{fontenc}        
\usepackage[utf8]{inputenc}
\usepackage[english]{babel}
\usepackage{amsmath}
\usepackage{amssymb}
\usepackage{braket}
\usepackage{graphicx}
\usepackage{doi}
\usepackage{float}
\usepackage{bbold}
\usepackage{amsthm}
\usepackage[dvipsnames]{xcolor}
\usepackage{mathtools}
\usepackage{hyperref}
\usepackage{thm-restate}
\usepackage{thmtools}
\usepackage[capitalise]{cleveref}
\usepackage{tikz}
\usepackage{tikzfig}
\usepackage{standalone}
\usepackage[numbers,sort&compress]{natbib}

\Crefname{thm}{Theorem}{Theorems}
\crefname{thm}{Thm.}{Thms.}
\Crefname{prop}{Proposition}{Propositions}
\crefname{prop}{Prop.}{Props.}
\Crefname{defi}{Definition}{Definitions}
\crefname{defi}{Def.}{Defs.}
\Crefname{prop}{Proposition}{Propositions}
\crefname{lemma}{Lemma}{Lemmas}
\crefname{coro}{Cor.}{Cors.}
\Crefname{coro}{Corollary}{Corollaries}
\crefname{exampe}{Ex.}{Exs.}
\Crefname{exampe}{Example}{Examples}
\crefname{remark}{Remark}{Remarks}
\crefname{conj}{Conj.}{Conjs.}
\Crefname{conj}{Conjecture}{Conjectures}

\newcommand{\MINO}{\ensuremath{\left[\!\begin{smallmatrix}M & N\end{smallmatrix}\!\right]}}
\newcommand{\MINOtype}[2]{\ensuremath{\left[\!\begin{smallmatrix}#1 & #2\end{smallmatrix}\!\right]}}
\usepackage{soul}

\title{A decompositional framework for process theories in spacetime}
\author{Matthias Salzger}
\email{matthias.salzger@outlook.com}
\affiliation{International Centre for Theory of Quantum Technologies, University of Gda\'nsk, Jana Ba\.zy\'nskiego 1A,
80-309 Gda\'nsk, Poland}
\author{John H. Selby}
\email{john.h.selby@gmail.com}
\affiliation{International Centre for Theory of Quantum Technologies, University of Gda\'nsk, Jana Ba\.zy\'nskiego 1A,
80-309 Gda\'nsk, Poland}
\affiliation{Theoretical Sciences Visiting Program, Okinawa Institute of Science and
Technology Graduate University, Onna, 904-0495, Japan}

\input{preamble}
\input{minreps.tikzstyles}

\begin{document}

\maketitle
\begin{abstract} 
There has been a recent surge of interest within the field of quantum foundations regarding incorporating ideas from general relativity and quantum gravity. However, many quantum information tools remain agnostic to the underlying spacetime. For instance, whenever we draw a quantum circuit the effective spacetime imposed by the connectivity of the physical qubits which will realize this circuit is not taken into account. In this work, we aim to address this limitation by extending the framework of process theories to include a background spacetime structure. We introduce the notion of process implementations, i.e., decompositions of a process. A process is then embeddable if and only if one of its implementations can be embedded in such a way that all the component processes are localized and all wires follow timelike paths. While conceptually simple, checking for embeddability is generally computationally intractable. We therefore work towards simplifying this problem as much as possible, identifying a canonical subset of implementations that determine both the embeddability of a process and the causal structures distinguishable at least in some process theory. Notably, we discover countably infinite ``zigzag'' causal structures beyond those typically considered. While these can be ignored in classical theory, they seem to be essential in quantum theory, as the quantum CNOT gate can be implemented by all zigzag structures but not in a standard causal structure, except in the trivial undecomposed way. These zigzags could be significant for quantum causal modeling and the study of novel quantum resources.
\end{abstract}
\newpage
\tableofcontents
\newpage
\section{Introduction}
Process theories provide an important framework for studying the foundations of physics as well as applied problems in quantum information theory \cite{coecke2015categorical,coecke2018picturing}. A process theory is comprised of two basic components, physical systems, which are denoted by labelled wires, and physical processes, which are denoted by labelled boxes with input systems at the bottom and output systems at the top. These can be wired together to form diagrams such as
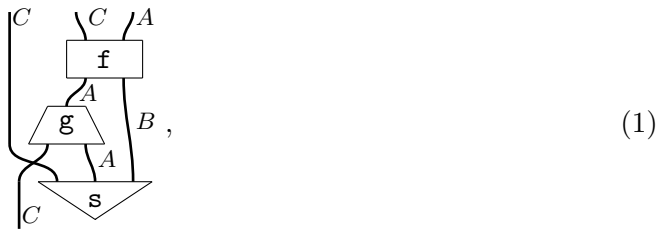
\begin{equation}
\input{Diagrams/diagram.tikz},
\end{equation}
which themselves must be processes in the theory. That is, a process theory is a collection of processes which are closed under forming diagrams. What kinds of systems and processes we care about depends on the physics that we are trying to capture, but commonly, in the context of quantum information, we consider them to be finite dimensional quantum systems and completely positive trace preserving maps, and in the case of classical information, finite sets and stochastic maps.

It is natural to want to interpret these diagrams as living inside some background spacetime, and, in particular, to imagine that we have an arrow of time going from bottom to top according to the input-output structure of the processes. That is, we may want to think of the wires as representing world-lines of the systems and the processes as localised to particular spacetime points. Indeed, often it is very useful to make this interpretation of processes, in particular when thinking about information processing protocols where, for example, we rely on spacelike separation of various parties to ensure security in some cryptographic task \cite{kent2012unconditionally, kaniewski2013secure, lunghi2015practical, adlam2015deterministic, buhrman2014position, unruh2014quantum, vilasini2019composable}.
Such an interpretation, however, is not always valid. For example, process theories can represent physical systems which are not localised in space, such as quantum fields \cite{gogioso2021functorial} or higher-order processes in terms of time-delocalised subsystems \cite{oreshkov2019time}. In such cases it no longer makes sense to view the diagram as a depiction of a process in some background spacetime as the systems no longer trace out world-lines but some higher dimensional object. What this tells us, is that we have nothing built in to the mathematical formalism of process theories to enable us to talk (directly) about when and where these processes are taking place.

Of course, it is however, very useful to be able to talk about such things, hence in this paper we begin the project of extending process theories to also incorporate a background spacetime. We are not the first to have done such a thing, however, existing approaches have tried to capture fundamental physical theories (such as field theories) \cite{gogioso2021functorial}, whereas here we are taking a more pragmatic approach tailored towards information processing tasks. This leads us to a related but distinct formalism. An important consequence of this perspective, is that we are not necessarily interested in the fundamental spacetime describing our universe, but also in effective spacetimes which might be provided, for example, by a quantum communication network or the physical architecture of a quantum computer.

The focus of this paper is on the question of whether or not a given process in a process theory can be embedded in a distributed way in a given spacetime. A key example is the Popescu-Rohrlich (PR) box \cite{prbox}. It can be thought of as a classical process with two binary inputs and two binary outputs, which achieves perfect correlation but without signalling. We can easily implement this locally at a point, that is, if the two parties are free to communicate, but it becomes impossible (even in quantum theory) to do in certain delocalised ways, in particular, if the two parties are spacelike separated. Previous definitions regarding embeddability of processes within a spacetime, see, for example, Refs.~\cite{vilasini2022embedding, causalbox}, have focused only on the signalling structure of the processes, while the above example of the PR box highlights a deficiency in this perspective. That is, we find that these previous definitions captured only a necessary rather than a sufficient condition\footnote{We note that a necessary condition is appropriate for no-go theorems (the focus of \cite{vilasini2022embedding}).}. Additionally, there are process theories in which one might not even have a well-defined notion for signalling. This is the case in non-terminal process theories (also referred to as non-causal process theories) \cite{dariano2014determinism}. While it is natural to assume a terminal process theory when one is talking about spacetime, being able to study how non-terminal theories would behave in spacetime would nevertheless be interesting. 

The contributions of this paper can be summarised as follows: in \cref{sec:embedPR}, we discuss the issue with no-signalling as a principle for embeddability mentioned in the previous paragraph in detail. This motivates us to define embeddability in terms of the decompositional structure of processes instead in \cref{sec:implementations}. That is we embed a process by considering decompositions of it, choosing one and then localising each constituent process of this decomposition at a point in spacetime such that wires can be assumed to follow timelike paths. We call these decompositions implementations, the intuition being that a process represents some abstract protocol and the decompositions/implementations represent the circuits/procedures that implement them. This then formally gives us a way to convert ``process theories'' into ``processes theories living in spacetime'' up to the question of how these processes living in spacetime compose exactly, which will be the subject of a future work. 

The set of all implementations of a process, however, is in general an unwieldy object, which makes it difficult to decide when a process is actually embeddable. We thus turn our focus towards reducing the amount of information we need to know about the implementations of a process to always be able to decide embeddability. In \cref{sec:ReducedDescription}, we do this from a theory-independent perspective. First, we show that we only need the order-theoretic information about the implementations and that there exists a preorder over implementations, where implementations higher up are easier to embed and the implementation set of a process essentially forms a downward-closed set\footnote{A downward-closed set $X \subseteq Y$ has the property that for all $x, y \in Y$ such that $x \succ y$, if $x \in X$, then also $y \in X$.}. Then, we show that the equivalence classes of the aforementioned preorder have strictly minimal elements, which we can choose as canonical (or minimal) representative. We thus show that it suffices to know which minimal representatives correspond to implementations of the process. They thus completely characterise the theory-independent setting. 

To go further, we consider specific theories in \cref{sec:thspec}. We argue that in a theory-specific setting, some minimal representatives become irrelevant, because every process that can be implemented by one of them can also be implemented by one which is higher up in the preorder. We characterise the set of relevant minimal representatives for terminal theories, listing them explicitly for processes with few inputs and outputs. There we find that already for processes  with two inputs and two outputs, there exist countably infinite relevant minimal representatives, which we call zigzag minimal representatives due to their shapes, beyond the ones corresponding to the usual causal structures. We then find the set of minimal representatives for cd theories (i.e., theories with an appropriate notion of copying), which include classical theory, recovering a well-known result from classical causal modelling that unobserved nodes in a causal model can be assumed to be parentless variables \cite{evans2018margins}. Finally, we discuss the case of quantum theory, where we show that any Clifford gate is implemented by all of the zigzag minimal representatives. In addition, we show that the quantum CNOT gate cannot be implemented by a minimal representative that is higher up in the preorder. Hence, the zigzag minimal representative represent a set of causal structures that should not be ignored in quantum theory. Finally, we analyse some results on quantum theory from the literature \cite{dolev2019constraining, lorenz2021causal, renner2023commuting} from the point of view of our results. 

\subsection{Terminal theories and signalling condition}

We briefly review two important concepts for this work, which we alluded to in the introduction, terminal theories \cite{chiribella2010probabilistic,coecke2013causal,coecke2014terminality} and the (no-)signalling condition. For a more complete introduction to process theories we refer to \cref{app:PTs}.

\begin{restatable}{defi}{causal}\label{causal}
A process theory is terminal if there exists a unique effect called discard for each system which satisfies the following conditions
\begin{equation}\label{causalconditions}
\input{Diagrams/discardprocess.tikz} \quad = \quad \input{Diagrams/discard.tikz}\quad \text{ and } \quad \input{Diagrams/discardComposition.tikz}.
\end{equation}
\end{restatable}
We can cast quantum theory both as a terminal and as a non-terminal theory. In the former case, the processes are the completely positive trace preserving (CPTP) maps and the trace plays the role of the discard. In the latter case, the processes are the completely positive trace non-increasing maps. This is essentially equivalent to quantum theory with post-selection. When we refer to quantum theory in this work, we mean the terminal version of quantum theory. 

We note that in the framework of operational probabilistic theories (OPT) the terminality condition instead states there is a unique \textit{deterministic} effect \cite{chiribella2010probabilistic}. While this seems like a weaker requirement, both this and \cref{causal} are physically the same as they both formalise the same idea, namely that future measurements should not influence past outcomes. For more details, see \cref{app:PTs}.

In terminal theories, there is a simple condition for a process to be non-signalling. We say that a process
\begin{equation}
\input{Diagrams/generic22process.tikz}
\end{equation}
does not signal from $Y$ to $A$ if
\begin{equation}\label{sigcond}
\input{Diagrams/sigcond1.tikz} \quad = \quad \input{Diagrams/sigcond2.tikz}
\end{equation}
for some process $\mathtt{A}$. Note that the systems $X, Y, A, B$ may be composite systems.

We note here that the terminality condition is frequently referred to as the causality condition. We choose the less frequently used ``terminality here as terms like ``causal'', ``causality'' and ``causal structure'' can be ambiguous and confusing as different communities define and use them in different ways. In particular, the word ``causal structure'' can refer to the signalling structure of a process, the structure of a decomposition of a process, the lightcone structure of spacetime,... In this work, we will attempt to avoid this confusion as much as by avoiding terms like ``causal'', ``causality'' and ``causal structure'' as much as possible. 

\subsection{Motivating example: embeddability of PR boxes}\label{sec:embedPR}
As a motivating example, let us take a look at the PR box. We denote its input systems as $X, Y$ and its output output systems as $A, B$, all of which represent binary variables (whose values we will denote with $x, y, a, b$ respectively). The PR box can be written as a stochastic map between classical systems
\begin{equation}
\input{Diagrams/PRBox.tikz}
\end{equation}
defined by the conditional probability distribution $p_{\texttt{PR}}(a,b|x,y)= \frac{1}{2} \delta_{a\oplus b, xy}$. The PR box is a non-signalling resource. The input $X$ signals neither to the output $A$ nor to the output $B$ and similarly for $Y$ (however, both $X$ and $Y$ as well as the composite system $XY$ signals to the composite system $AB$).

Consider now Minkowski spacetime and a hypothetical embedding of the PR box in this spacetime as follows 
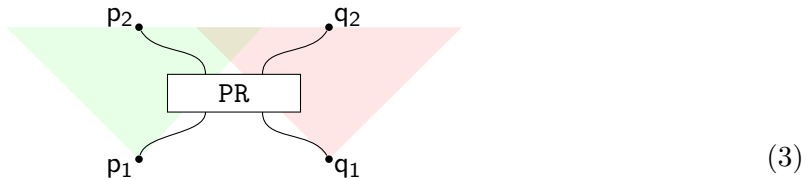
\begin{equation}\label{PRBoxImpossible}
\input{Diagrams/PRBoxinMink.tikz}
\end{equation}
where the points $\mathsf{p_1}, \mathsf{p_2}$ are spacelike separated from the points $\mathsf{q_1}, \mathsf{q_2}$ as indicated by the shaded light cones. Note here that we have only fixed the spacetime locations of the inputs and the outputs, while a priori we are not saying anything about where the ``box'' itself is located in spacetime (if it is embedded into spacetime at all). What we would like to answer is, \emph{is such an embedding consistent with a) the structure of the spacetime? and b) the particular process theory that we are considering?}

If we only require that the no-signalling principle is respected (such as is done in Ref.~\cite{vilasini2022embedding}, for example) such an embedding would be deemed consistent. Specifically, in this paradigm we can view an embedding as a map $\mathcal{E}$ from the inputs and outputs to the spacetime. We then say this embedding is consistent if and only if for any subset of inputs $S_I \subseteq \{X, Y\}$ that can signal to a subset of outputs $S_O \subseteq \{A, B\}$, there exist $s_I \in S_I$ and $s_O \in S_O$ s.t. $\mathcal{E}(s_I) \leq \mathcal{E}(s_O)$. This is the case here since for any $S_I$, the set $S_O = \{A, B\}$ is the only set of outputs $S_I$ signals to. Satisfying this condition is essential to ensure that, should such PR boxes exist embedded into a spacetime in this way, they cannot be used to send superluminal signals. 

What we are interested in here, however, is whether it is possible for us to actually implement a PR box in this way. This requires not only that we specify where the inputs and outputs go in the spacetime, but where the processes go as well. Intuitively speaking, due to the spacetime structure, the process would have to decompose as 
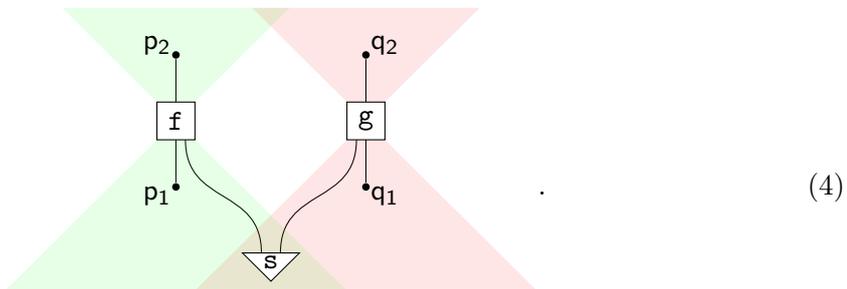
\begin{equation}\label{disjointprocess}
\input{Diagrams/PRBoxinMinkBell.tikz}.
\end{equation}
However, it is well known that the PR box cannot be realised in this way in quantum or classical theory \cite{prbox}. 

It is precisely this intuition which we capture and generalise in this paper. 

Indeed, the issue of \cref{PRBoxImpossible} and \cref{disjointprocess} is actually somewhat worse. Note that, within the no-signalling paradigm, the validity of the embedding in \cref{PRBoxImpossible} does not depend on whether there exists any common past of the named spacetime points. In fact, we could have a spacetime $\mathcal{M}$ consisting solely of two completely causally disconnected regions $\mathcal{M}_\mathtt{p}$ and $\mathcal{M}_\mathtt{q}$ such that $\mathtt{p}_1, \mathtt{p}_2 \in \mathcal{M}_{\mathtt{p}}$ and $\mathtt{q}_1, \mathtt{q}_2 \in \mathcal{M}_{\mathtt{q}}$. Here, with completely causally disconnected we mean that any two points $\mathtt{x} \in \mathcal{M}_\mathtt{p}, \mathtt{y} \in \mathcal{M}_\mathtt{q}$ are spacelike separated. Then, the embedding \cref{PRBoxImpossible} would still be valid as this changes nothing about the signalling relations. Thus, we would obtain perfect correlation in completely causally disconnected regions.

There is another aspect of the no-signalling paradigm that is potentially unsatisfactory. Note that at no point does the process theory come into play (beyond the fact that it is terminal). Indeed, the validity of the embedding \cref{PRBoxImpossible} depends only on the process itself and not on the underlying process theory. PR box correlations can be achieved in classical and quantum theory as well as boxworld \cite{barrett2007information}, but the former two require communication between the parties. Intuitively, the existence of an embedding \cref{PRBoxImpossible} should thus tell us that the underlying theory is boxworld (assuming there are no other candidates besides the aforementioned three) but we cannot make this statement by relying on the no-signalling principle alone. The same situation arises with Bell correlations. Again, by the non-signalling criterion, an embedding à la \cref{PRBoxImpossible} would be valid in both classical and quantum theory. Hence, our framework provides a formal tool for adjudicating between candidate theories of nature. In particular, whether or not some correlation is classical crucially depends on the underlying spacetime structure. 

\section{Embeddings from implementations}\label{sec:implementations}
Based on our discussion in the previous section, in order to understand whether or not a given process $\mathtt{F}$ within some process theory $\mathbf{Proc}$ can be implemented in a given way in a spacetime $\mathcal{M}$, it is not enough to think of the process as a monolithic block, for example,
\begin{equation}
\input{Diagrams/process.tikz},
\end{equation}
where $X_i, A_j$ are arbitrary systems (classical, quantum and/or post-quantum) of the theory.
Instead, we must consider the possible ways in which this process can be built out of other processes within the theory. For example, consider that the process $\mathtt{F}$ represents a process violating the standard Bell inequality, $\mathtt{F} = \mathtt{Bell}$. Then, in quantum theory it can be rewritten as 
\begin{equation}
\input{Diagrams/bell.tikz} \quad = \quad \input{Diagrams/belldecomp.tikz}
\end{equation}
where $\mathtt{s}$ represents a maximally entangled state and $\mathtt{f}, \mathtt{g}$ the appropriate measurements that allow Alice and Bob to violate the Bell inequality. 

Intuitively, this would allow us to embed the process $\mathtt{Bell}$ in a spacelike separated manner as in \cref{PRBoxImpossible}. Note that the above decomposition, however, would not be possible in classical theory.

To capture this, we introduce the notion of an implementation set for a given process within a given theory.
\begin{restatable}{defi}{def:implementationset}\label{def:implementationset} The implementation set $\mathcal{I}_\mathtt{F}$ of a process $\mathtt{F}\in \mathbf{Proc}$ is the set of diagrams, i.e., implementations, in $\mathbf{Proc}$ which are equal to $\mathtt{F}$. 
\end{restatable}
Note that in our formalism a process $\mathtt{F}$ comes equipped with a specified decomposition of the input into subsystems as this decomposition is crucial to deciding whether or not it is embeddable in spacetime\footnote{As we discuss in the appendix \ref{app:PTs}, this means that formally we consider a process theory to be a (coloured) Prop rather than a symmetric monoidal category, and that we can do so without any real loss of generality. }.

Note also that the implementation set for a process $\mathtt{F}$ is necessarily non-empty as, at a minimum, it must contain the process $\mathtt{F}$ itself. For example, for our Bell process in quantum theory, we would have an implementation set as
\begin{equation}\label{impset}
\mathcal{I}_\mathtt{Bell} = \left\{\input{Diagrams/bell.tikz},\ \ \input{Diagrams/belldecomp.tikz},\ \ \input{Diagrams/sequential1.tikz}, \input{Diagrams/sequential2.tikz}, \cdots \right\}.
\end{equation}
Intuitively, one should think about some $i_\mathtt{F} \in \mathcal{I}_\mathtt{F}$ as an experimental procedure, where each process in the diagram corresponds to an experimental component (or perhaps even a whole lab). In \cref{impset}, the second diagram would thus represent a genuine Bell violation, while the third and fourth represent a Bell violation exploiting the communication loophole, with either Alice or Bob being able to communicate to the other\footnote{Or more accurately, it represents an experimental realisation where we cannot \textit{exclude} the communication loophole. It could still be that the lower boxes $\mathtt{a}$ and $\mathtt{b'}$ actually represent the composition of a common cause and Alice/Bob's local lab, without the possibility of signalling.}. In order to embed in spacetime, we then localise each of these components at some spacetime location, in a way that is consistent with locality, i.e., such that wires can be thought of as following timelike paths. 

To formalise this, we conceptualise the spacetime $\mathcal{M}$ as a partial order. The elements of this partial order are spacetime points and the relations between these elements are given by the light cone structure. Indeed, as shown by Malament \cite{malament1977}, spacetime is little more than such a partial order. With the partial order between the spacetime points fixed, all that is left is a single metric parameter per spacetime point. Hence, the order relations capture spacetime up to a conformal mapping.

We can then also associate each possible implementation of $\mathtt{F}$ to a partial order with a little extra structure, which we will call (adapting the terminology of \cite{timetravel}) a \emph{framed} partial order (FPO). 

\begin{restatable}{defi}{fpodef}\label{fpodef}
A framed partial order is a triple $(S, \mathcal{I}, \mathcal{O})$ where $S$ is a finite partial order and $\mathcal{I}, \mathcal{O}$ are disjoint lists of elements in $S$ such that for all $I \in \mathcal{I}$ and all $O \in \mathcal{O}$, $I, O \in S$ and $I$ is minimal in $S$ and $O$ is maximal in $S$. We call $\mathcal{I}$ the set of inputs, $\mathcal{O}$ the set of outputs and both together the frame. Their elements are called frame elements and other elements of $S$ are called internal elements. We will generally treat $\mathcal{I}, \mathcal{O}$ to be implicit and denote the framed partial order simply with $S$.
\end{restatable}
For example, we have the following mapping from a diagram to a framed partial order
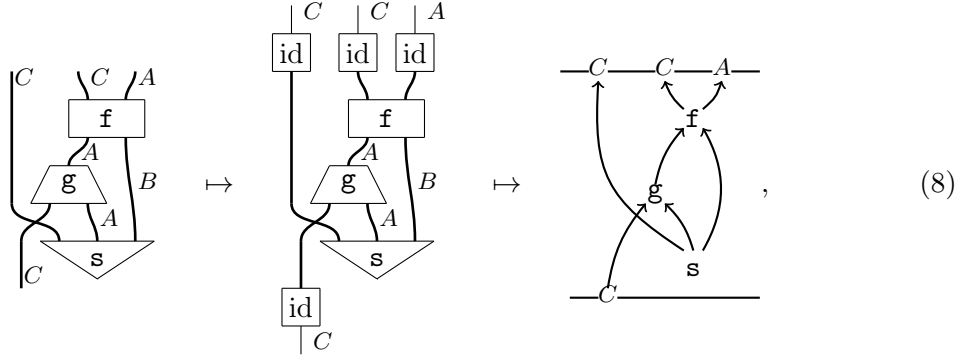
\begin{equation}\label{Gfig}
\input{Diagrams/diagram.tikz}\quad \mapsto \quad \input{Diagrams/fpodots.tikz} \quad \mapsto \input{Diagrams/fpodots2.tikz}, 
\end{equation}
where frame elements are red and internal, i.e., non-frame, elements are black dots and we added a pair of horizontal black lines to further indicate the frame. Note that the partial order is the one induced by either of the two graphs in the above equations (the arrow from $\mathtt{s}$ to $\mathtt{f}$, representing the corresponding wire in the diagram, in the graph in the middle is thus redundant and we can omit it, as is done in the graph to the right). 

We now give a more detailed definition of the conversion.
\begin{restatable}{defi}{G}\label{G}
The map $\mathcal{G}$ from the set of diagrams to the set of framed partial orders is defined in the following way: Given $i_\mathtt{F}$ a diagram, the input (output) elements of $\mathcal{G}(i_\mathtt{F})$ are in one-to-one correspondence to the inputs (outputs) of $i_\mathtt{F}$. The order of the two lists is the order of the corresponding inputs/outputs in the diagram. Additionally, the internal elements of $\mathcal{G}(i_\mathtt{F})$ are in one-to-one correspondence to boxes in the diagram $i_\mathtt{F}$ excluding boxes whose sole input is an input of the diagram or whose sole output is an output of the diagram. 


\end{restatable}
We illustrate this definition in the example below. We will also take this opportunity to simplify the notation we will use to depict FPOs a bit, removing the line on the frame and the explicit arrows, assuming implicitly that the arrows point upwards.
\begin{gather}\label{runningexample}
\begin{aligned}
\input{Diagrams/gexample1.tikz}  \mapsto \quad \input{Diagrams/gexample4.tikz}
\end{aligned}
\end{gather}
We explain in the detail why each part of the diagram has or does not have a corresponding element in the FPO. In the diagram on the left hand side we have the four inputs $I_i$ and four outputs $O_i$ which show up on the right hand side as the four corresponding input and four corresponding output elements. The ordering of the list of inputs/outputs is provided by the left-to-right order of these elements in the diagrammatic depiction (hence, swapping the position of $O_1$ and $O_2$ would yield a different FPO corresponding to a different diagram. On the other hand, the internal elements can be freely moved around). We see that the FPO only has two internal elements $\mathtt{k}, \mathtt{h}$. There are no internal elements for the other boxes as  their sole input/output is an input/output of the diagram. For $\mathtt{f}$, we see that it has only one input which is also an input of the diagram and only one output which is an output of the diagram. The box $\mathtt{g}$ has multiple outputs but its sole input is the diagram input $I_2$. The box $\mathtt{l}$ has a single input which is not a diagram input, but also only a single output which is the diagram output $O_3$. On the other hand, $\mathtt{k}$ has multiple inputs and while it has only one output this is not an output of the diagram. The process $\mathtt{h}$ has multiple inputs and outputs and is thus also not excluded.

We note that we could have also defined $\mathcal{G}(i_\mathtt{F})$ to have an internal element for each box, including ones whose sole input is an input of the diagram or whose sole output is an output of the diagram. All results would still hold in an appropriately modified form even with this alternate definition. However, the definition in \cref{fpodef} allows us to keep FPOs a bit more compact and to represent later results in a cleaner way. 

We will usually abuse notation by using the same names for systems and processes in the diagram as for elements in the framed partial order. This way it is easy to see which element in the partial order corresponds to which system or process in the diagram. On the other hand, as we will see in \cref{sec:ReducedDescription}, the labels of the elements in the FPO are essentially arbitrary anyway (intuitively, the important information for embeddability in spacetime should be the relationships between boxes/events, i.e., how they are connected, and not their names).

With these definitions in place we can now formally ask what it means for a process to be embedded in a particular way in spacetime. For this purpose, we need order-preserving maps. A map $\mathcal{E}: S_1 \rightarrow S_2$ where $S_1, S_2$ are partial orders is called order-preserving if for all $x, y \in S_1$, it holds that $x \leq y \implies \mathcal{E}(x) \leq \mathcal{E}(y)$.
\begin{restatable}{defi}{embimp}\label{embimp}
An embedded implementation of a process $\mathtt{F} \in \mathbf{Proc}$ in a spacetime $\mathcal{M}$ is a pair $(i_\mathtt{F}, \mathcal{E})$ where $i_\mathtt{F} \in \mathcal{I}_\mathtt{F}$ and $\mathcal{E}: \mathcal{G}(i_\mathtt{F}) \rightarrow \mathcal{M}$ is an order-preserving map.
\end{restatable}
Embedded implementations necessarily exist, as one could always map everything to the same point in spacetime. This is of course not particularly interesting, so in practice, we will therefore not be concerned with the existence of an arbitrary embedding, but whether one exists such that the input and output systems are embedded at particular spacetime points. 

\begin{restatable}{defi}{clocal}\label{clocal}
A localisation of endpoints is a map $\mathcal{C}$ from the inputs and outputs of a process $\mathtt{F} \in \mathbf{Proc}$ to the spacetime $\mathcal{M}$. Given a specific localisation of endpoints $\mathcal{C}$, we  say that an embedding $(i_\mathtt{F}, \mathcal{E})$ is $\mathcal{C}-$local if $\mathcal{C}(A) = \mathcal{E}(A)$ for all inputs and outputs $A$\footnote{Note that this constraint is making use of our abuse of notation in that the $A$ on the LHS is an input/output of a process, whilst the $A$ on the RHS is the associated element in the framed partial order.}. 
\end{restatable}

Now we have all of the machinery set up to formalise the question we intuitively were asking about embeddability of our process $\mathtt{Bell}$. Given the process $\mathtt{Bell}$ in quantum theory, we want to know if it can be embedded in Minkowski spacetime in a spacelike separated way. That is, does there exist a $\mathcal{C}$-local embedding where $\mathcal{C}$ is such that 
\begin{equation}\label{bellconstraint}
\input{Diagrams/Bellconstraint.tikz}.
\end{equation}

We then pick an implementation from the implementation set of $\mathtt{Bell}$ \cref{impset}, apply $\mathcal{G}$ to it to turn it into an FPO and find a $\mathcal{C}$-local embedding $\mathcal{E}$,

\begin{equation}
\input{Diagrams/bell.tikz} \quad = \quad \input{Diagrams/belldecomp.tikz} \quad \substack{\mathcal{G} \\ \rightarrow} \quad \input{Diagrams/bellminrep.tikz} \quad \substack{\mathcal{E} \\ \rightarrow} \quad \input{Diagrams/BellinMink.tikz}.
\end{equation}
Note that it is crucial that we are in quantum theory, as one would expect. If we were instead asking if the process $\mathtt{Bell}$, as a process in classical theory has a $\mathcal{C}$-local embedding, the answer would be negative as the common cause implementation above is not an implementation of the \textit{classical} process $\mathtt{Bell}$. In order to $\mathcal{C}$-locally embed the classical $\mathtt{Bell}$, the localisation $\mathcal{C}$ must be such that at least one input is mapped to a spacetime point in the past of the spacetime localisations of both outputs, i.e., such that the communication loophole is not closed. The framework thus gives us a way to distinguish between classical and quantum theory (or generally, between two process theories), just as we desired.

\subsection{Revisiting the PR box}\label{sec:revisit}
We motivated our framework by pointing out that  ``unphysical'' scenarios such as perfect correlation without even a common past cannot be excluded by just appealing to no-signalling as a physical principle. We now show how this motivation is addressed in our framework and how, in particular, space-like extended PR boxes are impossible in quantum theory in the framework. In fact, this follows from a more general statement, namely that in all terminal theories a process can be embedded into spacelike separated regions iff it has a common cause decomposition.
\begin{restatable}{prop}{primpossible}\label{primpossible}
Let $\mathtt{F} \in \mathbf{Proc}$, with $\mathbf{Proc}$ terminal, be a process with two inputs $X, Y$ and two outputs $A, B$. There exists a $\mathcal{C}$-local embedded implementation of $\mathtt{F}$ into Minkowski spacetime with the localisation of endpoints $\mathcal{C}$
\begin{equation}\label{prconstraint}
\input{Diagrams/Bellconstraint.tikz}
\end{equation}
if and only if $\mathtt{F}$ has an implementation
\begin{equation}
\input{Diagrams/belldecomp.tikz} \in \mathcal{I}_\mathtt{F}
\end{equation}
for some $\mathtt{f}, \mathtt{g}, \mathtt{s} \in \mathbf{Proc}$.
\end{restatable}
\begin{proof} 
``$\Rightarrow$'': We show below that if there exists an implementation that can be $\mathcal{C}$-locally embedded into Minkowski spacetime,
\begin{equation}
\input{Diagrams/FinMink.tikz},
\end{equation}
then
\begin{equation}\label{PRbelldecomp}
\input{Diagrams/belldecomp.tikz}
\end{equation}
is an implementation of $\mathtt{F}$ for some processes $\mathtt{f}, \mathtt{g}, \mathtt{s}$. 

Suppose that we have a $\mathcal{C}$-local embedded implementation $(i_\mathtt{F}, \mathcal{E})$. Consider the following partition of the spacetime
\begin{equation}
\input{Diagrams/correspondingregions.tikz}.
\end{equation}
For each of the regions in this diagram, we will compose all boxes that get mapped into this region into a single box. That is we compose all boxes $\mathtt{x}$ in the implementation $i_\mathtt{F}$ which satisfy $\mathcal{C}(X) \leq \mathcal{E}(\mathtt{x}) \leq \mathcal{C}(A)$ into a single box $\mathtt{f}$, all boxes $\mathtt{y}$ that satisfy $\mathcal{C}(Y) \leq \mathcal{E}(\mathtt{y}) \leq \mathcal{C}(B)$ into a single box $\mathtt{g}$, all other boxes $\mathtt{z}$ that satisfy $\mathcal{E}(\mathtt{z}) \leq \mathcal{C}(A)$ or $\mathcal{E}(\mathtt{z}) \leq \mathcal{C}(B)$ into a single box $\mathtt{s}$ and finally all remaining boxes into a single box $\mathtt{a}$. 

Note that this implies that $\mathtt{a}$ is not connected via one of its outputs to the outputs $A$ or $B$. If it were, then there would exist a box $\mathtt{b}$ which got composed into $\mathtt{a}$ in the original diagram connected in this way to w.l.o.g. $A$. But then $\mathcal{E}(\mathtt{b}) \leq \mathcal{E}(A) = \mathcal{C}(A)$ due to $\mathcal{E}$ being order-preserving. But then, $\mathtt{b}$ would have been boxed into either $\mathtt{f}$ (if additionally, $\mathcal{C}(X) \leq \mathcal{E}(\mathtt{b})$) or $\mathtt{s}$ (otherwise). Hence, it could not have been boxed into $\mathtt{a}$. Therefore, $\mathtt{a}$ is an effect.

The process theory $\mathbf{Proc}$ is terminal, so $\mathtt{a}$ is the discard and can be decomposed as
\begin{equation}
\input{Diagrams/effectfactorize.tikz} = \input{Diagrams/effectfactorize2.tikz}.
\end{equation}
We can now absorb the discards into the boxes $\mathtt{f}$ and $\mathtt{g}$.  We end up with a diagram of the form of \cref{disjointprocess}. Since we only used composition of processes and \cref{causalconditions} , this diagram must still be equal to $\mathtt{F}$, and is thus an implementation of it.

``$\Leftarrow$'': Trivial.
\end{proof}
Note that whether the PR box (or any other process) can be embedded in a $\mathcal{C}$-local way according to \cref{prconstraint} depends on the process theory under consideration. For example, in boxworld PR boxes can be decomposed as in \cref{PRbelldecomp} and can thus be embedded in the desired way \cite{barrett2007information}. The assumption of terminality is also critical as otherwise we could not have absorbed the effect $\mathtt{a}$ in the above proof. As alluded to before, allowing non-terminal effects is akin to allowing post-selection and it is known that in such settings PR boxes are achievable \cite{marcovitch2007quantum}. 

If the two spacetime regions in which we embed are completely disjoint (i.e., there exists no common past or future), we find an analogous equivalence which even holds for non-terminal theories.

\begin{restatable}{prop}{sepspace}\label{sepspace}
Let $\mathtt{F} \in \mathbf{Proc}$ be a process with two inputs $X, Y$ and two outputs $A, B$ in a process theory $\mathbf{Proc}$. Let $\mathcal{M}$ be the disjoint union of two spacetimes $\mathcal{M}_p$ and $\mathcal{M}_q$ such that $\mathcal{C}(X) \leq \mathcal{C}(A) \in \mathcal{M}_p$ and $\mathcal{C}(Y) \leq \mathcal{C}(B) \in \mathcal{M}_q$. There exists a $\mathcal{C}$-local embedded implementation of $\mathtt{F}$ into $\mathcal{M}$ iff 
\begin{equation}\label{separableimp}
\input{Diagrams/separableimp.tikz} \in \mathcal{I}_\mathtt{F}
\end{equation}
for some $\mathtt{f}, \mathtt{g} \in \mathbf{Proc}$.
\end{restatable}
\begin{proof}
``$\Rightarrow$'': Let $(i_\mathtt{F}, \mathcal{E})$ be a $\mathcal{C}$-local embedded implementation. Compose all boxes $\mathtt{x}$ which satisfy $\mathcal{E}(\mathtt{x}) \leq \mathcal{C}(A)$ or $\mathcal{C}(X) \leq \mathcal{E}(\mathtt{x})$ into a single box $\mathtt{f}$. Then, compose all boxes $\mathtt{y}$ which satisfy $\mathcal{E}(\mathtt{y}) \leq \mathcal{C}(B)$ or $\mathcal{C}(Y) \leq \mathcal{E}(\mathtt{y})$ into a single box $\mathtt{g}$. Note that due to order-preservation of $\mathcal{E}$, the disjointness of $\mathcal{M}$ and our choice for $\mathcal{C}$ all boxes fall into exactly one of the these two cases. The resulting implementation is of the form of \cref{separableimp}. Any wires from $\mathtt{f}$ to $\mathtt{g}$ or vice versa would contradict order-preservation of $\mathcal{E}$ and disjointness of $\mathcal{M}$.

``$\Leftarrow$'': Trivial.
\end{proof}
From this it follows that the PR box has no such embedded implementation, as long as the process theory has a full classical subtheory \cite{gogioso2017categorical,selby2021reconstructing}. 
\begin{restatable}{coro}{primpossible2}\label{primpossible2}
There exists no embedded implementation $(i_\mathtt{PR}, \mathcal{E})$ of the PR box into a spacetime $\mathcal{M}$, which is the disjoint union of two spacetimes $\mathcal{M}_p$ and $\mathcal{M}_q$, and a localisation of endpoints such that $\mathcal{C}(X) \leq \mathcal{C}(A) \in \mathcal{M}_p$ and $\mathcal{C}(Y) \leq \mathcal{C}(B) \in \mathcal{M}_q$ in any process theory with a full classical subtheory.
\end{restatable}
\begin{proof}
Assume the statement is false. Then, by \cref{sepspace}, the PR box has an implementation $i_\mathtt{PR}$ of the form of \cref{separableimp}. However, as the PR box and thus $\mathtt{f}$ and  $\mathtt{g}$ have only classical inputs and outputs and the classical subtheory is, by assumption, full, $\mathtt{f}$ and  $\mathtt{g}$ must be stochastic maps. Thus, $i_\mathtt{PR}$ decomposes as a product of stochastic maps, but then $i_\mathtt{PR}$ cannot possibly be an implementation of the PR box since the PR box cannot be written as such a product. This is a contradiction, so the statement must be true.
\end{proof}

\subsection{Embeddability implies no-signalling}\label{sec:EmbedAndSignalling}
So far, we have argued that consistency with no-signalling is not sufficient to decide whether a process can be realised in spacetime. However, at least in causal theories, it is, as one might expect, necessary.
\begin{restatable}{prop}{nosignec}
Let $\textbf{Proc}$ be a terminal theory, $\mathtt{F} \in \textbf{Proc}$ such that it has a $\mathcal{C}$-local embedding. Let $S_I$ be a subset of inputs and $S_O$ be a subset of outputs of $\mathtt{F}$, then if $\mathcal{C}(A) \not \leq \mathcal{C}(B)$ for all $A \in S_I, B \in S_O$, $S_I$ does not signal to $S_O$ (that is the composite system $\bigotimes_{A \in S_I} A$ does not signal to the composite system $\bigotimes_{B \in S_O} B$).
\end{restatable}
\begin{proof}
Let $\mathtt{F}$ be a process and $(i_\mathtt{F}, \mathcal{E})$ a $\mathcal{C}$-local embedding. Let $A \in S_I, B \in S_O$. The assumption that $\mathcal{C}(A) \not \leq \mathcal{C}(B)$ together with the fact that $\mathcal{E}$ is an order-preserving map implies that $A \not \leq B$ as elements of $\mathcal{G}(i_\mathtt{F})$. Going back to the diagram, by the definition of $\mathcal{G}$, there are therefore no future-directed paths from system $A$ to the system $B$ (even indirectly via other processes) and since this holds for all $A, B$, there are not even wires going from the subset of systems $S_I$ to the subset of systems $S_O$. This then implies that by rewriting $i_\mathtt{F}$, we obtain a new implementation of $\mathtt{F}$ 
\begin{equation}\label{eq:nosignec}
\input{Diagrams/nosignec.tikz},
\end{equation}
up to a reordering of the inputs and outputs. To see this, consider the fact that since $S_I$ and $S_O$ are not connected, any box in the implementation $i_\mathtt{F}$ can only be connected to $S_I$ or $S_O$ (or possibly neither). We can now compose all boxes connected to $S_I$ into a single process $\mathtt{i}$ and all other processes into another process $\mathtt{o}$. Then $\mathtt{i}$ and $\mathtt{o}$ can only be connected via an output from $\mathtt{o}$ to an input to $\mathtt{i}$ (otherwise $\mathtt{i}$ would be connected to both $S_I$ and $S_O$) and as such we obtain the form in \cref{eq:nosignec}. Such a process does not signal from $S_I$ to $S_O$, i.e., it satisfies \cref{sigcond}, which can be straightforwardly shown by discarding the outputs $O \backslash S_O$,
\begin{equation}
\input{Diagrams/nosignec1.tikz} \quad = \quad \input{Diagrams/nosignec2.tikz}
\end{equation}
\end{proof}
This can be seen as a further formalisation of the results of \cite{kissinger2017equivalence} on the relationship between process terminality and the lightcone structure of spacetime. This result also implies that any process in a terminal process theory that is $\mathcal{C}$-locally embeddable in our framework, is also embeddable with the same localisation of the inputs and outputs in signalling-based frameworks like Ref.~\cite{vilasini2022embedding}.

\subsection{Undecidability of embeddings}
 
 The general problem of deciding whether or not a given process $\mathtt{F}\in \mathbf{Proc}$ is $\mathcal{C}$-locally embeddable turns out to be computationally undecidable. To see this, we can consider the case where $\mathbf{Proc}$ is quantum theory,  $\mathtt{F}$ is a bipartite nonsignalling classical channel, and we consider $\mathcal{C}$-localisation into Minkowski spacetime as described in \cref{bellconstraint}. \Cref{primpossible} above tells us that such an embedding exists iff there is a decomposition of $\mathtt{F}$ as 
\beq
\input{Diagrams/belldecomp.tikz}.
\eeq 
In other words, iff this classical correlation admits of a quantum realisation in a Bell scenario. This is a problem which is known to be undecidable \cite{fu2021membership}. 
 
That the general problem is undecidable, however, does not mean that all particular instances are. Indeed, if we consider the same problem as above but where $\mathbf{Proc}$ is taken to be classical theory then testing this is a simple linear program \cite{kaszlikowski2000violations,zukowski1999strengthening}. It is therefore worthwhile to try to simplify the problem as much as possible, which is exactly what we do in the following sections.

\section{What is needed to decide embeddability?}\label{sec:ReducedDescription}
When we ask whether a process $\mathtt{F}$ is embeddable with $\mathcal{C}$-local inputs and outputs, we can think of this as the question of whether any of its implementations $i_\mathtt{F} \in \mathcal{I}_\mathtt{F}$ yield a valid $\mathcal{C}$-local embedding. It is not necessarily always the same implementation that gives us a valid $\mathcal{C}$-local embedding. For example, consider our implementations of the Bell process from \cref{impset}
\begin{equation}\label{simpleimps}
s_\mathtt{Bell} \quad = \quad \input{Diagrams/sequential1.tikz}, \quad s'_\mathtt{Bell} \quad = \quad \input{Diagrams/sequential2.tikz}, \quad p_\mathtt{Bell} \quad = \quad \input{Diagrams/belldecomp.tikz}.
\end{equation}
The implementation $s_\mathtt{Bell}$, requires Alice to communicate to Bob while $s'_\mathtt{Bell}$ requires the opposite. Thus, in order to find a valid embedding we would need Alice in the past of Bob (or formally, $\mathcal{C}(X) < \mathcal{C}(B)$) in the case of $s_\mathtt{Bell}$ and Bob in the past of Alice ($\mathcal{C}(Y) < \mathcal{C}(A)$) in the case of $s'_\mathtt{Bell}$. Thus, we cannot ignore either implementation as sometimes one condition may be fulfilled and sometimes the other one. When deciding embeddability we would thus always have to check both implementations. On the other hand, when it comes to the implementation $p_\mathtt{Bell}$, Alice and Bob can be spacelike separated. The constraints for this implementation to yield a valid embedding thus seem to be strictly weaker. Indeed, it suffices to check whether $p_\mathtt{Bell}$ can be embedded in a $\mathcal{C}$-local way. Hence, if $p_\mathtt{Bell}$ is an implementation, $s_\mathtt{Bell}$ and $s'_\mathtt{Bell}$ can safely be ignored.

The question is then, given some set of implementations $\mathcal{I}_\mathtt{F}$ what is the minimal subset that still always allows us to decide whether $\mathtt{F}$ can be $\mathcal{C}$-locally embedded in spacetime or not? Is there some condition to find it? Since embeddability depends on the (framed) partial order associated to an implementation, it makes sense to look at framed partial orders directly to answer this question. 

\begin{restatable}{defi}{def:MINO}\label{def:MINO}
For all $M, N \in \mathbb{N}$, we define the set of framed partial orders of class \MINO  
\beq\mathcal{S}_{\MINO} := \{(S, \mathcal{I}, \mathcal{O}) \text{ framed partial order}||\mathcal{I}| = M, |\mathcal{O}| = N\}.
\eeq
\end{restatable}
Additionally, we will also need an abstract equivalent of spacetime embeddings.

\begin{restatable}{defi}{FOP}\label{FOP}
Let $(S_1, \mathcal{I}_1, \mathcal{O}_1), (S_2, \mathcal{I}_2, \mathcal{O}_2) \in \mathcal{S}_{\MINO}$ be framed partial orders. We call $\mathcal{E}: S_1 \rightarrow S_2$ a frame- and order preserving (FOP) map if it is order-preserving, i.e., $x \leq y \Rightarrow \mathcal{E}(x) \leq \mathcal{E}(y)$, and maps the $n$-th element of $\mathcal{I}_1$ ($\mathcal{O}_1$) to the $n$-th element of $\mathcal{I}_2$ ($\mathcal{O}_2$) for $1 \leq n \leq M$ ($1 \leq n \leq N$).
\end{restatable}
The notion of a FOP map, then immediately gives us a very useful preorder.

\begin{restatable}{defi}{def:succ}\label{def:succ}
We say that $S_1$ is embeddable in $S_2$ or that $S_1$ is more useful than $S_2$, writing $S_1 \succ S_2$, if there exists a FOP map $\mathcal{E}: S_1 \rightarrow S_2$.
\end{restatable} 
This is indeed a preorder as compositions of FOP maps are again FOP, ensuring transitivity, and the identity is a valid FOP map, ensuring reflexivity.

We illustrate these two definitions with an example:
\begin{equation}
\input{Diagrams/fopexample.tikz}.
\end{equation}
Here, the arrows indicate the FOP map. The red arrows indicate the mapping of frame elements which are fixed by the requirement that the map is frame-preserving (note that we also cannot map $A$ to $\delta$ as the ordering of the frame elements on the page matters). The black arrow indicates the mapping of the internal element $\mathtt{s}$ which in this case must be mapped to $\alpha$ to preserve the order (in the general case, we may have more freedom as we will see later on).

Note that if we have two implementations of the same process, $i_\mathtt{F}, j_\mathtt{F}\in\mathcal{I}_\mathtt{F}$, then their associated FPO types, $\mathcal{G}(i_\mathtt{F}), \mathcal{G}(j_\mathtt{F})$, belong to the same \MINO class. That is for any $\mathtt{F}$ there exists $M$ and $N$ such that $\mathcal{G}(\mathcal{I}_\mathtt{F})\subseteq \mathcal{S}_{\MINO}$.
Interestingly, this puts implementations and spacetimes on the same level, in a sense. We can embed implementations into spacetime but we can also embed one implementation into another (more accurately into the associated FPO). We can then treat the spacetime as an FPO by considering the image of the constraint $\mathcal{C}$ as the frame. This means that our two notions of embeddability \cref{clocal,def:succ} are equivalent in some sense.

The formalisation of this idea leads to a hierarchy where FPOs are more useful for deciding embeddability the higher up they are in terms of $\succ$.
\begin{restatable}{thm}{embedimplication}\label{embedimplication}
Given two implementations $i_\mathtt{F}, j_\mathtt{F}\in\mathcal{I}_\mathtt{F}$, $\mathcal{G}(i_\mathtt{F}) \succ \mathcal{G}(j_\mathtt{F})$ iff $i_\mathtt{F}$ is $\mathcal{C}$-locally embeddable whenever $j_\mathtt{F}$ is $\mathcal{C}$-locally embeddable.
\end{restatable}
\begin{proof}
``$\Rightarrow$'': By assumption, there exists a FOP map $\mathcal{E}_1: \mathcal{G}(i_\mathtt{F}) \rightarrow \mathcal{G}(j_\mathtt{F})$. Let now $\mathcal{E}_2: \mathcal{G}(j_\mathtt{F}) \rightarrow \mathcal{M}$ be a $\mathcal{C}$-local embedding. As compositions of order-preserving maps are again order-preserving maps, we have that $\mathcal{E} := \mathcal{E}_2 \circ \mathcal{E}_1$ is an order-preserving map from $\mathcal{G}(i_\mathtt{F})$ to $\mathcal{M}$. Since $\mathcal{E}_1$ preserves the frame and $\mathcal{E}_2$ is $\mathcal{C}$-local, we have for any input or output $A$ 
\begin{equation}
\mathcal{C}(A) = \mathcal{E}_2(A) = \mathcal{E}_2 \circ \mathcal{E}_1(A) = \mathcal{E}(A)
\end{equation}
and hence $\mathcal{E}$ is also $\mathcal{C}$-local.

``$\Leftarrow$'': The implementation $j_\mathtt{F}$ is trivially embeddable into the ``spacetime'' $\mathcal{M}:=\mathcal{G}(j_\mathtt{F})$ with the localisation $\mathcal{C}(A) = A$ for all inputs and outputs $A$. By assumption, we can then also $\mathcal{C}$-locally embed $i_\mathtt{F}$ into $\mathcal{G}(j_\mathtt{F})$, i.e., there exists an order-preserving map $\mathcal{E}: \mathcal{G}(i_\mathtt{F}) \rightarrow \mathcal{G}(j_\mathtt{F})$ which due to $\mathcal{C}$-locality satisfies $\mathcal{E}(A) = A$ for all inputs and outputs $A$. Thus, $\mathcal{E}$ is FOP and by definition, we have $\mathcal{G}(i_\mathtt{F}) \succ \mathcal{G}(j_\mathtt{F})$.
\end{proof}
Note that we can now view \cref{primpossible,sepspace} as corollaries of the above theorem. Additionally, the preorder $\succ$ gives us a lot of information about the contents of the implementation set. 
If $S_1 \succ S_2$ and $S_1$ corresponds to an implementation of $\mathtt{F}$, then so does $S_2$ up to relabelling of the FPO. We formalise the notion of relabelling in the definition below.

\begin{restatable}{defi}{relabel}\label{relabel}
A frame-preserving order embedding (FOE) $\mathcal{E}: S_1 \rightarrow S_2$ is a FOP map that is also order-reflecting, i.e., $x \leq y \Leftarrow \mathcal{E}(x) \leq \mathcal{E}(y)$. A relabelling or frame-preserving order isomorphism is a FOE map which is surjective. 
\end{restatable}
Note that an order-reflecting map is necessarily injective. An equivalent definition for a relabelling is to say that it is an isomorphism which is FOP  with a FOP inverse. An example of an FOE is
\begin{equation}
\input{Diagrams/foeexample.tikz}
\end{equation}

The set of FPOs associated to implementations $\mathcal{G}(\mathcal{I}_{\mathtt{F}})$ would thus form a downwards closed set under $\succ$ if it were not for the ambiguity of the labels. The labels of the FPO are ultimately arbitrary and irrelevant. Hence, we could instead think about ``unlabelled'' FPOs which we can formalise as the equivalence classes under relabelling.

\begin{restatable}{defi}{FPOtype}\label{FPOtype}
Two FPOs $(S_1, \mathcal{I}_1, \mathcal{O}_1), (S_2, \mathcal{I}_2, \mathcal{O}_2) \in \mathcal{S}_{\MINO}$ belong to the same framed partial order type, writing $S_1 \cong S_2$, if there exists a relabelling $\mathcal{E}: S_1 \rightarrow S_2$. 
We define the set of framed partial order types of class \MINO, $\mathcal{T}_{\MINO}$ as the quotient set of $\mathcal{S}_{\MINO}$ under the equivalence relation of relabelling, i.e., $\mathcal{T}_{\MINO} := \mathcal{S}_{\MINO}/\cong$. For $S \in T \in \mathcal{T}_{\MINO}$, we call $S$ a labelling of $T$.
\end{restatable}
Graphically, we can depict this as follows
\begin{equation}
\input{Diagrams/bellminrep2.tikz} = \left\{\input{Diagrams/fpotypeexample.tikz},\ \ \input{Diagrams/fpotypeexample2.tikz},\ \ \input{Diagrams/fpotypeexample3.tikz}, \cdots \right\}.
\end{equation}

Since the difference between different FPOs belonging to the same FPO type is essentially just the labels of elements, many properties define invariants. That is if for some FPO type $T$, all $S \in T$ have some property, it is well-defined to say that the FPO type itself has this property. For example, all framed partial orders of the same FPO type have the same order, i.e., number of elements, due to relabellings being isomorphisms, meaning it is also well-defined to speak about the order of an FPO type. Importantly, it is also well-defined to say an FPO type $T_1$ is embeddable into an FPO type $T_2$, $T_1 \succ T_2$. This simply means that for any or equivalently for all $S_i \in T_i$, it holds that $S_1 \succ S_2$. 

We will also always refer to an FPO $S \in T$ as a labelling, never as an element, to avoid confusion with the elements of an FPO. We will still sometimes speak of elements of the FPO type, by which we mean the elements of the labellings (e.g., ``all internal elements of the FPO type $T$ are minimal elements of the FPO type $T$'' should be read as ``for all $S \in T$, the internal elements of $S$ are minimal elements of $S$).

The set of FPO types associated to implementations of a process $\mathcal{G}(\mathcal{I}_\mathtt{F}) / \cong$ is then actually a downwards closed set under $\succ$. In fact, we can also prove a converse of this statement, yielding an equivalence between $T_1 \succ T_2$ and the possibility of converting diagrams associated with $T_1$ into a diagram associated with $T_2$. This is formalised and shown in the theorem below.

\begin{restatable}{thm}{conversion}\label{conversion}
Let $T_1, T_2 \in \mathcal{T}_{MN}$. The following two statements are equivalent.
\begin{enumerate}
\item $T_1 \succ T_2$
\item For all processes $\mathtt{F}$ in all process theories, if $\mathtt{F}$ has an implementation $i_\mathtt{F}$ such that $\mathcal{G}(i_\mathtt{F}) \in T_1$, then $\mathtt{F}$ also has an implementation $j_\mathtt{F}$ such that $\mathcal{G}(j_\mathtt{F}) \in T_2$
\end{enumerate}
\end{restatable}
The proof of this statement can be found in \cref{conversionappendix}.

Finally, let us briefly remark on the term ``usefulness'' as a synonym for ``embeddability'' in \cref{def:succ} in the light of the two main results in this section \cref{embedimplication,conversion}. We will extend this terminology to implementations as well, i.e., we will say $i_\mathtt{F}$ is more useful than $j_\mathtt{F}$ if $\mathcal{G}(i_\mathtt{F}) \succ \mathcal{G}(i_\mathtt{F})$. Given a process $\mathtt{F}$, implementations of $\mathtt{F}$ which are more useful then deserve this label because a) they are better at helping us decide whether the process is $\mathcal{C}$-local embeddable into a given spacetime due to \cref{embedimplication} and b) they tell us more about the decompositional structure of the process due to \cref{conversion}.

\subsection{Equivalence classes of implementations}\label{sec:minreps}
In the previous section, we have defined a preorder on framed partial orders and shown that implementations higher up in this preorder are more useful for deciding embeddability. This already quite significantly reduces the information necessary from $\mathcal{I}_{\mathtt{F}}$. We can go a step further by considering equivalence classes under $\succ$.
\begin{restatable}{defi}{eqclasses}\label{eqclasses}
We say that two FPOs $S_1, S_2$ are equivalent $S_1 \sim S_2$, if $S_1 \succ S_2$ and $S_2 \succ S_1$. We say that two FPO types $T_1, T_2$ are equivalent, $T_1 \sim T_2$, if $T_1 \succ T_2$ and $T_2 \succ T_1$. For two equivalence classes of FPO types $\mathcal{T}_1, \mathcal{T}_2$ we say that $\mathcal{T}_1 \succ \mathcal{T}_2$ if for any $T_1 \in \mathcal{T}_1, T_2 \in \mathcal{T}_2$ it holds that $T_1 \succ T_2$. On the set of equivalence classes, $\succ$ is a partial order.
\end{restatable}
Due to \cref{embedimplication}, it suffices to check a single element of an equivalence class when deciding embeddability. Additionally, due to \cref{conversion} if an FPO type corresponds to an implementation, then all FPO types in the same equivalence class correspond to an implementation too. Thus, it suffices to talk about equivalence classes instead of individual FPO types.

Note that this actually simplifies our problem considerably as every equivalence class has an infinite number of FPO types. To see this, let $S_0$ be an arbitrary FPO, e.g.
\begin{equation}
\input{Diagrams/bellminrep.tikz}.
\end{equation}
Now define $S_{n+1}$ recursively by adding a disjoint internal element $x_{n+1}$ to $S_n$
\begin{equation}
\input{Diagrams/Sn.tikz} \quad \rightarrow \quad \input{Diagrams/Sn1.tikz}
\end{equation}
These FPOs are all equivalent ($S_0 \succ S_n$ because $S_0 \subseteq S_n$ and $S_n \succ S_0$ because the elements $x_n$ can be mapped to arbitrary elements in $S_0$ without violating order relations). Thus, by \cref{conversion} if $S_0$ corresponds (up to relabelling) to an implementation of $\mathtt{F}$, then so do all $S_n$. 

To take this even further, we would also like canonical representatives with nice properties to deal with these equivalence classes in an elegant manner. Going back to our previous example of the FPO $S_n$, if these made up the entire equivalence class, the obvious choice would be to pick $S_0$ (or its FPO type) as a representative since it is the smallest FPO and also a strict subset of all the other FPOs. Or in other words, all the other FPOs are simply $S_0$ with redundant information (the elements $x_n$).

It turns out that this idea can be generalized.
\begin{restatable}{thm}{minrep}\label{minrep}
Let $\mathcal{T}$ be an equivalence class of FPO types under $\sim$. Then, there exists a unique FPO type $R \in \mathcal{T}$ which we call the minimal representative such that for all $S \in R$, all $T \in \mathcal{T}$ and all $S' \in T$, any FOP map $\mathcal{E}: S \rightarrow S'$ is an order-embedding. In particular, this means that $|R| \leq |T|$, where $|T|$ is the order, i.e., the number of elements, of $T$,  with equality iff $R = T$.
\end{restatable}
\begin{proof}
Consider a minimal representative $R\in \mathcal{T}$, that is, we assume that $|R|$ is minimal in $\mathcal{T}$. We will shortly show that for all $S \in R$ and any other $S'\in T \in \mathcal{T}$ any FOP map $\mathcal{E}:S\rightarrow S'$ is an order-embedding. It then immediately follows from the notion of an order-embedding that $|S|\leq |S'|$ and that $|S|=|S'|$ iff $S$ and $S'$ are order isomorphic, which implies $R = T$. Hence, minimal representatives are unique.

Now let us show that any FOP map $\mathcal{E}:S\to S'$ for a labelling of a minimal representative $S$ and arbitrary $S'\in T \in \mathcal{T}$ is an order-embedding. To begin, note that as $S$ and $S'$ are both from the same equivalence class under $\sim$, there necessarily exists such a FOP map $\mathcal{E}:S\rightarrow S'$ and, also, there exists a FOP map $\mathcal{E}':S'\rightarrow S$.

Let us, moreover, define $S_1:=\text{Im}(\mathcal{E})$ and $S_2:=\text{Im}(\mathcal{E}'\circ\mathcal{E})$. As these are defined as the images of FOP maps, it immediately follows from the definition of $\succ$ that $S\succ S_1 \succ S_2$. Moreover, as $S_2\subseteq S$ and $S_2$ contains the frame of $S$, we also have that $S_2\succ S$. Putting all this together we have that $S\succ S_1\succ S_2\succ S$, hence, we have that the $S_i$ are also in the same equivalence class under $\sim$.

Next, as the cardinality of images can be at most the cardinality of the domain, we have that $|S_i| \leq |S|$ for $i=1,2$. But by assumption, the order of $S$ is already minimal in the equivalence class and so $|S_i| = |S|$ for $i=1,2$. Recalling that $S_1=\text{Im}(\mathcal{E})$, the fact that $|S_1|=|S|$ immediately means that $\mathcal{E}:S\rightarrow S'$ is injective, and restricting the codomain of $\mathcal{E}:S'\to S$ to $S_1$, which we denote as $\mathcal{E}|^{S_1}:S\to S_1$, we find that this is bijective.

Next, we restrict the domain of $\mathcal{E}':S'\rightarrow S$ to $S_1$, which we denote as $\mathcal{E}'|_{S_1}:S_1\rightarrow S$. Noting that $\mathcal{E}'\circ \mathcal{E} = \mathcal{E}'|_{S_1}\circ \mathcal{E}|^{S_1}$, and recalling that  $S_2=\text{Im}(\mathcal{E}'\circ \mathcal{E})$ and $|S_2|=|S|$, we therefore have that $\mathcal{E}'|_{S_1}\circ \mathcal{E}|^{S_1}$ is bijective. Given that we already know that $\mathcal{E}|^{S_1}$ is a bijection then this immediately tells us that $\mathcal{E}'|_{S_1}$ is too. Additionally, we also have $S = S_2$ since $S_2 \subseteq S$.

Further, since $S$ is finite, $\mathcal{E}' \circ \mathcal{E}$ can thus be seen as a permutation of a finite number of elements. Hence, there exists $n$ such that $(\mathcal{E}' \circ \mathcal{E})^n = \text{id}$. From this it follows that $(\mathcal{E}' \circ \mathcal{E})^{n-1} \circ \mathcal{E}'|_{S_1}$ is an order-preserving inverse of $\mathcal{E}|^{S_1}: S \rightarrow S_1$. A bijective order-preserving map with an order-preserving inverse is an order isomorphism. Thus, $S \cong S_1 \subseteq S'$, which means that $\mathcal{E}$ is an order-embedding from $S$ to $S'$. Finally, if $|R| = |T|$, then $|S| = |S'|$ by definition, and $\mathcal{E}$ is a surjective FOE, i.e., a relabelling. Hence, $S \cong S'$ and thus $R = T$.
\end{proof} 

In many ways, minimal representatives are the ``correct'' way to think about implementations as a whole. For example, one might be interested in the upper bound for the number of gates (i.e., elements in the partial order) needed to embed some process. But if one were to look at all partial orders, this question is not particularly well-defined since, as we mentioned earlier, we can always add additional elements to any partial order without changing anything about its embeddability. Of course, since these additional elements correspond to trivial processes, the intuitive thing is to disregard them or alternatively, strip away all the redundant elements, but neither of these is a very rigorous notion a priori. Formalisation of these ideas leads one directly to minimal representatives. Minimal representatives also give us an appealing interpretation of equivalence classes as all the ways one can add redundant information to the minimal representatives. 

The information we need to know about a process for our purposes is then simply the set of associated minimal representatives.

\begin{restatable}{defi}{imprepset}\label{imprepset}
We denote the set of minimal representatives of class \MINO with $\mathcal{R}_{\MINO} \subseteq \mathcal{T}_{\MINO}$. The set of the minimal representatives of the implementations of a process $\mathtt{F}$ (we will usually simply say the minimal representatives of $\mathtt{F}$ for brevity) is the set
\begin{equation}
\mathcal{R}_\mathtt{F} = \{R \in \mathcal{R}_{\MINO} | \exists i_\mathtt{F} \in \mathcal{I}_\mathtt{F}: \mathcal{G}(i_\mathtt{F}) \in R\}.
\end{equation}
\end{restatable}
In analogy to the above nomenclature, for some $R \in \mathcal{R}_\mathtt{F}$, we will say that $R$ is a minimal representative of $\mathtt{F}$.

\Cref{minrep} also immediately gives us a useful criterion to check if an arbitrary partial order is a minimal representative.
\begin{restatable}{coro}{mincrit}\label{mincrit}
Let $T \in \mathcal{T}_{\MINO}$. Then $T$ is a minimal representative iff for all $S, S' \in T$ all FOP maps $\mathcal{E}: S \rightarrow S'$ are relabellings\footnote{We note here that if $S$ and $S'$ belong to the same FPO type, this means by definition that there exists a relabelling $\mathcal{E}: S \rightarrow S'$. However, this does not imply a priori that all FOP maps from $S$ to $S'$ are relabellings. For example, think of the FPO type corresponding to an input and a single internal element connected to it. Then, mapping both the input and the internal element to the input defines a valid FOP map but is not a relabelling as it is not surjective}.
\end{restatable}
\begin{proof}
``$\Rightarrow$'': Since $T$ is a minimal representative and $S \sim S'$, by \cref{minrep}, $\mathcal{E}$ is FOE. Further, $|S| = |S'|$ as they have the same order type, which implies that $\mathcal{E}$ is bijective and thus a relabelling.

``$\Leftarrow$'': By \cref{minrep}, there exists $R \sim T$ a minimal representative. Then, again by \cref{minrep}, for any $\tilde{S} \in R$, there exists a FOP map $\mathcal{F}: S \rightarrow \tilde{S}$ and an FOE $\tilde{\mathcal{F}}: \tilde{S} \rightarrow S$. By assumption, we have that $\mathcal{E} := \tilde{\mathcal{F}}\circ \mathcal{F}$ is a relabelling as it is a FOP self-map on $S$. This implies that $\mathcal{F}$ is an FOE as well. To see this, let $x,y\in S$ such that $\mathcal{F}(x) \leq \mathcal{F}(y)$. Then, since $\tilde{\mathcal{F}}$ is order-preserving we find that $\tilde{\mathcal{F}} \circ \mathcal{F}(x) \leq \tilde{\mathcal{F}} \circ \mathcal{F}(y)$ and since $\tilde{\mathcal{F}} \circ \mathcal{F}$ is a relabelling, and thus, in particular order-reflecting, we find $x \leq y$. Thus, we have $\mathcal{F}(x) \leq \mathcal{F}(y) \Rightarrow x \leq y$ and $\mathcal{F}$ is FOE. This then implies $|S| \leq |\tilde{S}|$ and thus $|T| \leq |R|$. But since $R$ is a minimal representative, by \cref{minrep}, this is only possible if $|T| = |R|$ and thus $T = R$, i.e., $T$ is a minimal representative. 
\end{proof}

In \cref{sec:addprop}, we show that a minimal representative is minimal within its equivalence class in some other senses as well. Minimal representatives have strictly the fewest order relations between elements and their Hasse diagrams have the smallest size (i.e., number of edges). Their height and width (i.e., length of the longest chain and antichain respectively) are minimal (although, not strictly). 

In the aforementioned appendix, we also analyse some features which are useful for finding minimal representatives. We show that there always exists a map from an FPO to its minimal representative which intuitively behaves like a projection when one identifies elements in the minimal representative with elements in the FPO via an order emebedding. We also analyse how one can construct more complex minimal representatives by essentially ``composing'' simpler ones. In particular, we show that any connected component of a minimal representative is itself a minimal representative. Hence, when in \cref{sec:simplemin} we explicitly list minimal representatives it will suffice to list only those with a single connected component (i.e., they are fully connected).

\section{Theory-specific considerations}\label{sec:thspec}

Up to this point, we have managed to reduce the set of FPO types we need to consider as far as is possible in a theory-independent setting. Essentially, instead of all FPO types, it is only the ones which are minimal representatives that are actually relevant. That we cannot drop any more FPO types in the theory-independent setting follows from \cref{conversion} and the fact that $\succ$ is, by definition, a partial order, instead of a preorder, on minimal representatives. However, we can go further if we specify the set of theories that we are actually interested in. If we do this, it can turn out that even some minimal representatives are not actually relevant. We formalise this notion now. 

\begin{restatable}{defi}{irrelevant}\label{irrelevant}
Let $\mathcal{P}$ be a collection of process theories. Let $R, R' \in \mathcal{R}_{\MINO}$ be two minimal representatives. We say that $R$ supersedes $R'$ for $\mathcal{P}$, $R \succ_\mathcal{P} R'$ if $R \succ R'$ and for any $\mathtt{F} \in \mathbf{Proc} \in \mathcal{P}$ it holds that $R' \in \mathcal{I}_F \implies R\in \mathcal{I}_F$. We call a minimal representative $R$ 
\begin{itemize}
\item \textbf{relevant} for $\mathcal{P}$ if it is a maximal element in $\mathcal{R}_{\MINO}$ under $\succ_\mathcal{P}$, 
\item \textbf{irrelevant} for $\mathcal{P}$ if there exists a maximal element in $\mathcal{R}_{\MINO}$, $R' \neq R$, under $\succ_\mathcal{P}$ and $R' \succ_\mathcal{P} R$ and 
\item \textbf{quasi-relevant} for $\mathcal{P}$ if $R$ is neither relevant nor irrelevant.
\end{itemize}
\end{restatable}

We will often drop the specification ``for $\mathcal{P}$'' if it is clear which process theories we are considering.

The intuition here is the following: if $R' \succ R$, then we only need $R$ when discussing a process which does not have $R'$ as an implementation. If $R$ is now irrelevant (and $R'$ is the relevant minimal representative from the definition), then such processes do not exist. Hence, we can essentially ignore $R$ entirely. 

Quasi-relevant minimal representatives are an interesting edge case in this regard. We could have also defined them equivalently as follows: We call $R$ quasi-relevant for $\mathcal{P}$ if the following holds: $R$ is not a maximal element in $\mathcal{R}_{\MINO}$ under $\succ_\mathcal{P}$, but there also exists no maximal element $R'$ in $\mathcal{R}_{\MINO}$ under $\succ_\mathcal{P}$ such that $R' \succ_\mathcal{P} R$. This equivalent framing also makes it immediately clear that any minimal representative which supersedes $R$ is itself quasi-relevant. In fact, the existence of a quasi-relevant minimal representative $R$ immediately implies the existence of an infinite sequence of distinct quasi-relevant minimal representatives $... \succ_\mathcal{P} R_{n+1} \succ_\mathcal{P} R_n \succ_\mathcal{P} ... \succ_\mathcal{P} R_0 = R$. A simple example of this exists in theories with cups and caps\footnote{When working with a theory with cups and caps, one can essentially forget about the distinction between inputs and outputs as an input can always be turned into an output by attaching a cup or vice versa by attaching a cap. The Choi-Jamiołkowski representation in quantum theory can be seen as a particular example of this. In many applications, one can equivalently use a channel or its Choi-Jamiołkowski representation. We note, however, that in the formalisation of quantum theory as a process theory used in the present work, the cup is not a valid state and the cap is not a valid effect due to terminality.}, i.e., theories that have a bipartite state called a cup and a bipartite effect called a cap,
\begin{equation}
\input{Diagrams/cup.tikz} \quad \input{Diagrams/cap.tikz}
\end{equation}
satisfying the snake equation
\begin{equation}
\input{Diagrams/snake.tikz} \quad = \quad \input{Diagrams/tensorsystemstriv2.tikz}.
\end{equation}
Since we can always add an identity to a process we have for all processes $\mathtt{F}$ in a theory with cups and caps that
\begin{equation}\label{cupcap}
\input{Diagrams/trivcomp2.tikz} \quad = \quad \input{Diagrams/fsnake.tikz}.
\end{equation}
The RHS is easier to embed than the LHS (intuitively, as we no longer need that the input is embedded into the past of the output). It is also easy to verify that the corresponding FPO types are inequivalent minimal representatives
\begin{equation}
\input{Diagrams/11minrepsnake.tikz} \quad \succ \quad \input{Diagrams/11minrep.tikz}.
\end{equation}
Since \cref{cupcap} holds for all processes, we can then also replace $\succ$ with $\succ_\mathcal{P}$ where $\mathcal{P}$ is now the collection of process theories with cups and caps. Hence, it appears we do not need the trivial minimal representatives in such theories as every process has a more useful minimal representative. However, we can play the same game as above an arbitrary number of times, i.e., we can add an arbitrary number of identities to $\mathtt{F}$ and replace them with snakes 
\begin{equation}\label{cupcap}
\input{Diagrams/trivcomp2.tikz} \quad = \quad \input{Diagrams/fnsnakes.tikz}
\end{equation}
where the ellipsis stand for an arbitrary number of snakes.

Hence, we find that in a theory with cups and caps that 
\begin{equation}\label{cupcap}
... \quad \succ_{\mathcal{P}} \quad \input{Diagrams/11minrep3snakes.tikz} \quad \succ_{\mathcal{P}} \quad \input{Diagrams/11minrep2snakes.tikz}\quad \succ_{\mathcal{P}} \quad \input{Diagrams/11minrepsnake.tikz}
\end{equation}
and there exists no minimal representative which is greater under $\succ_{\mathcal{P}}$ than all of these (to see this, consider that the minimal representatives of the identity in a theory that consists of only the identity, the cup and the cap are exactly the above).

Physically, what all this means is that in a theory with cups and caps we can $\mathcal{C}$-locally embed all processes for any $\mathcal{C}$ that does not map inputs and outputs into completely disconnected parts of the spacetime. 

We will see another (possible) example of quasirelevant minimal representatives in \cref{sec:quantum}.

Note further that the condition $R' \succ R$ is indeed needed to ensure that $R$ is actually ``irrelevant''. For example, consider the following minimal representatives which are unrelated under $\succ$\footnote{We stress that these FPO types are indeed inequivalent as the frames are lists and FOP maps map the first element of the list of an FPO to the first element of the list of the other FPO, the second element to the second element, etc. In our graphical notation, exchanging the positions of frame elements on the page hence changes the equivalence class, while exchanging the positions of internal elements can always be done freely. Physically, this captures the idea that there is a difference whether Alice can send a message to Bob (which, choosing to arbitrarily associate Alice with the left half of an FPO and Bob with the right half of an FPO, we could associate to the FPO on the left of \cref{leftrightchannel}) or Bob can send a message to Alice (which would then correspond to the FPO on the right of \cref{leftrightchannel}.}
\begin{equation}\label{leftrightchannel}
\input{Diagrams/leftrightchannel.tikz}.
\end{equation}
We could imagine a process theory where every process which has the minimal representative on the left also always has the minimal representative on the right. Still, the minimal representative on the right is not irrelevant as it can be $\mathcal{C}$-locally embedded in situations where the one on the right cannot be.

No minimal representative is irrelevant for all process theories. The process theory $\mathbf{FreeProc}$ which we define in \cref{conversionappendix} is a process theory for which all minimal representatives are relevant. Indeed, this fact is essentially what we use to prove the converse direction of \cref{conversion}.

On the other hand, we can also come up with an example of a process theory where the only relevant minimal representative is the most useful one (i.e., the FPO type that contains only the frame and no order relations between them). The process theory $\mathbf{TrivProc}$ consists of the closure under parallel and sequential composition of the following: a primitive system $A$, the discard effect on $A$ and preparation of a fixed state on $A$. This still works if we add more primitive systems or allow for general effects and states as long as we do not allow any non-separable states or effects.

\subsection{Terminal theories}\label{sec:causaltheories}
A class of process theories of particular interest are terminal theories \cite{chiribella2010probabilistic,coecke2013causal,coecke2014terminality}. When working with spacetime structure, very often one will want to work with these theories and so it makes sense to give them special consideration in our framework as we did already in \cref{sec:EmbedAndSignalling}.

Given some implementation $i_\mathtt{F}$, we can use the properties of the discard \cref{causalconditions} to get rid of any effects in the diagram as illustrated in the example below. 
\begin{gather}\label{causalabsorb}
\begin{aligned}
\input{Diagrams/causaldecompex1.tikz} \quad &= \quad \input{Diagrams/causaldecomp2.tikz} \\
&= \quad \input{Diagrams/causaldecomp3.tikz} \\
&= \quad \input{Diagrams/causaldecomp4.tikz}
\end{aligned}
\end{gather}
First, factorize all effects into separate discards as in the first line above. Now merge each discard with the box it is connected to as in the second line above. This can then be repeated for any boxes that are now effects, for example, for $\mathtt{F}'$ in the last line above, which can be turned into a pair of discards. The end result is a new implementation $j_\mathtt{F}$ where the only remaining effects are ones directly connected to an external input.

Such implementations correspond to framed partial orders where all maximal elements of the FPO are inputs or outputs (note that any discard connected to an external input is identified with the corresponding frame element under $\mathcal{G}$). The relevant minimal representatives for terminal theories are then the ones where all maximal elements of the minimal representative are part of the frame\footnote{We remind the reader of our convention of saying an FPO type has some property if the property defines an invariant on the FPO type, i.e., all its labellings have this property. Hence, the above statement is explicitly: The relevant minimal representatives for terminal theories are then the ones where for every labelling of the FPO type all maximal elements of the FPO are part of the frame.}. Note that any minimal representative with this property is indeed relevant for the collection of terminal process theories. This is analogous to the proof of the converse direction of \cref{conversion}. We can consider the same free process theory but impose terminality on it (this is also known as the affine reflection of the process theory \cite{huot2019universal}). Then, for any minimal representative without internal maximal elements we can construct a diagram which cannot be rewritten into a more useful implementation.

\subsubsection{Lists of relevant minimal representatives for terminal theories}\label{sec:simplemin}

In this section, we list the relevant minimal representatives for terminal theories of the classes  \MINOtype{0}{2}, \MINOtype{0}{3}, \MINOtype{1}{3}, \MINOtype{2}{2} (the cases with less than two outputs are trivial and so we omit them). Minimal representatives that can be obtained via ``parallel composition'' of minimal representatives (see \cref{parmin} in \cref{sec:addprop}) or via permutation of the frame elements of the listed minimal representatives will not be listed explicitly (except in the case of \MINOtype{0}{2}). The minimal representatives are represented by their Hasse diagrams.

The proofs that all of these FPO types are indeed minimal representatives and that there are no further minimal representatives can be found in \cref{sec:minproofs}. We omit the proofs for the ordering relationship between these minimal representatives as it is easy to verify them.

\begin{restatable}{thm}{thm:simplemin}\label{simplemin}
The set of relevant minimal representatives for terminal process theories of the classes \MINOtype{0}{2}, \MINOtype{0}{3}, \MINOtype{1}{2}, \MINOtype{1}{3} and \MINOtype{2}{2} are exactly those in \cref{eq:0I2O,eq:0I3O,eq:1I2O,eq:1I3O,eq:2I2O} and those which can be obtained from these via permutation of the frame elements or via \cref{parmin}.
\end{restatable}
The minimal representatives for the \MINOtype{0}{2} case are up to permutation of frame elements
\begin{equation}\label{eq:0I2O}
\input{Diagrams/0I2O.tikz}.
\end{equation}
The minimal representatives for the \MINOtype{0}{3} case are up to permutation of frame elements and \cref{parmin}
\begin{equation}\label{eq:0I3O}
\input{Diagrams/0I3O.tikz}.
\end{equation}
The minimal representatives for the \MINOtype{1}{2} case are up to permutation of frame elements and \cref{parmin}
\begin{equation}\label{eq:1I2O}
\input{Diagrams/1I2O.tikz}.
\end{equation}
The minimal representatives for the \MINOtype{1}{3} case are up to permutation of frame elements and \cref{parmin}
\begin{equation}\label{eq:1I3O}
\input{Diagrams/1I3O_v3.tikz}.
\end{equation}
The ellipsis represents the fact that the ``zigzag'' structure formed by the internal elements in these minimal representatives can be arbitrarily long. We can define labellings of the zigzags (and thus the FPO types of the zigzags) explicitly as follows: The $N$-zigzag of class \MINOtype{2}{2} consists of the input $I$, the outputs $O_1, O_2, O_3$ and $2N$ internal elements labelled $x_1,..., x_{2N}$ such that $I < x_1$ and $x_{2N} < x_{2N-1}, O_3$ and $x_{2n}< x_{2n-1}, x_{2n+1} < O_1, O_2$ for all $n=1,...,N-1$. The longer the zigzag structure is (i.e., the more elements it contains), the more useful it is under $\succ$. For \MINOtype{2}{2}, we get an overall simpler structure, but we have similar zigzag minimal representatives,
\begin{equation}\label{eq:2I2O}
\input{Diagrams/2I2O_v3.tikz}.
\end{equation}
Again we can define labellings of the zigzags (and thus the FPO types of the zigzags) explicitly as follows: The $N$-zigzag of class \MINOtype{2}{2} consists of the inputs $I_1, I_2$, the outputs $O_1, O_2$ and $2N+1$ internal elements labelled $x_1,..., x_{2N+1}$ such that $I_1 < x_1$ and $I_2 < x_{2N+1}$ and $x_{2n}< x_{2n-1}, x_{2n+1} < O_1, O_2$ for all $n=1,...,N$. 

All of these causal structures have received much attention in the literature, with the exception of the zigzag ones (although, the least useful zigzag, i.e., the 1-zigzag, minimal representative has been studied under the name non-local computation \cite{ishizaka2008asymptotic}). It is therefore an interesting open question as to whether the zigzag causal structures are also relevant for quantum information processing.  We make some progress on this question in \cref{sec:quantum}.

\subsection{Theories with copying}\label{sec:markov}

If a theory allows for an appropriate notion of copying, together with a property that intuitively means that any randomness can ultimately be viewed as uncertainty about the initial conditions, we can reduce the number of relevant minimal representatives significantly, namely to the ones where the internal elements are all minimal elements of the minimal representative. This reduction is essentially equivalent to the concept of exogenisation in classical causal modelling \cite{evans2018margins} which states that all the latent nodes (internal elements in our language) in a classical causal model can always be assumed to be parentless variables (i.e., they are minimal elements). All other minimal representatives are irrelevant. Let us first discuss what this statement means for the minimal representatives from the previous section before we go on to formalise it. For the cases \MINOtype{0}{2}, \MINOtype{1}{2} and \MINOtype{0}{3}, all minimal representatives are potentially relevant. However, for \MINOtype{1}{3} and \MINOtype{2}{2}, the zigzags become irrelevant. What is more, for \MINOtype{2}{2} the trivial minimal representative is also irrelevant, meaning \textit{all} minimal representatives less useful than the two-way communication minimal representative are irrelevant. Thus, in theories with copying \MINOtype{2}{2} processes can always be achieved via two-way communication.

A copying process is a process with one input and two outputs

\begin{equation}
\input{Diagrams/copy.tikz}
\end{equation}

with the properties

\begin{equation}
\input{Diagrams/copycopyright.tikz} \quad = \quad \input{Diagrams/copycopyleft.tikz}, \quad \input{Diagrams/copydiscleft.tikz} \quad = \quad \input{Diagrams/identityLong.tikz} = \quad \input{Diagrams/copydiscright.tikz}, \quad \input{Diagrams/copyswap.tikz} \quad = \quad \input{Diagrams/copylong.tikz}.
\end{equation}

Categories that admit such a process are called cd categories \cite{cho2019disintegration} and we will extend this nomenclature to process theories (i.e., a process theory is cd if the underlying category is cd).

A process $\mathtt{F}$ in such a process theory is called deterministic iff
\begin{equation}
\input{Diagrams/Fthencopy.tikz} \quad = \quad \input{Diagrams/copythenF.tikz}.
\end{equation}

A deterministic dilation of a process $\mathtt{f}$ is a deterministic process $\mathtt{F}$ such that
\begin{equation}
\input{Diagrams/dilationlhs.tikz} \quad = \quad \input{Diagrams/dilationrhs.tikz}
\end{equation}
for some state $\rho$.

\begin{restatable}{prop}{markovreduction}\label{markovreduction}
 Let $\mathbf{Proc}$ be a process theory that is cd and where every process has a deterministic dilation. Then, the relevant minimal representatives are (a subset of) those where all internal elements are minimal elements of the minimal representative. All other minimal representatives are irrelevant and there are no quasi-relevant minimal representatives. 
 
 If, moreover, the theory is deterministic, i.e., one in which every process is deterministic, then the relevant minimal representatives are (a subset of) those without any internal elements. All other minimal representatives are irrelevant and there are no quasi-relevant minimal representatives.
\end{restatable}

The proof of this statement can be found in \cref{relevancecq}.

From this it also immediately follows that any process in such a process theory, can be implemented by the minimal representative which consists only of the frame and every input is in the past of every output (depicted below for the case \MINOtype{3}{4})

\begin{equation}\label{classtriv}
\input{Diagrams/classtriv3i4o.tikz}.
\end{equation}

One example of a theory that is cd and admits of deterministic dilations is (probabilistic) classical theory. In this case, we can show that the relevant minimal representatives are precisely the set where all internal elements are minimal elements of the minimal representative.
\begin{restatable}{prop}{classicalall}\label{classicalall} 
Let $R$ be a minimal representative such that all its internal elements are minimal elements of the minimal representative. Then, $R$ is a relevant minimal representative for probabilistic classical theory.
\end{restatable}

The proof of this statement can be found in \cref{relevancecq}.

In the introduction, we argued that no-signalling is not sufficient as a principle for embeddability. In the very restricted setting of deterministic classical theory our notion of embeddability matches the one under the no-signalling principle, that is processes are embeddable iff the no-signalling principle is satisfied. Note that this is already not the case anymore for probabilistic classical theory. For example, 
\begin{equation}
\input{Diagrams/splitminrep22.tikz} \quad \succ \quad \input{Diagrams/bellminrep2.tikz}
\end{equation}
are both relevant minimal representatives for probabilistic classical theory, hence, for either there exist processes for which these are the most useful minimal representatives. Both of these minimal representatives imply no-signalling between the two sides but the one on the LHS is easier to embed. Hence, there exist processes in probabilistic classical theory which have the same signalling relations but do not embed in the same way.

Finally, to illustrate the results of this section, we depict graphically how the relevant minimal representatives in cd theories with deterministic dilation compare to the ones for all terminal theories for the case $\MINOtype{2}{2}$, i.e., \cref{eq:2I2O},
\begin{equation}\label{eq:relevant2I2Oclass}
\input{Diagrams/relevant2I2Oclass.tikz}
\end{equation}
\subsection{Quantum theory}\label{sec:quantum}

\subsubsection{Classical processes in quantum theory}

In quantum theory, the story is, perhaps not surprisingly, significantly more complex than in classical theory. However, the fact that classical theory is a subtheory of quantum theory gives us some information about processes with classical inputs and outputs (i.e., the device-independent case). As we will argue in more detail below, it follows from this result that the device-independent processes in quantum theory can always be implemented with the minimal representative of \cref{classtriv} (or its generalisation to the arbitrary \MINO  case) and that the zigzag minimal representatives are never the most useful minimal representatives of a process in the $\MINOtype{1}{3}$ and $\MINOtype{2}{2}$ cases. Whether the same holds true for generalisations of the zigzags in the general $\MINO$ case remains an open question.

If a process $\mathtt{F}$ with classical inputs and outputs has an implementation corresponding to a minimal representative $R$ in classical theory, then the same is true in quantum theory. The converse does not hold, hence, in general $\mathtt{F}$ viewed as a classical process within quantum theory can have more useful minimal representatives than the ones it has when viewed as a classical processes. The most-well known example is a conditional probability distribution which violates a Bell inequality. The most useful minimal representative that implements this process in classical theory is the one-way communication minimal representative. From this it trivially follows that the one-way communication minimal representative can also be achieved in quantum theory. However, we know that we can also realize it with the Bell minimal representative. We can generally say that there exists a quantum advantage for a device-independent process if the process has a minimal representative when viewed as a process within quantum theory, which it does not have when viewed as a process within classical theory. 

From this fact, it also immediately follows that any device-independent process in quantum theory can be implemented by the minimal representative which consists only of the frame and where every input is in the past of every output (cf. \cref{classtriv}). For the \MINOtype{2}{2} case, this is the two-way communication minimal representative. Hence, in a slight abuse of our terminology, we could say that for $\MINOtype{2}{2}$ device-independent processes, the zigzags (as well as the trivial minimal representative) from \cref{eq:2I2O} are irrelevant. 

One can now wonder if the zigzags are relevant in the case of \MINOtype{1}{3}, since the trivial minimal representative
\begin{equation}
\input{Diagrams/1i3otriv.tikz}
\end{equation}
is already the least useful classically relevant minimal representative. In principle, it could be that there exists a probability distribution, $P(o_1, o_2, o_3 | i)$, where $o_1, o_2, o_3$ are the values of the random variables associated with the ouputs and $i$ is the value of the random variable associated with the input, only realisable in the trivial way in classical theory, but which can be realised with a zigzag in quantum theory, i.e.,
\begin{equation}
P(o_1, o_2, o_3 | i) \quad = \quad \input{Diagrams/1I3O1zigzag.tikz}
\end{equation}
if $\mathtt{s}, \mathtt{f}$ are allowed to be a quantum state and a quantum map respectively, but not if they are required to be stochastic maps.  However, notice that any such probability distribution cannot signal from the input $I$ to the output on the right $O_3$,
\begin{equation}
\input{Diagrams/1I3O1zigzagdiscard.tikz} \quad = \quad \input{Diagrams/1I3O1zigzagdiscard2.tikz}.
\end{equation}
Hence, the probability distribution must factorise according to\footnote{To see this let $P(ab|c)$ be a probability distribution such that $c$ does not signal to $b$, i.e., $\sum_a P(ab|c) = P(b)$. Let $P(c)$ be an arbitrary probability distribution for $c$. Note that $P(bc) = P(b|c) P(c) = P(b) P(c)$ by assumption. Then, $P(ab|c) P(c) = P(abc) = P(a|bc) P(bc) = P(a|bc) P(b) P(c)$. Since this holds for arbitrary $P(c)$, we have that $P(ab|c) = P(a|bc) P(b)$.}
\begin{equation}
P(o_1, o_2, o_3 | i) = P(o_1, o_2| o_3, i) P(o_3).
\end{equation}
However, this decomposition of the probability distribution defines a classical implementation according to the 1-zigzag
\begin{equation}
P(o_1, o_2, o_3 | i) \quad = \quad  \input{Diagrams/1I3Ozigzagclass.tikz}.
\end{equation}
Hence, any probability distribution which is implemented by the 1-zigzag in quantum theory is also implemented by the 1-zigzag in classical theory. But then, since the 1-zigzag (as well as all other zigzags) is irrelevant in classical theory, the process must have an implementation according to the minimal representative just above the zigzags,
\begin{equation}
\input{Diagrams/1I3Otriplecc.tikz} \quad \in \quad \mathcal{R}_{P(o_1, o_2, o_3 | i)}.
\end{equation}
This then also holds for quantum theory. Therefore, the zigzags are also irrelevant for device-independent processes in the case \MINOtype{1}{3}.

Note that this does not imply that the set of minimal representatives that is relevant for device-independent processes is (a subset of) the one that is relevant for classical theory. Indeed, the authors of \cite{centeno2024significance} give a counter-example for the \MINOtype{0}{4} minimal representative 
\begin{equation}
\input{Diagrams/danisminrep.tikz}.
\end{equation}
Since this minimal represenative has an internal element which is not minimal, it is not classically relevant. However, the authors show that there exists a probability distribution which in quantum theory can be implemented with this minimal representative, but there exists no implementation corresponding to a minimal representative that is both more useful and classically relevant.

All other minimal representatives from \cref{sec:simplemin} are relevant for device-independent processes. This follows from no-signalling requirements and arguments analogous to the one we used to show that all these minimal representatives are relevant in classical theory in \cref{classicalall}. 

\subsubsection{Zigzag minimal representatives are (quasi-)relevant}

Given the results of our discussion of device-independent processes, it is a natural question to ask whether the zigzags are irrelevant for quantum theory as a whole. This is not the case. Indeed, as it turns out and as we will show in the following, the quantum CNOT gate can be implemented using any zigzag minimal representative but not with the two-way communication minimal representative\footnote{Note that the classical CNOT gate can be implemented with the two-way communication minimal representative since it is a classical process. However, the classical CNOT gate is not the same as the quantum CNOT gate so our earlier comments about device-independent processes do not apply here.}. This implies that the zigzags are either all quasi-relevant or some subset (or even all of them) are relevant. 

Generally, we can show that all Clifford gates\footnote{A Pauli string is a state $P =  e^{i k \pi/2} \sigma_{i_0} \otimes ... \otimes \sigma_{i_n}$ for $k, i_0,...,i_n = 0,1,2,3$ and $n$ is arbitrary. The unitary $U: A \otimes B \rightarrow C \otimes D$ is a \MINOtype{2}{2} Clifford gate if for all Pauli strings $P_A \in A, P_B \in B$, there exist Pauli strings $P_C \in C, P_D \in D$ such that $U (P_A \otimes P_B) U^\dagger = P_C \otimes P_D$.} \cite{gottesman1998theory} can be implemented with all zigzag minimal representatives.

\begin{restatable}{prop}{cliffcasc}\label{cliffcasc}
In quantum theory, any \MINOtype{2}{2} unitary $U$ belonging to the Clifford group has an implementation associated with every zigzag minimal representative
\end{restatable} 

The proof can be found in \cref{relevancecq}. 

One way to look at this proposition is that we are using gate teleportation \cite{jozsa2006introduction} on the Clifford gates.

For the quantum CNOT gate, the zigzags are the most useful FPO types, proving that they are at least quasi-relevant, if not relevant.

\begin{restatable}{prop}{zigzagcnot}\label{zigzagcnot}
The \MINOtype{2}{2} quantum CNOT gate
\begin{equation}
\text{CNOT}(\ket{i} \otimes \ket{j}) = \ket{i} \otimes \ket{i \oplus j}
\end{equation}
has an implementation associated with every zigzag minimal representatives, but not the two-way communication minimal representative.
\end{restatable} 

The proof can again be found in \cref{relevancecq}. 

An interesting implication of the above result is that any embedded implementation of the CNOT gate can be improved in the sense of our ordering $\succ$. That is if one builds a CNOT gate (i.e., implements and embeds it in spacetime), then there always exists another implementation higher up $\succ$ and thus embeddable into more spacetime configurations. This does, however, come at the cost of increasing the number of components (i.e., boxes or nodes).

Finally, the question remains whether the zigzags are relevant or just quasi-relevant. It could be that any quantum process, like the CNOT, which can be implemented by one zigzag can be implemented by all of them. Further, if the zigzags are indeed relevant, then  in the \MINOtype{2}{2} case the set of relevant minimal representative for quantum theory matches exactly with the one for terminal theories in general. An interesting question would then be if this holds for the general \MINOtype{M}{N} case as well.  We leave both of these questions as open problems. 

On another note, we can ask if there are any zigzag minimal representatives which are relevant or at least quasi-relevant for device-independent processes. As discussed earlier, there are minimal representatives which are irrelevant for classical theory, which are nevertheless relevant for device-independent processes, however, the known examples do not include zigzag structures. A process, like the CNOT gate, for which the zigzags are the best implementations, could serve as a building block for an implementation of a larger device-independent process. It might then not be possible to remove the zigzag structure from the implementation, without making it less useful. 

Let us again illustrate our results by comparing how the minimal representatives from \cref{eq:2I2O} appear in quantum theory,
\begin{equation}
\input{Diagrams/relevant2I2Oquant.tikz}.
\end{equation}
The fact that the trivial minimal representative is relevant follows from the so-called position-based cryptography protocol \cite{chandran2009position}. The protocol is supposed to verify that an agent is at the position they claim and can be abstractly viewed as a CPTP map with two inputs, where one input is a quantum state and the other is a classical bit determining the measurement basis and two outputs, each of which output the outcome of the measurement in the chosen measurement basis on the quantum state. It is possible to achieve this protocol approximately with the 1-zigzag, which means that two agents, sharing sufficient entanglement, can collaborate to spoof a different position. However, it is not possible to achieve the protocol exactly with the 1-zigzag \cite{buhrman2014position}. 
\subsubsection{Other decompositional works through the lens of minimal representatives}

Finally, to finish out this section, we will discuss a number of other works, analysing them from the lens of our framework and in particular relate them to our notion of relevant and irrelevant minimal representatives.

\

In Ref. \cite{dolev2019constraining}, the author argues that the only thing that matters for the realisability of a quantum circuit is its input/output ordering as long as one has enough entanglement available. Two quantum circuits (or implementations, FPOs, minimal representatives) have the same input/output ordering if the ordering relationships between inputs and outputs is the same in both. In particular, this implies that the trivial minimal representative for the \MINOtype{2}{2} case would be irrelevant as it can be realised with the 1-zigzag.  There is, however, a key caveat to this. Ref. \cite{dolev2019constraining} uses both regular quantum teleportation and port-based quantum teleportation \cite{ishizaka2008asymptotic} to convert circuits into the desired form. In port-based teleportation, the teleported state ends up in one of $N$ ports, say the port $i$, and the value of $i$ is only known to the sender, not the receiver. However, if $N$ is not infinite, this protocol can only work probabilistically (otherwise, this protocol would allow for signalling as the receiver can pick one of the $N$ ports at random and choose the right one with probability $1/N > 0$). Thus, with our implicit requirement for equalities to be exact, this reduction does not work. However, this does pose an interesting future direction that consists of defining a suitable notion of approximation for our framework. Another limitation is that  Ref. \cite{dolev2019constraining} essentially only considers the case where there exists some global common cause which needs to be in the past of all elements which are not inputs. Considering our minimal representatives for the \MINOtype{2}{2}, we see that this is the case only for the Bell minimal representative, the 1-zigzag and the trivial minimal representative. Hence, even though the input/output ordering of the two-way communication minimal representative is the same as that of all the zigzags and the trivial minimal representative, we cannot apply the results of \cite{dolev2019constraining} to them. 

\

Ref. \cite{lorenz2021causal} posits the conjecture that for every unitary, there exists a faithful circuit consisting of unitaries, that is to say the circuit matches the unitary's signalling structure (they prove this conjecture for the cases where there are either at most three inputs or three outputs or precisely four inputs and four outputs). In terms of minimal representatives, this would imply a unitary $U: I_1 \otimes ... \otimes I_N \rightarrow O_1 \otimes ... \otimes O_N$ has some minimal representative such that $O_i > I_j$ in the minimal representative iff $I_j$ signals to $O_i$. One may therefore wonder if the signalling structure is sufficient to decide embeddability for unitaries at least. This is not the case. The reason is that that many minimal representatives can have the same input/output ordering, hence, many minimal representatives are faithful for a unitary, but the conjecture only tells us that at least one of them is a valid implementation. Consider, for example, the trivial minimal representative, the zigzags and the two-way communication minimal representative in \MINOtype{2}{2}. All of these have the same input/output ordering. In this case, the conjecture tells us (trivially) that any unitary with two-way signalling has the trivial minimal representative as an implementation. However, as we have seen, the CNOT gate can be implemented by any zigzag and has thus more useful faithful decompositions. Generally, two unitaries with the same signalling structure can have different implementation sets. For example, the CNOT gate and the process $\ket{ij}^A \ket{kl}^B \mapsto \ket{ik}^C \ket{jl}^D$ both have two-way signalling but the latter can be implemented by the two-way communication minimal representative.

Furthermore, it could also be that $U$ has a minimal representative $S$ which is not faithful, but also not less useful than any of its faithful minimal representatives. Since signalling between an input-output pair implies that the pair is related in any implementation, we then have that $S$ has at least one input-output pair which does not signal in $U$ but is related in $S$. This implies that $S$ can also not be more useful than the faithful minimal representatives. Thus, $S$ would have to be unrelated to the faithful minimal representatives. 

Finally, a small caveat to add here is that \cite{lorenz2021causal} only make their claim if the circuit is a routed quantum circuit \cite{vanrietvelde2021routed}. In a routed quantum circuit, wires carry indices which correspond to subspaces of the wire. For example, wire $A$ might carry the index $i$, wire $B$ the index $ij$ and wire $C$ the index $j$. Then, if the state in $A$ is in the subspace corresponding to, say $i=0$, and the state in $C$ is in the subspace $j=1$ then the state on $B$ will also be in the subspace corresponding to $ij = 01$. Such a routed quantum circuit can always be converted into a circuit made up of CPTP maps with the same circuit structure.  

\

Ref. \cite{renner2023commuting} shows that if two CPTP maps $\mathtt{F}: A \otimes B \rightarrow C \otimes B$ and $\mathtt{G}: B \otimes D \rightarrow B \otimes E$ commute, i.e., they fulfill the first equality in \cref{commuting}, and the resulting map is unital (i.e., maps the maximally mixed state to itself), then the composition factorises, i.e., there exists some appropriate CPTP $\mathtt{H}: B \rightarrow B \otimes B$ such that the second equality in \cref{commuting} is fulfilled. 
\begin{equation}\label{commuting}
\input{Diagrams/commutingops1.tikz} \quad = \quad \input{Diagrams/commutingops2.tikz} \quad = \quad \input{Diagrams/commutingops3.tikz}
\end{equation}
Note that this does not imply that whenever a process has both one-way communication minimal representatives (cf. \cref{leftrightchannel}), then it also has the Bell minimal representative.
The reason is that in addition to the process having implementations associated to both minimal representatives, it is necessary that both implementations use the same maps $\mathtt{F}, \mathtt{G}$ as seen in \cref{commuting}. A concrete counter-example is the PR box, which in quantum theory can be implemented with either one-way communication minimal representative but we cannot use the same maps for both.

This is an interesting case as similar to the case of \cref{zigzagcascade}, information about how exactly the process decomposes into less useful minimal representatives allows us to conclude that there must also be a more useful minimal representative.

\section{Conclusions}
\subsection{Summary}
In this work, we have developed a framework for studying process theories in a fixed background spacetime. We have argued that only relying on the signalling structure of processes for this purpose does not necessarily give us a sufficient condition for embeddability of processes into spacetime. Indeed, as we have shown this is not even the case in classical theory, unless we only consider deterministic processes. Instead, one must use the decompositional structure of processes. We then related the question of embeddability to a purely order-theoretic concept, defining a preorder on the framed partial orders associated to the decompositions of processes and showing the existence of minimal representatives. The latter tell us everything we need to know, not just about embeddability but also the decompositional structure of processes. Hence, in this work, order theory has proved a key tool for studying the decomposition of processes in process theories. 

We have also considered which minimal representatives are irrelevant when restricting to specific (classes of) process theories, due to there existing a more useful minimal representative that is an implementation whenever the irrelevant minimal representative is an implementation. We found all the relevant minimal representatives for terminal theories for zero or one input(s) and less than three outputs as well as for two inputs and two outputs. In particular, we found that for one input and three outputs as well as two inputs and two outputs, there is an infinite number of minimal representatives, which we called zigzags, whose causal structure has not been studied before. We showed that for theories with an appropriate notion of copying all minimal representatives which have internal non-minimal elements are irrelevant and for quantum theory we showed that the zigzag minimal representatives are not irrelevant. 

\subsection{Discussion}

Our framework gives us a clearer picture on how (non-)classicality depends critically on the spacetime embedding of processes. The interesting thing about Bell correlations or PR boxes is not, strictly speaking, the correlations themselves but that they can be achieved with the parties spacelike separated in some theories, but not in others. In our framework, the implementation sets of processes encode this theory-dependence. Additionally, the generality of our framework would allow us to go beyond the standard Bell scenario. In general, we can say that there is a quantum advantage for a process which is present in both classical and quantum theory if it has a more useful minimal representative in quantum theory. An analogous statement can be made for post-quantum advantages.

In general (and at least in quantum theory), the task of deciding embeddability is a computationally undecidable task as it subsumes a known undecidable task, namely deciding whether some correlations have a Bell realisation \cite{fu2021membership}. Various computational techniques, for example, from classical and quantum causal modelling may help here. Additionally, it is also likely that finding minimal representatives and in particular verifying that all minimal representatives have been found is similarly a difficult task, which may evade an analytical solution, but more research is required. 

The next step in this research programme is to define the compositional structure of process theories in spacetime. This quickly becomes quite involved when the embedded processes are ``overlapping'' in the sense that they cannot be viewed as strictly timelike or spacelike separated. For example, we can have situations where parts of one process are embedded in the future of another whilst other parts are embedded in its past.  We believe that the formalism for generalised process theories developed in Ref. \cite{selby2025generalised} will provide a suitable framework for studying this kind of compositional structure.

Many of our techniques could be further generalized, for example, one could replace localisation of endpoints with some other constraint and minimal representatives would still exist. This would potentially allow our framework to be extended to other interesting settings, like indefinite causal order processes in a background spacetime, with the right constraint.

In this work, we have assumed that processes can be localised to an arbitrary degree. In practice, the experimental setups which these processes represent will take up some finite spacetime volume and, unless these volumina are small compared to the distances in space and time between them, this assumption is not well-founded. To account for this fully, we would likely have to go beyond our current framework, which models spacetime as a partial order, for example, by considering spacetime to consist of regions which can be split and combined. Each box in an implementation would then be associated to a region of an appropriate volume. To avoid issues like Sorkin's paradox \cite{sorkin1993impossible} we might require that our embedding is still order-preserving, in the sense that if $\mathtt{f} \leq \mathtt{g}$ in the FPO, then the region associated to $\mathtt{f}$ lies fully in the past lightcone of the region associated to $\mathtt{g}$. In the case where disjoint regions are far apart or more accurately that the overlap of a region with the lightcone of another disjoint region is either the whole region or the empty set this simplifies again to our current setup.

Similarly, one could attempt to identify the set of relevant minimal representatives for other interesting collections of process theories than the ones we considered in \cref{sec:thspec}.

Another avenue for further research is studying what is the effect of weakening our embeddability condition again. In particular, one could study the spacelike extended PR boxes discussed in \cref{sec:embedPR}). To this end, one could take quantum theory in spacetime but then supplement it with such PR boxes. The advantage of using our framework is that there is a clear distinction between the ``normal'', i.e., embeddable via decomposition, and the ``exotic'', i.e., embeddable by some other or even unknown condition, processes.

As pointed out in \cref{sec:quantum}, extending the framework such that implementations only need to approximate the process would allow us to use the results of \cite{dolev2019constraining} to reduce the set of minimal representatives for quantum theory, for example, making all trivial \MINOtype{M}{N} minimal representative for $M, N > 1$ irrelevant. It would also be interesting to see if there are any other cases of minimal representatives besides the trivial minimal representative that are relevant in ``exact'' quantum theory, but become irrelevant in ``approximate'' quantum theory. An immediate question in that regard is whether the quantum CNOT gate can be approximately implemented with the two-way communication minimal representative, or more generally, if any quantum process that has all the zigzag minimal representatives can be approximately implemented with the two-way communication channel. A promising path in that direction would be checking \cref{evcond} for its robustness against small errors. Additionally, allowing approximations would also go nicely with our intuition to view implementations as experimental realisations of an abstract protocol, which are necessarily subject to error.

This framework may also lead to new protocols, or a cleaner formalisation of existing protocols, in relativistic quantum information processing. The zigzag minimal representative are again a potentially useful starting point for such a project. What is more, as these are irrelevant in classical theory, processes involving these minimal representatives have a strong chance to exhibit interesting non-classical features. The general \MINOtype{M}{N} case then also promises to hold even more exotic causal structures.

Minimal representatives may also be a useful tool in the field of causal modelling, in particular, for quantum (and post-quantum) causal modelling. An important open question in this field is which causal models are ``interesting'' if some of the nodes in a causal model are allowed to be quantum (cf, for example, \cite{henson2014theory}). A minimal representative can essentially be viewed as a causal model (up to the ambiguity of which graph a partial order corresponds to). They can thus serve as a starting point to find such interesting causal models. A first taste of this is our example of the quantum CNOT gate and the zigzag minimal representatives. This example also shows a need to go beyond the device-independent setting, as the zigzag minimal representatives only become ``interesting'' when we allow quantum inputs and outputs. However, it is likely that for more inputs and outputs we will find an even richer structure and, in particular, we might find device-independent protocols that require generalisations of the zigzags for their causal models. 

Another recent and independent work proposing an order-theoretic framework to study decompositions of processes is that of Ref. \cite{van2025order}. Indeed, the author there defines the concepts of circuits and morphisms which capture the same idea as our FPOs and FOP maps respectively. Further, the author essentially shows that circuits/FPOs arising as the so-called concept lattice of an input/output ordering are always minimal representatives and moreover they are the least useful minimal representative with that input/output ordering (this makes it the minimal representative for the most processes with a given input/output ordering, which is what the author is interested in). Applying this to the \MINOtype{2}{2} case would yield the trivial minimal representative and the two one-way communication minimal representatives (as well as some more minimal representatives which are not relevant in terminal theories) but not the Bell minimal representative, the two-way communication or the zigzags. The Bell minimal representative could however be obtained fairly straightforwardly from one of these ``concept lattice'' minimal representatives which is not relevant in terminal theories by removing its maximal internal elements. It would be interesting to see if the zigzags could be viewed as arising, not as a concept lattice of course, but some appropriately generalised mathematical structure.

\

{\it Acknowledgements---} 
MS is supported by the National Science Centre, Poland (Opus project, Categorical Foundations of the Non-Classicality of Nature, project no. 2021/41/B/ST2/03149). JHS was funded by the European Commission by the QuantERA project ResourceQ under the grant agreement UMO-2023/05/Y/ST2/00143. JHS conducted part of this research while visiting the Okinawa Institute of Science and Technology (OIST) through the Theoretical Sciences Visiting Program (TSVP). The authors would like to thank Tein van der Lugt and Augustin Vanrietvelde for insightful discussions about decompositional aspects of quantum theory (in particular, for suggesting the quantum CNOT gate as a potentially interesting process) as well as Elie Wolfe, Daniel Centeno and Marina Maciel Ansanelli for interesting exchanges regarding the relationship between this framework and causal modelling and the anonymous referees for Quantum who provided countless valuable suggestions to improve the manuscript. 

\bibliographystyle{unsrtnat}
\bibliography{generic2}


\appendix

\section{Introduction to process theories}\label{app:PTs}

In this appendix, we will expand upon the description of process theories which we gave in the main text. As mentioned there, a process theory is a collection of systems (which we can graphically represent as labelled wires) and processes (which we can graphically represent as labelled boxes) and rules on how to compose them. Particularly important classes of processes are states, which are processes without input wires, effects, which are processes without output wires, the identity process on a system which corresponds to doing nothing and the swap which corresponds to the interchange of two systems
\begin{equation}
\input{Diagrams/state1.tikz} \quad \input{Diagrams/effect.tikz} \quad \input{Diagrams/tensorsystemstriv2.tikz} \quad \input{Diagrams/swap.tikz}.
\end{equation}
The underlying mathematical formalism is that of symmetric monoidal categories (SMC), although, we note that in our case it would be more accurate to say that the underlying mathematical formalism is that of (coloured) Props \cite{carette2022propification,patterson2021wiring,selby2025generalised}. This is because in our framework we care about how the input and output systems of a process are decomposed into subsystems. Below we first give an introduction to SMCs and then say how this relates to process theories. At the end of the section, we give an introduction to Props.

A category $\mathcal{C}$ consists of a collection of objects, commonly denote by $|\mathcal{C}|$ and morphisms between these objects $f: A \rightarrow B \in \mathcal{C}$ for $A, B \in |\mathcal{C}|$. A standard example of a category is $\mathbf{Set}$, where the objects are sets and the morphisms are functions between the sets. A monoidal category $(\mathcal{C}, \otimes, I)$ consists of a category $\mathcal{C}$ equipped with a bifunctor $\otimes: \mathcal{C} \otimes \mathcal{C} \rightarrow \mathcal{C}$, called the monoidal product, together with a distinguished unit object $I \in |C|$ which satisfies associativity and unitality
\begin{gather}\label{eq:assouni}
\begin{aligned}
&f \otimes (g \otimes h) \cong  (f \otimes g) \otimes h \\
&A \otimes I \cong I \otimes A \cong A 
\end{aligned}
\end{gather}
for all objects $A \in |C|$ and all morphisms $f, g, h \in \mathcal{C}$ up to isomorphism (denoted by $\cong$). Further, it satisfies the interchange law
\begin{equation}\label{eq:interchange}
(g \circ f) \otimes (k \circ h) \cong (g \otimes k) \circ (f \otimes h)
\end{equation}
for all morphisms $f, g, h, k \in \mathcal{C}$ up to isomorphism. A strict monoidal category is a monoidal category where the above isomorphisms are in fact equalities. It is always possible for any monoidal category to find a strict monoidal category which is equivalent to the former \cite{mac1998categories}, hence, we can without any real loss of generality restrict to the case of strict monoidal categories.

A symmetric strict monoidal category is then one where $\otimes$ also satisfies as a symmetry condition, i.e., there exist for all systems $A, B \in |\mathcal{C}|$ an isomorphism (called the swap) $\sigma_{AB}: A \otimes B \rightarrow B \otimes A$ such that $\sigma_{AB}^{-1} = \sigma_{BA}$ and for all morphisms  $f: A \rightarrow C, g: B \rightarrow D$ it holds that 
\begin{equation}\label{eq:symmetry}
\sigma_{CD} \circ (f \otimes g) = (g \otimes f) \circ \sigma_{AB}
\end{equation}
for all $f , g \in \mathcal{C}$. 

The connection to process theories is summarized in the table below
\begin{table}[h!]
\centering
\begin{tabular}{|l|l|l|}
\hline
\textbf{Process theory version} & \textbf{SMC version} \\ \hline
System & Object \(A, B, C, \dots \in \mathrm{Ob}(\mathcal{C})\) \\ \hline
Process &  Morphism \(f : A \to B\) \\ \hline
Sequential composition &  Categorical composition \(g \circ f : A \to C\) \\ \hline
Parallel composition & Monoidal product \(f \otimes g : A \otimes B \to C \otimes D\) \\ \hline
Trivial system & Monoidal unit \(I\) \\ \hline
Symmetry of systems & Natural isomorphism \(\sigma_{A,B} : A \otimes B \to B \otimes A\) \\ \hline
Composite system & Monoidal product \(A \otimes B\) \\ \hline
Composite process & Monoidal product of morphisms \(f \otimes g\) \\ \hline
\end{tabular}
\end{table}

In process theories, we generally work with a diagrammatic calculus. In this case, parallel composition of two systems is achieved by simply putting the two wires together
\begin{equation}
\input{Diagrams/tensorsystems.tikz}.
\end{equation}
There exists a trivial system $I$ which we can think of as ``nothing''. When we compose a system with the parallel system we just obtain the original system
\begin{equation}
\input{Diagrams/tensorsystemstriv1.tikz} \quad = \quad \input{Diagrams/tensorsystemstriv2.tikz}.
\end{equation}
We can also represent parallel composition of processes in the same way as systems
\begin{equation}
\input{Diagrams/parcomp.tikz}
\end{equation}
We can sequentially compose two processes in which case we need to ensure that the types match, 
\begin{equation}
\input{Diagrams/seqcomp.tikz}
\end{equation}
i.e., the output wire of $\mathtt{f}$ is of the same type as the output wire of $\mathtt{g}$ (here, they are both the system $B$). Using these operations we can build more complex diagrams like
\begin{equation}
\input{Diagrams/diagram.tikz}.
\end{equation}
By following the mantra ``only connectivity matters'', associativity \cref{eq:assouni}, the interchange law \cref{eq:interchange} and symmetry \cref{eq:symmetry} are then automatically fulfilled in the diagrammatic calculus. We allow boxes to be ``moved around the page'' and wires to be arbitrarily deformed in a diagram as long as we keep the border the same. For example, for the symmetry condition, we have that
\begin{equation}
\input{Diagrams/symmetry1.tikz} \quad = \quad \input{Diagrams/symmetry2.tikz}.
\end{equation}
We leave it up to the reader to verify for themselves that one side of the equation can be obtained from the other by sliding the boxes $\mathtt{f}, \mathtt{g}$ around on the page without changing the positions of the ends of the incoming and outgoing wires.  Another example is the identity process which we represent as a wire without a box. Sequentially composing a process with the identity process is supposed to simply yield the process again. This is obvious in the diagrammatic calculus as lengthening the wire on a box would not change the connectivity of the diagram in any way
\begin{equation}
\input{Diagrams/trivcomp1.tikz} \quad = \quad \input{Diagrams/trivcomp2.tikz}
\end{equation}

We can cast (finite-dimensional) quantum theory as a process theory. There are many ways to do this, we present two of them here.

The systems are the spaces of trace-class operators of finite-dimensional Hilbert spaces $\mathcal{L}(\mathbb{C}^d)$ with the trivial system $I = \mathcal{L}(\mathbb{C}) \cong \mathbb{C}$ and processes are CPTP maps $\mathtt{F}: \mathcal{L}(\mathbb{C}^d \rightarrow \mathcal{L}(\mathbb{C}^{d'})$. States are then normalised density operators $\rho: \mathbb{C} \rightarrow \mathcal{L}(\mathbb{C}^d) \cong \mathcal{L}(\mathbb{C}^d)$ and there exists a unique effect for each system, the partial trace $\text{tr}: \mathcal{L}(\mathbb{C}^d) \rightarrow \mathbb{C}$. A measurement is represented by a CPTP map $\mathtt{F}: \mathcal{L}(\mathcal{C}^{d_1}) \otimes \mathcal{L}(\mathcal{C}^{d_x}) \rightarrow  \mathcal{L}(\mathcal{C}^{d_1}) \otimes \mathcal{L}(\mathcal{C}^{d_a})$ where the $d_1, d_2$-dimensional systems carry the quantum state before and after the measurement while the systems $d_x, d_a$ represent the classical measurement setting and the classical outcome. Hence, we assume that the state on system $d_a$ is diagonal in some preferred basis, at least whenever the input to $d_x$ is diagonal\footnote{Another solution, which is perhaps more elegant, would be to introduce a new type of system to model classical wires, instead of reusing the quantum systems as we do here. We will not get into this here, however, as classical wires do not figure strongly into our results.}.

Alternatively, we can drop trace-preservation and choose the processes to be the CPTNI maps $\mathtt{F}: \mathcal{L}(\mathbb{C}^d) \rightarrow \mathcal{L}(\mathbb{C}^{d'})$. In this case, the states are subnormalised density operators and besides the partial trace there exist many effects for each system which correspond to different measurement outcomes of a measurement where we discard the quantum state afterwards. Hence, we could think of this formalisation as quantum theory with post-selection.

We can think of the latter formulation as allowing post-selection while the former one does not.

This is a good point to talk about terminal process theory, which are ones satisfying \cref{causal}, reproduced below.

\causal*

We see that of our two formulations of quantum theory above only the former satisfies the terminality condition with the partial trace playing the role of the discard. In the second formulation, we can essentially use post-selection to signal. The diagram
\begin{equation}
\input{Diagrams/noncausalsig.tikz}
\end{equation}
depends on which effect we plug in. If we think of the state as being shared by Alice and Bob, then this means the state that Bob has depends on which effect we plug in on Alice's side, i.e., she could signal to him. The terminality condition thus encodes the requirement that signalling is only possible with some direct connection. 

We note that in the framework of operational probabilistic theories (OPT) the terminality condition instead states there is a unique \textit{deterministic} effect \cite{chiribella2010probabilistic}. While this seems like a weaker requirement, both this and \cref{causal} are physically the same as they both formalise the same idea, namely that future measurements should not influence past outcomes. The reason for these different definitions are ultimately due to differences in the underlying frameworks (OPTs versus process theories). Importantly, we do not lose anything by assuming there are no effects other than the discard. In the OPT framework, a measurement (called a test) is formalised as a collection of subnormalised effects (called an event), labelled by an index $\{e_x\}_{x \in X}$ for some outcome set $X$. The subnormalised effects $e_x$ are formally part of the framework but not considered to be physical by themselves. The terminality condition of OPTs then states that for every system there exists only one test from that system to the trivial system with a single outcome, i.e., $|X| =1$. In the framework of process theories, we would not model a measurement as a collection of effects, but as a process with a classical output wire which encodes the measurement outcome as we have done above. Hence, either framework can model what the other can model and they make the same predictions.

The above constitutes the standard narrative in the literature regarding the relationship between process theories, their diagrammatic representation, and (strict) symmetric monoidal categories. This, however, is not quite the full story, as, typically, monoidal categories come equipped with equalities such as $A\otimes B = C\otimes D\otimes E$, and these equalities do not play (particularly) nicely with the diagrammatic representation. If we try to incorporate them we end up with diagrammatic equations which (at first glance) do not look well typed, for example:
\beq
\input{Diagrams/Prop1.tikz}
\eeq
Typically this causes little problem when working with these things in practice so the situation can be (and typically is) glossed over. However, in our case, the LHS and the RHS of the above will have different embeddability properties! 

We therefore need to consider only (strict) symmetric monoidal categories which do not have this kind of problematic equality, which means working with (strict) monoidal categories in which the object monoid is free. These are known as (coloured) Props. Conveniently, and quite surprisingly, there is no real loss of generality to this. Any symmetric monoidal category is equivalent to a (coloured) Prop \cite{carette2022propification}. Essentially, given any (strict) symmetric monoidal category $\mathcal{C}$, we can construct a Prop $\mathcal{P_C}$ which has objects given by finite lists of objects in $\mathcal{C}$, where the monoidal product is given by concatenation of these lists, and where the monoidal unit is the empty list.  

\section{Convertability and the preorder are equivalent}\label{conversionappendix}

In this appendix, we will prove \cref{conversion}. In this proof, we will need the concept of a free process theory. A more detailed review of some of the concepts discussed here can be found in \cite{fong2018seven}.

\begin{restatable}{defi}{freeproc}\label{freeproc}
Let $\mathcal{F}$ be a set of processes. The free process theory $\mathbf{Free}(\mathcal{F})$ is the process theory which contains $\mathcal{F}$, where any non-identity process in $\mathbf{Free}(\mathcal{F}) \backslash \mathcal{F}$ can be decomposed using only processes in $\mathcal{F}$ and the identity processes and there are no further equalities between diagrams except those enforced by the requirement that $\mathbf{Free}(\mathcal{F})$ is a process theory. We call the processes in $\mathcal{F}$ primitive.
\end{restatable}

It will also be useful to spell out explicitly what it means for there to be ``no equalities between diagrams except those enforced by the requirement that $\mathbf{Free}(\mathcal{F})$ is a process theory''. In any process theory, the following are always allowed operations on a diagram:

\begin{itemize}
\item Composing multiple constituent processes

\item Replacing a single process with one of its decompositions

\item Adding the trivial system

\item Adding the trivial scalar

\item Adding the identity process
\end{itemize}

Hence, if a diagram can be converted into another using these steps, the diagrams must be equal regardless of any other features of the process theory. In an arbitrary process theory, there can be more than these trivial equalities, an example are the terminality conditions in causal theories \cref{causalconditions}. A free process theory is then such that there are no other equalities.

For the converse direction of the proof, we will need a free process theory that for some system $A$ (which cannot be further decomposed) has a primitive process from $A^{\otimes m}$ to $A^{\otimes n}$ for all $m, n \in \mathbb{N}_0$. 
\begin{restatable}{defi}{freeproc2}\label{freeproc2}
We define the free process theory $\mathbf{FreeProc} := \mathbf{Free}(\mathcal{F})$ for $\mathcal{F} = \{\mathtt{F}_{mn}: A^{\otimes m} \rightarrow A^{\otimes n}|m, n \in \mathbb{N}_0\}$, where $A$ is some non-trivial system such that $A = B \otimes C$ implies that either $B$ or $C$ is the trivial system. 
\end{restatable}

This process theory is essentially equivalent to what \cite{fong2018seven} calls the free prop on a signature.

For the proof, it will also be useful to consider some examples on how applying the above operations changes the associated FPO type. Consider a fragment of a circuit and the associated fragment of the FPO

\begin{equation}
\input{Diagrams/convertingdiags1.tikz}  \quad \substack{\mathcal{G} \\ \rightarrow} \quad \input{Diagrams/convertingfpo1.tikz}.
\end{equation}

We can compose the boxes in the diagram into a single box, obtaining (denoting any changes in the diagram or FPO type in blue)

\begin{equation}
\input{Diagrams/convertingdiags2.tikz}  \quad \substack{\mathcal{G} \\ \rightarrow} \quad \input{Diagrams/convertingfpo2.tikz}.
\end{equation}

We can add the trivial system $I$ between $\mathtt{F}$ and $\mathtt{H}$

\begin{equation}
\input{Diagrams/convertingdiags3.tikz}  \quad \substack{\mathcal{G} \\ \rightarrow} \quad \input{Diagrams/convertingfpo3.tikz}.
\end{equation}

We can add the trivial scalar $1$

\begin{equation}
\input{Diagrams/convertingdiags4.tikz}  \quad \substack{\mathcal{G} \\ \rightarrow} \quad \input{Diagrams/convertingfpo4.tikz}.
\end{equation}

We can add the identity process $\mathtt{id}$ on one of the wires

\begin{equation}
\input{Diagrams/convertingdiags5.tikz}  \quad \substack{\mathcal{G} \\ \rightarrow} \quad \input{Diagrams/convertingfpo5.tikz}.
\end{equation}

\conversion*
\begin{proof} 
We first prove that statement 1 implies statement 2. Let $\mathtt{F}$ be some arbitrary process in some arbitrary process theory and assume there exists some implementation $i_\mathtt{F}$ such that $\mathcal{G}(i_\mathtt{F}) = S_1 \in T_1$. Let $S_2 \in T_2$. Since $S_1 \succ S_2$, there exists a FOP map $\mathcal{E}: S_1 \rightarrow S_2$. The idea of the proof is to use the existence of $\mathcal{E}$ to construct a valid implementation $j_\mathtt{F}$ with $\mathcal{G}(j_\mathtt{F}) \cong S_2$. For every $y \in S_2$, define the pre-image $\mathcal{E}^{-1}(y)$. If $\mathcal{E}^{-1}(y) = \emptyset$, add the scalar 1 to $i_\mathtt{F}$. Note that the scalar 1 exists in every process theory as it is the trivial transformation on the trivial system. If $\mathcal{E}^{-1}(y) \neq \emptyset$, compose all boxes in $i_\mathtt{F}$ corresponding to elements in $\mathcal{E}^{-1}(y)$. Let us now show that this yields a valid diagram, i.e., one without cycles. This diagram has the following property: if one can reach a box $y_2$ from another box $y_1$ along a future-directed path, then $y_1 \leq y_2$ in $S_2$. To see this, note that there must exist boxes $x_1, x_2$ in $i_\mathtt{F}$ such that $x_2$ can be reached from $x_1$ and $\mathcal{E}(x_i) = y_i$ for $i=1,2$. Further, since $x_1 \leq x_2$ in $S_1$ by definition of $\mathcal{G}$, we have $y_1 \leq y_2$ due to order-preservation of $\mathcal{E}$. This property ensures that the new diagram has no cycles, i.e., is indeed a valid diagram, as otherwise $S_2$ would not be a partial order, but a preorder.   

The elements of the associated partial order of this new diagram are now in one-to-one correspondence with the elements of $S_2$. To complete the conversion, for every pair $x, y \in S_2$ with $x < y$ add a wire carrying the trivial system from the box corresponding to $x$ to the box corresponding to $y$  unless there already is a wire connecting them in this way (if $x$ or $y$ is a frame element which in the diagram is connected only to another frame element along an identity wire, this requires to first add an explicit identity process on the wire). Note that this yields another valid diagram without cycles. We start from a diagram without cycles. If there exists a future directed path from $x$ to $y$, then as we have shown $x \leq y$ in $S_2$, hence, $y \not < x$ and we do not add a wire from $y$ to $x$. Thus, no cycles can appear as a result of this step. Call the new diagram $j_\mathtt{F}$ and by construction we have $\mathcal{G}(j_\mathtt{F}) \cong S_2$. Since we only used composition and added the trivial system and identity process, the diagram $j_\mathtt{F}$ is still equal to $\mathtt{F}$ and thus an implementation. 

\
 
Let us now show that statement 2 implies statement 1. For this purpose, consider $\mathbf{FreeProc}$ from \cref{freeproc2}. We can build a process that has $T_1$ as an implementation by composing the basic processes $\mathtt{F}_{mn}$ such that the resulting diagram has the shape of the Hasse diagram of $T_1$. That is for every internal element of $T_1$, we have one instance of a process $\mathtt{F}_{mn}$ with $m$ equal to the number of parents and $n$ equal to the number of children and the frame elements correspond to a single copy of the system $A$ connected to an instance of $\mathtt{F}_{1n}$ in case of an input where $n$ is the number of children of the frame element and $\mathtt{F}_{m1}$ in case of an output where $m$ is the number of parents of the frame element. The processes are then wired together so that they have the same connectivity as the Hasse diagram. Let us call the resulting process $\mathtt{F}$ and the implementation we constructed $i_\mathtt{F}$. By construction, $\mathcal{G}(i_\mathtt{F}) \in T_1$. Thus, by assumption there exists another implementation $j_\mathtt{F}$ such that $\mathcal{G}(j_\mathtt{F}) \in T_2$. Since both are equal to the same process $\mathtt{F}$ and since $\mathbf{FreeProc}$ is a free process theory, there exists some procedure, i.e., list of steps, consisting of applying the basic axioms of process theories to $i_\mathtt{F}$ to obtain $j_\mathtt{F}$ (as was discussed in the beginning of this appendix). Note further that by construction the constituent processes in $i_\mathtt{F}$ are primitive, i.e., they cannot be further decomposed. This leaves the following valid steps: composing multiple constituent processes, adding the trivial system between two constituent processes and adding the identity process (including adding it as a scalar). Assume that a single of these steps is enough. We show that in each case, there exists a FOP map $\mathcal{E}: \mathcal{G}(i_\mathtt{F}) \rightarrow \mathcal{G}(j_\mathtt{F})$, i.e., $\mathcal{G}(i_\mathtt{F}) \succ \mathcal{G}(j_\mathtt{F})$, i.e., $T_1 \succ T_2$. 

\begin{itemize}
\item Composing multiple constituent processes: by construction, $\mathcal{G}(i_\mathtt{F})$ can be viewed as a coarse-graining of $\mathcal{G}(j_\mathtt{F})$, namely the elements corresponding to the composed processes are coarse-grained into a single element. Let $X$ be the set of elements in $\mathcal{G}(i_\mathtt{F})$ that are coarse-grained into a single element $x$ in $\mathcal{G}(j_\mathtt{F})$, then $\mathcal{E}(y) = x$ if $y \in X$ and $\mathcal{E}(y) = y$ if $y \not \in X$ defines a FOP map.

\item Adding the trivial system: In this case, $\mathcal{G}(i_\mathtt{F})$ and $\mathcal{G}(j_\mathtt{F})$ can be viewed as containing the same elements and if $x < y$ in $\mathcal{G}(i_\mathtt{F})$, then $x < y$ in $\mathcal{G}(j_\mathtt{F})$ (since adding the trivial system somewhere can at most add additional ordering relations). Thus, $\mathcal{E}(x) = x$ is a FOP map from $\mathcal{G}(i_\mathtt{F})$ to $\mathcal{G}(j_\mathtt{F})$.

\item Adding the trivial scalar: In this case, $\mathcal{G}(j_\mathtt{F})$ is the union of $\mathcal{G}(i_\mathtt{F})$ and a disconnected element. The map $\mathcal{E}(x) = x$ is a FOP map from $\mathcal{G}(i_\mathtt{F})$ to $\mathcal{G}(j_\mathtt{F})$.

\item Adding the identity process:  In this case, $\mathcal{G}(j_\mathtt{F})$ is the union of $\mathcal{G}(i_\mathtt{F})$ and an additional element. For all $x, y \in \mathcal{G}(i_\mathtt{F})$, it holds that if $x \leq y$ in $\mathcal{G}(i_\mathtt{F})$, then the same is true in $\mathcal{G}(j_\mathtt{F})$. The map $\mathcal{E}(x) = x$ is a FOP map from $\mathcal{G}(i_\mathtt{F})$ to $\mathcal{G}(j_\mathtt{F})$.
\end{itemize}
In all three cases, we see that $\mathcal{G}(i_\mathtt{F}) \succ \mathcal{G}(j_\mathtt{F})$. By transitivity, we then have $\mathcal{G}(i_\mathtt{F}) \succ \mathcal{G}(j_\mathtt{F})$ even if multiple steps are needed and thus $T_1 \succ T_2$. 
\end{proof}

\section{Additional properties of minimal representatives}\label{sec:addprop}

Given that labellings of minimal representatives always order-embed into other equivalent FPOs, we can ask if there is a map with convenient features in the other direction. In fact, it turns out that there always exists a FOP map which one can view as essentially a projection onto the minimal representative. 
\begin{restatable}{prop}{intomin}\label{intomin}
Let $R$ be a minimal representative and $T \sim R$. Then, for all $S \in R, S' \in T$ and all FOP maps $\mathcal{E}: S \rightarrow S'$, there exists a surjective FOP map $\mathcal{P}: S' \rightarrow S$ such that $\mathcal{P} \circ \mathcal{E} \circ \mathcal{P} = \mathcal{P}$.
\end{restatable}
\begin{proof}  
To begin, consider an arbitrary FOP map $\mathcal{E}':S'\to S$ which necessarily exists as, by assumption, $S\sim S'$. Now, as all of the partial orders we consider are finite, we have that $|\text{Im}(\mathcal{E}' \circ \mathcal{E})| < \infty$. This, together with the fact that $(\mathcal{E}' \circ \mathcal{E})(X) \subseteq X$ for any $X \subseteq S$, implies that there exists $m$ such that for all $m' > m$, 
\beq
\tilde{S} := \text{Im}((\mathcal{E}' \circ \mathcal{E})^{m'}) = \text{Im}((\mathcal{E}' \circ \mathcal{E})^m).
\eeq
In other words, $(\mathcal{E}' \circ \mathcal{E})|_{\tilde{S}}$ is a bijection and as its domain is finite, it is a permutation of a finite number of elements. For any permutation of a finite number of elements, there exists $n$ such that applying the map $n$ times yields the identity, i.e., in our case \beq(\mathcal{E} \circ \mathcal{E}')^n|_{\tilde{S}} = \text{id}_{\tilde{S}}.\eeq Note that we can choose $n > m$ (if this is not the case replace $n$ with $nm$).

We can therefore define $\mathcal{P}:= \mathcal{E}' \circ (\mathcal{E} \circ \mathcal{E}')^{n-1}$ which satisfies \beq (\mathcal{E}\circ \mathcal{P})|_{\tilde{S}}= (\mathcal{E} \circ \mathcal{E}')^n|_{\tilde{S}} =\text{id}_{\tilde{S}}.\eeq Now for any $x\in S'$, we have by definition
\beq
(\mathcal{E}\circ \mathcal{P})(x)=(\mathcal{E} \circ \mathcal{E}')^n (x) \in \text{Im}[(\mathcal{E} \circ \mathcal{E}')^n] = \tilde{S}
\eeq
and so 
\beq \mathcal{E} \circ \mathcal{P} \circ \mathcal{E} \circ \mathcal{P}(x) = \text{id} \circ \mathcal{E} \circ \mathcal{P}(x) = \mathcal{E} \circ \mathcal{P}(x).\eeq  
As this is true for all $x\in S'$ and as $\mathcal{E}$ is an order embedding (\cref{minrep}), and hence is injective, this simplifies to
\beq
\mathcal{P} \circ \mathcal{E} \circ \mathcal{P} = \mathcal{P}.
\eeq
All that remains is to show that $\mathcal{P}$ is surjective. By definition, we have that $\tilde{S}\subseteq S$ and hence $\tilde{S}\succ S$. Moreover, we have that $\mathcal{P}\circ \mathcal{E}:S\rightarrow \tilde{S}$ is a FOP map such that $S\succ \tilde{S}$. Putting these together we have that $S\sim \tilde{S}$. The fact that $\tilde{S}\subseteq S$ tells us that $|\tilde{S}|\leq |S|$, while the fact that $S$ is a labelling of a minimal representative and $S\sim \tilde{S}$ tell us that $|\tilde{S}|\geq |S|$. Putting these together, we find that $|S|=| \tilde{S}|$.  By \cref{minrep}, we then have $S \cong \tilde{S}$ and by $\cref{mincrit}$, $\mathcal{P}\circ \mathcal{E}: S \rightarrow \tilde{S}$ must be a relabelling and hence, $\mathcal{P}$ must be surjective.
\end{proof}

The combination of this result with \cref{mincrit} gives us another nice corollary, which in simple terms, lets us assume that it suffices to restrict our attention to FOP self-maps that are projections instead of all FOP self-maps for \cref{mincrit}.

\begin{restatable}{coro}{minproj}\label{minproj}
Let $T \in \mathcal{T}_{MN}$. Then $T$ is a minimal representative iff for all $S \in T$ and all FOP maps $\mathcal{E}: S \rightarrow S$ such that $\mathcal{E} \circ \mathcal{E} = \mathcal{E}$, $\mathcal{E}$ is a relabelling.
\end{restatable}

\begin{proof}
``$\Rightarrow$'': By \cref{minproj}, if $S$ is a labelling of a minimal representative, then any self-map is a relabelling.

``$\Leftarrow$'': We prove by contraposition. That is assume $T$ is not a minimal representative. Let $R \sim T$ be a minimal representative, $S' \in R$ and $\mathcal{E}: S' \rightarrow S$ an arbitrary FOP map, then by \cref{intomin} there exists a FOP map $\mathcal{P}: S \rightarrow S'$ such that $\mathcal{P} \circ \mathcal{E} \circ \mathcal{P} = \mathcal{P}$. Define $\mathcal{P}' := \mathcal{E} \circ \mathcal{P}$, then 
\beq
\mathcal{P}' \circ \mathcal{P}' = \mathcal{E} \circ \mathcal{P} \circ \mathcal{E} \circ \mathcal{P} = \mathcal{E} \circ \mathcal{P} = \mathcal{P}'.
\eeq
Note that $\mathcal{P}': S \rightarrow S$ as the composition of FOP maps is also a FOP map. Finally, $\mathcal{P}'$ is not surjective as 
\beq
|\text{Im}(\mathcal{P}')| = |\text{Im}(\mathcal{E} \circ \mathcal{P})| = |\text{Im}(\mathcal{P})| \leq |S'| < |S|,
\eeq
where we used that $\mathcal{E}$ is injective due to \cref{minrep}, that the cardinality of the image is less or equal than the cardinality of the domain in the third and finally that minimal representatives are strictly the smallest elements of equivalence classes in terms of their order.

\end{proof}
Further, minimal representatives turn out to be minimal in some more ways. They have strictly the fewest order relations between elements and their Hasse diagrams have the smallest size (i.e., number of edges). 
\begin{restatable}{coro}{minrepsize}\label{minrepsize}
Let $R$ be a minimal representative and $T \neq R$ is equivalent to $R$. The number of order relations in $R$, $r_R$ ($r_{T}$), is the number of unique pairs $x, y \in S$ ($S'$) for any $S \in R$ ($S' \in T$) such that $x \leq y$. Then, $r_R < r_{T}$.
\end{restatable}
\begin{proof}
We first note that the number of order relations in order-isomorphic sets must be the same. Since $R$ is the minimal representative of $T$, there exists a (non-surjective) FOE $\mathcal{E}: S \rightarrow S'$. Thus, $S$ is a relabelling of $\text{Im}(\mathcal{E}) \subsetneq S'$. Since a set has at least as many order relations as any of its subsets we find $r_R \leq r_{T}$. Since $\text{Im}(\mathcal{E}) \subsetneq S'$, there exists $x \in S'$ such that $x \not \in \text{Im}(\mathcal{E})$ and since $x \leq x$ due to reflexivity of $\leq$, we find that $r_{T} \geq r_R + 1 > r_R$. 
\end{proof}
\begin{restatable}{coro}{minrephasse}\label{minrephasse}
Let $R$ be a minimal representative and $T \neq R$ is equivalent to $R$. The size of the Hasse diagram (i.e., the number of edges) of $R$ is smaller than or equal to the size of the Hasse diagram of $T$. Their sizes are the same iff the FPOs of FPO type $T$ are relabellings of the union of an FPO of FPO type $R$ and an arbitrary number of disconnected elements.
\end{restatable}
\begin{proof}
Let $S \in R, S' \in T$, $\mathcal{E}: S \rightarrow S'$ an FOE and $\mathcal{P}: S' \rightarrow S$ a FOP map as in \cref{intomin}. Denote the Hasse diagram of $S$ (of $S'$) with $H$ (with $H'$). Denote with $E \subseteq S \times S$ (with $E' \subseteq S' \times S'$) the edges in $H$ (in $H'$). That is, we say that $(x, y) \in E$ iff $x < y$. Hence, $|E|$ ($|E'|$) is equal to the size of $H$ (of $H'$). We will show that $|E| \leq |E'|$ by constructing an injective map $\mathcal{D}: E \rightarrow E'$. Due to order-preservation and injectiveness of $\mathcal{E}$ if $(x, y) \in E$, then either $(\mathcal{E}(x), \mathcal{E}(y)) \in E'$ or there exist $x_1,...,x_n \in S' \backslash \text{Im}(\mathcal{E})$ such that $(\mathcal{E}(x), x_1), (x_1, x_2), ..., (x_n, \mathcal{E}(y)) \in E'$. If $x_i \in \text{Im}(\mathcal{E})$, then due to the fact that $\mathcal{E}$ is order-reflecting, it follows that 
\beq
\mathcal{E}(x) < x_i < \mathcal{E}(y) \Rightarrow x < \mathcal{E}^{-1}(x_i) < y
\eeq
and $(x, y)$ would not be an edge. 

In the former case, define $\mathcal{D}(x, y) = (\mathcal{E}(x), \mathcal{E}(y))$. In the latter case, define $\mathcal{D}(x, y) = (x_1, \mathcal{E}(y))$ if $\mathcal{P}(x_1) = x$, otherwise $\mathcal{D}(x, y)  = (\mathcal{E}(x), x_n)$. 

We then have that if $\mathcal{D}(x, y) = (x', y')$, then $\mathcal{P}(x') = x$ and $\mathcal{P}(y') = y$, i.e.,
\begin{equation}\label{pd}
\mathcal{D}(\mathcal{P}(x'), \mathcal{P}(y')) = (x', y').
\end{equation}
To see this, distinguish by cases: if $x' = \mathcal{E}(x)$, then $\mathcal{P}(x') = \mathcal{P}(\mathcal{E}(x))$. As $\mathcal{P} \circ \mathcal{E}$ is surjective, there exists $\tilde{x} \in S$ such that $\mathcal{P}(\mathcal{E}(\tilde{x})) = x$. Then,
\beq
x = \mathcal{P} \circ \mathcal{E}(\tilde{x}) = \mathcal{P} \circ \mathcal{E} \circ \mathcal{P} \circ \mathcal{E}(\tilde{x}) = \mathcal{P} \circ \mathcal{E}(x) = \mathcal{P}(x')
\eeq
where we used that $\mathcal{P} = \mathcal{P} \circ \mathcal{E} \circ \mathcal{P}$ (\cref{intomin}).

If $x' \neq \mathcal{E}(x)$, we have $\mathcal{P}(x') = x$ by definition of $\mathcal{D}$. Analogously, it follows that $\mathcal{P}(y') = y$. 

Let us now show that $\mathcal{D}$ is injective. Hence, assume that $(x, y), (x', y') \in E$ are such that 
\beq
\mathcal{D}(x, y) = \mathcal{D}(x', y') = (\tilde{x}, \tilde{y}) \in E'.
\eeq
Then it follows from \cref{pd} that $x = \mathcal{P}(\tilde{x}) = x'$ and analogously $y = \mathcal{P}(\tilde{y}) = y'$. Thus, $\mathcal{D}$ is injective and the first claim follows.

If the sizes of the Hasse diagrams are equal, then $\mathcal{D}$ must be bijective. Then, $\mathcal{D}^{-1}: E' \rightarrow E$ is well-defined and bijective and due to \cref{pd}, we have for all $(x', y') \in E'$ that 
\begin{equation}\label{dinv}
\mathcal{D}^{-1}(x', y') = (\mathcal{P}(x'), \mathcal{P}(y')) \in E.
\end{equation} 
Thus, we also have that
\begin{equation}\label{edgeeq}
(x', y') \in E' \Leftrightarrow (\mathcal{P}(x'), \mathcal{P}(y')) \in E.
\end{equation}
We now wish to show that $S'$ is the union of a relabelling of $S$ and an arbitrary number of disconnected elements. This is equivalent to showing that for all $x' \in S'$ which are not unrelated to all other elements, $x' \in \text{Im}(\mathcal{E}) \cong S$.

If $x'$ is not unrelated to all other elements, then there must exist at least one edge that includes $x'$, i.e., there exists $y' \in S'$ such that $(x', y') \in E'$ or $(y', x') \in E'$. Assume w.l.o.g. the former as the dual case is analogous. Then, due to \cref{dinv,edgeeq}, we have that $\mathcal{D}^{-1}(x', y') = (\mathcal{P}(x'), \mathcal{P}(y')) \in E$. On the other hand, using $\mathcal{P} \circ \mathcal{E} \circ \mathcal{P} = \mathcal{P}$ (\cref{intomin}), we have that
\beq
(\mathcal{P} \circ \mathcal{E} \circ \mathcal{P}(x'), \mathcal{P} \circ \mathcal{E} \circ \mathcal{P}(y')) = (\mathcal{P}(x'), \mathcal{P}(y')) \in E
\eeq
which due \cref{edgeeq} implies that $(\mathcal{E} \circ \mathcal{P}(x'), \mathcal{E} \circ \mathcal{P}(y')) \in E'$. We then have
\beq
\mathcal{D}^{-1}(\mathcal{E} \circ \mathcal{P}(x'), \mathcal{E} \circ \mathcal{P}(y')) = (\mathcal{P} \circ \mathcal{E} \circ \mathcal{P}(x'), \mathcal{P} \circ \mathcal{E} \circ \mathcal{P}(y')) = (\mathcal{P}(x'), \mathcal{P}(y')) = \mathcal{D}^{-1}(x', y')
\eeq
where we again used \cref{dinv} and that $\mathcal{P} \circ \mathcal{E} \circ \mathcal{P} = \mathcal{P}$. As $\mathcal{D}^{-1}$ is injective, it hence follows that
\beq
(\mathcal{E} \circ \mathcal{P}(x'), \mathcal{E} \circ \mathcal{P}(y')) = (x', y')
\eeq
and in particular $\mathcal{E} \circ \mathcal{P}(x') = x'$. This implies that $x' \in \text{Im}(\mathcal{E})$, completing the proof.

The converse statement is trivial.
\end{proof}
Minimal representatives also have the shortest longest (anti)chains (a chain is a totally ordered subset of a partial order while an antichain is a subset of a partial order such that none of its elements are related). Although, in this case we can no longer assume that the inequality is strict, even when excluding the FPOs that are order-isomorphic to the minimal representative with additional disconnected elements as above.
\begin{restatable}{coro}{minchain}\label{minchain}
Let $R$ be a minimal representative. Let $n$ be the length of the longest (anti)chain of $R$. Then, for all $T \sim R$, it holds that the longest (anti)chain in $T$ has at least length $n$. 
\end{restatable}
\begin{proof}
Let $S \in R, S' \in T$. By \cref{minrep}, there exists an FOE from $S$ into $S'$. An FOE is in particular order-preserving which implies that it maps chains to chains. However, an order-embedding is also injective and thus the image of the longest chain in $S$ is a chain of the same length. Thus, the longest chain $S'$ has at least length $n$. The result for antichains is proved analogously, using the fact that the order embedding is order-reflecting instead of using order-preservation. As these properties are clearly invariant under relabelling the statements are true for the FPO types as well.
\end{proof}

\begin{restatable}{prop}{2parents}\label{2parents}
Let $R$ be a minimal representative and $S \in R$. For all $x \in S$ internal, it holds that for all $y > x$, there exists $z > x$ such that $y, z$ are unrelated. An analogous statement holds if $y < x$. 
\end{restatable}
\begin{proof}
Assume there exists $y \in S$ such that no such $z \in S$ exists, i.e., for all $z > x$ it holds that either $y \leq z$ or $y \geq z$. We use \cref{mincrit} (or more specifically \cref{minproj}) to show that $S$ is not a minimal representative. Define the map $\mathcal{E}: S \rightarrow S$ as follows
\begin{equation}
\mathcal{E}(a) = \begin{cases}  y, &x \leq a \leq y \\ a, &\text{otherwise} \end{cases}.
\end{equation}
This is frame-preserving: for all $I \in S$ which are inputs, $x \not \leq I$ as $I$ is minimal in $S$ and $x \neq I$. Hence, $\mathcal{E}(I) = I$. On the other hand, for all $O \in S$ which are outputs, if $x \leq O \leq y$, then $O = y$ as $O$ is maximal in $S$. Hence, $\mathcal{E}(O) = y = O$. If $O$ is not in between $x$ and $y$, then $\mathcal{E}(O) = O$ by definition.

Additionally, $\mathcal{E}$ is order-preserving: For $a < b$, if $\mathcal{E}(a) = y$ it follows that $x \leq a \leq y$ by definition of $\mathcal{E}$ and thus by transitivity, $x \leq a < b$. Then by assumption, either $b \leq y$ in which case $\mathcal{E}(b) = y$ as well, or $b \geq y$ and thus $\mathcal{E}(b) = b \geq y \geq a$. If $\mathcal{E}(a) = a$, then either $\mathcal{E}(b) = b$, which is trivial, or $\mathcal{E}(b) = y$ in which case $b \leq y$ and thus by transitivity also $a < y$. We hence have $\mathcal{E}(a) = a < y = \mathcal{E}(b)$.

Thus, $\mathcal{E}$ is FOP. It is also not injective as $\mathcal{E}(x) = \mathcal{E}(y) = y$. 
 
The dual statement with $y < x$ follows analogously.
\end{proof}
The number of equivalence classes and thus minimal representatives quickly increases in terms of the number of inputs and outputs. Therefore, it is useful to have tools to construct new minimal representatives directly from already known ones. For example, the ``parallel composition'' of two minimal representatives is again a minimal representative. Let us first define what we mean by parallel composition of partial orders.
\begin{restatable}{defi}{parpo}\label{parpo}
Let $S_1 \in T_1 \in \mathcal{T}_{\MINOtype{N_1}{M_1}}, S_2 \in T_2 \in \mathcal{T}_{\MINOtype{N_2}{M_2}}$. Their parallel composition $S_1 \otimes S_2 \in \mathcal{S}_{(N_1+N_2)(M_1+M_2)}$ is the framed partial order with the following properties: the elements, inputs and outputs are the disjoint union of the elements, inputs and outputs of $S_1$ and $S_2$ and the order relations are given by $x \leq y$ iff $x, y \in S_1$ and $x \leq y$ in $S_1$ or $x, y \in S_2$ and $x \leq y$ in $S_2$. The FPO type $T_1 \otimes T_2$ is the FPO type of $S_1 \otimes S_2$.
\end{restatable}
This definition is consistent with parallel composition of processes, i.e., $\mathcal{G}(i_\mathtt{F} \otimes j_\mathtt{F}) = \mathcal{G}(i_\mathtt{F}) \otimes \mathcal{G}(j_\mathtt{F})$ for two implementations $i_\mathtt{F}, j_\mathtt{F}$. This follows directly from the definition of $\mathcal{G}$. 
\begin{restatable}{prop}{parmin}\label{parmin}
Let $T_1, T_2 \in \mathcal{T}_{\MINO}$. Then $T = T_1 \otimes T_2$ is a minimal representative iff $T_1, T_2$ are minimal representatives. 
\end{restatable}
The intuition here is that if $T$ were not a minimal representative, then by $\cref{mincrit}$ there exists a non-bijective self-map $\mathcal{E}: S \rightarrow S$ for some $S \in T$. The fact that $T_1$ and $T_2$ are essentially ``space''-like separated means the restrictions of $\mathcal{E}$ to $S_i$ must be self-maps on $S_i$. But then we would have at least one non-bijective self map on a minimal representative. 

\begin{proof}
``$\Leftarrow$'': Assume $T_1, T_2$ are minimal representatives. W.l.o.g., we can assume that $T_1, T_2$ cannot be further decomposed, i.e., if $T_1 = T'_1 \otimes T'_2$, then either $T_1 = T'_1$ and $T'_2$ is the FPO type of the empty set or vice versa, and analogously for $T_2$. Let $S \in T, S_i \in T_i$ for $i=1,2$ such that $S_1 \otimes S_2 = S$ and $\mathcal{E}: S \rightarrow S$ be a FOP map. Note that $\text{Im}(\mathcal{E}) \cap S_i \neq \emptyset$ as it contains at least the respective frame elements. We now claim that $x \in S_i$ implies $\mathcal{E}(x) \in S_i$ for $i=1,2$. Since we assumed that $T_i$, and thus $S_i$, cannot be further decomposed, there exists some path $x = x_0, x_1, ... x_n, x_{n+1} = O_j$ (in the sense that $x_k$ and $x_{k+1}$ are related for all $k=0,...,n+1$) from $x$ to some arbitrary frame element $O_j \in S_i$. We now prove the claim by induction. The frame element $O_j = x_{n+1}$ is mapped to itself and thus in $S_i$. Assume now for some $m \leq n+1$, all $x_k$ for $k \geq m$ are mapped into $S_i$ by $\mathcal{E}$. Since $x_{k-1}$ is related to $x_k$, $\mathcal{E}(x_{k-1})$ is related to $\mathcal{E}(x_k)$. But two elements can only be related if they are either both in $S_1$ or both in $S_2$. Thus, since $\mathcal{E}(x_k) \in S_i$, we find that $\mathcal{E}(x_{k_1}) \in S_i$. By induction, we have $\mathcal{E}(x) \in S_i$. Hence, $\mathcal{E}_i := \mathcal{E}|_{S_i}$ defines a FOP self-map on $S_i$. Since $S_i$ is a labelling of a minimal representative, $\mathcal{E}_i$ must be surjective (\cref{mincrit}) and since this holds for both $i=1,2$, we have that $\mathcal{E}$ must be surjective. Since this holds any $\mathcal{E}$, $T$ is a minimal representative due to \cref{minproj}.

``$\Rightarrow$:'' Assume now that $T_1$ is not a minimal representative. Then, there exists $S_1 \in T_1$ and a non-surjective FOP map $\mathcal{E}_1: S_1 \rightarrow S_1$. Define for some $S_2 \in T_2$,

\begin{equation}
\mathcal{E}(x) = \begin{cases} \mathcal{E}_1(x), \text{ if } x \in S_1 \\ x, \text{ if } x \in S_2 \end{cases}.
\end{equation}

Evidently, this is frame-preserving. Let now $x < y$, then $x \in S_1$ iff $y \in S_1$. Thus, either $\mathcal{E}(x) = \mathcal{E}_1(x) < \mathcal{E}_1(y) = \mathcal{E}(y)$ if $x, y \in S_1$ or $\mathcal{E}(x) = x < y = \mathcal{E}(y)$ if $x, y \in S_2$. In both cases, we have order-preservation. Thus, $\mathcal{E}$ is a FOP map and it is non-surjective since $\mathcal{E}_1$ is non-surjective. Therefore, $T$ is not a minimal representative.
\end{proof}

This result tells us that it suffices to find the minimal representatives with a single connected component, i.e., where every element is connected to every other element via a sequence or path of elements such that two neighbours are related. We can make a related statement for a more stringent notion of connectivity that excludes frame elements from the path. Let us first formally define both notions.

\begin{restatable}{defi}{path}\label{path}
Let $S$ be an FPO and $x, y \in S$. A path from $x$ to $y$ is a sequence $x_0, x_1, ... x_n \in S$ such that $x_0 = x, x_n = y$ and $x_i, x_{i+1}$ are related for all $0 \leq i \leq n$. If such a path exists, we say that $x$ is connected to $y$. The path length is the number of elements in the sequence. A path is called internal if all $x_i \neq x, y$ are internal and we say that $x$ and $y$ are internally connected. The (internal) distance of $x$ and $y$ is the minimum over all (internal) paths from $x$ to $y$. 
\end{restatable}

With this in hand, we can define the internal connected component.

\begin{restatable}{defi}{intcon}\label{intcon}
Let $S$ be an FPO and $x \in S$. The connected component of $x$ $\mathsf{Con}(x)$ is the set of all $y \in S$ that are connected to $x$. The internal connected component of $x$, $\mathsf{ConInt}(x)$ is defined as the set of all $y \in S$ that are internally connected to $x$.
\end{restatable}

For this notion of internal connection, we can show an analogous, although slightly weaker, result to \cref{parmin}. 

\begin{restatable}{prop}{onecon}\label{onecon}
Let $R$ be a minimal representative and $S \in R$. For all $x \in S$, it holds that $\mathsf{ConInt}(x)$ is a labelling of a minimal representative. Additionally, for $x, y \in S$, if $\mathsf{ConInt}(x) \neq \mathsf{ConInt}(y)$, then $S_x := \mathsf{ConInt}(x) \cup I \cup O, S_y := \mathsf{ConInt}(y) \cup I \cup O$ are unrelated under $\succ$. 
\end{restatable}
\begin{proof}
Assume there exists $x \in S$ such that $\mathsf{ConInt}(x)$ is not a labelling of a minimal representative. Then, there exists a non-surjective FOP map $\mathcal{E}: \mathsf{ConInt}(x) \rightarrow \mathsf{ConInt}(x)$. Let us show that the trivial extension of $\mathcal{E}$ to $S$ that acts as the identity on $S \backslash \mathsf{ConInt}(x)$ defines a non-surjective FOP map which would imply that $R$ is not a minimal representative. Let $y < z$. If neither or both are in $\mathsf{ConInt}(x)$, then order preservation follows from the fact that both $\mathcal{E}$ and the identity are order-preserving. Thus, assume $y \in \mathsf{ConInt}(x)$ and $z \not \in \mathsf{ConInt}(x)$. By the definition of internal connected components, this is only possible if $y$ is in the frame. But then both $y, z$ are mapped to themselves and order-preservation is trivial. A dual argument can be used to show order-preservation when $y \not \in \mathsf{ConInt}(x)$ and $z \in \mathsf{ConInt}(x)$. Thus, we have found a non-surjective FOP map on $S$ but this is impossible since $S$ is a minimal representative. Thus, $\mathsf{ConInt}(x)$ must be a labelling of a minimal representative.

Let now $x, y \in S$ such that w.l.o.g. $S_x \succ S_y$. Then, by definition, there exists a FOP map $\mathcal{E}': S_x \rightarrow S_y$. Define now the following self-map on $S$,

\begin{equation}
\mathcal{E}''(z) = \begin{cases}  \mathcal{E}'(z), & z \in S_x \\ z, &\text{otherwise} \end{cases}.
\end{equation}

Since both $\mathcal{E}'$ and the identity are frame-preserving, this is also frame-preserving.Order-preservation follows from an analogous proof to the one we used to show that internal connection components are labellings of minimal representatives. Thus, $\mathcal{E}''$ is a FOP map and it is not surjective as $\mathcal{E}''(z) \in S_y \subsetneq S$. Hence, by contradiction, we have that $S_x$ and $S_y$ are unrelated.

\end{proof}

\section{Proofs for relevant minimal representatives in terminal theories}\label{sec:minproofs}

In this appendix, we prove that the FPO types in \cref{sec:simplemin} are indeed minimal representatives and that they are all the relevant minimal representatives for terminal theories up to the restrictions stated. The ordering relations under $\succ$ between minimal representatives are easy to check so we will omit their proofs. In all proofs, we will assume a labelling such that the ordered list of inputs is $\mathcal{I} = (I_1, I_2,..., I_M)$ unless there is only one input, in which case, we will use simply $I$. We will deal with outputs analogously.

The following is a useful lemma which we will use explicitly and implicitly several times.

\begin{restatable}{lemma}{collapse}\label{collapse}
Let $R$ be a minimal representative and $S \in R$. If there exists some $x \in S$ such that for all inputs $I_i$ and all outputs $O_j$ it holds that $I_i \leq x \leq O_j$, then, $x$ is the only internal element of $S$ or $S$ consists only of the frame.
\end{restatable}

\begin{proof}
Assuming that such an $x$ exists, then we can define a FOP map $\mathcal{E}:S\to S$ such that
\begin{equation}
\mathcal{E}(y) = \begin{cases}  x, &y \text{ internal} \\ y, &\text{otherwise} \end{cases}.
\end{equation}
Clearly, this is frame-preserving. The only non-trivial situation when it comes to checking order-preservation is when we have an internal element that is related to a frame element. Let now $I_i$ be an input and $y$ an internal element such that $I_i < y$. Then, by our assumption about $x$, $I_i < x = \mathcal{E}(y)$. Analogously, we can check order-preservation in the case of an output and an internal element that are related. Hence, $\mathcal{E}$ is order-preserving. Since $R$ is, by assumption, a minimal representative, then any FOP map on $S$, such as $\mathcal{E}$, must be surjective. It is clear that $\mathcal{E}$ can only be surjective when there exists no internal element besides $x$ (or no internal element if $x$ is part of the frame) which completes the proof. 
\end{proof}

\subsection{Example 1: No inputs}

\begin{restatable}{prop}{0I2O}\label{0I2O}
The relevant minimal representatives for terminal theories with no inputs and two outputs are 
\begin{equation}\label{eq:0I2Ob}
\input{Diagrams/0I2O.tikz}.
\end{equation}
\end{restatable}
\begin{proof}
It is obvious that both FPO types are minimal representatives. Any FPO type containing only the frame is a minimal representative (as any FOP self-map has to map the frame to itself), meaning the FPO type on the left is a minimal representative. For the FPO type on the left, if it is not a minimal representative, then there exists a FOP map which maps the internal element to the frame. However, there is no frame element that is in the past of both outputs and as such no such FOP map can exist.

We will now show that these are the only two minimal representatives. Assume now $R$ is a minimal representative with a labelling $S \in R$. If there exist no internal elements, then the FPO type of $S$ is the one on the left-hand side of \cref{eq:0I2O}. Assume now there exists $x \in S$ internal. By \cref{2parents} and the fact that relevant minimal representatives for terminal theories do not have internal maximal elements, $x$ is in the past of both outputs. Hence, $x$ satisfies the requirements of \cref{collapse} and is thus the only internal element. Thus, the FPO type of $S$ is one on the right-hand side of \cref{eq:0I2O}.
\end{proof}

\begin{restatable}{prop}{0I3O}\label{0I3O}
The non-trivial\footnote{Here, a trivial minimal representative is one which can be obtained via \cref{parmin} or permutation of frame elements} relevant minimal representatives for terminal theories with no inputs and three outputs are 
\begin{equation}\label{eq:0I3Ob}
\input{Diagrams/0I3O.tikz}.
\end{equation}
\end{restatable}
\begin{proof}
Let $S$ be a labelling of one of the depicted FPO types and $\mathcal{E}: S \rightarrow S$ be a FOP map. Note that a necessary condition for $\mathcal{E}$ being a FOP map is that if $x < O_i$, then $\mathcal{E}(x) \leq O_i$ for all $x \in S$, but we see from the depicted partial orders that this is only possible if $\mathcal{E}(x) = x$ and thus $\mathcal{E}$ is surjective. This shows that the depicted FPO types are indeed minimal representatives.

Let us now show that we have found all minimal representatives. Let $S$ be a labelling of a minimal representative. Because of \cref{collapse}, we can assume that there is no internal element that is in the past of all three outputs, as the only such case is the third minimal representative that we have already found. Using additionally \cref{2parents} and terminality, for any internal element $x$, there are exactly two outputs in its future. For each pair of outputs $O_i, O_j$ with $i \neq j$ choose an internal element $x_{ij} < O_i, O_j$ if such an element exists. Define
\begin{equation}
\mathcal{E}(x) = \begin{cases}  x_{ij}, & x < O_i, O_j \text{ internal} \\ x, &\text{otherwise} \end{cases}.
\end{equation}
This is frame-preserving as an output is only in the past of itself and thus gets mapped to itself. Let now $x < y$. If (w.l.o.g.) $y < O_1, O_2$, then by transitivity $x < O_1, O_2$ and thus $\mathcal{E}(x) = x_ij = \mathcal{E}(y)$. If $y = O_1$, then $\mathcal{E}(x) = x_{1i} < O_1$ for either $i=2$ or $i=3$. Hence, $\mathcal{E}$ is order-preserving. If $\mathcal{E}$ is surjective as required by $S$ being a labelling of a minimal representative, then $S$ has for each pair $i, j$ at most one element $x_{ij}$ with the property $x_{ij} < O_i, O_j$. The non-trivial relevant minimal representatives that fulfill this property are exactly those depicted in \cref{eq:0I3O}.
\end{proof}

\subsection{Example 2: One input}

\begin{restatable}{prop}{1I2O}\label{1I2O}
The non-trivial relevant minimal representatives for terminal theories with one input and two outputs are 
\begin{equation}\label{eq:1I2Ob}
\input{Diagrams/1I2O.tikz}.
\end{equation}
\end{restatable}
\begin{proof}
The right FPO type is a minimal representative, as any FPO type that only contains the frame is a minimal representative. For the FPO type on the left, if it is not a minimal representative, then there exists a FOP map which maps the internal element to the frame. However, there is no frame element that is in the past of both outputs and as such no such FOP map can exist. 

Let us now show that we have found all non-trivial relevant minimal representatives. Note first that any internal element has to be connected to both outputs by \cref{2parents} and the fact that we are considering terminal theories. Moreover, note that if an internal element exists it cannot be connected to the input due to a combination of \cref{2parents,collapse}. The internal element would be the unique internal element of the FPO but then it cannot have more than one unrelated elements in the past as required by \cref{2parents}. Further, there cannot be more than one internal element (otherwise, we can obtain a non-surjective FOP map by choosing one internal element and mapping all other internal elements to that element). 

This leaves the FPO types we listed as options.
\end{proof}

\begin{restatable}{prop}{1I3O}\label{1I3O}
The non-trivial relevant minimal representatives for terminal theories with one input and three outputs are 
\begin{equation}\label{eq:1I3Ob}
\input{Diagrams/1I3O_v3.tikz}.
\end{equation}
\end{restatable}
\begin{proof}
Let $S$ be a labelling of a minimal representative and assume for now that if the input $I < x$, then $x$ is necessarily an output. Note that this is the case for all the FPO types in \cref{eq:1I3Ob} except the ones in the third row. We will show that given the assumption, an FPO type is a minimal representative iff it is one of those. 

Any FOP map $\mathcal{E}: S \rightarrow S$ must be surjective, which implies that, in particular, any FOP map $\mathcal{E}': S \backslash \{I\} \rightarrow S \backslash \{I\}$ must be surjective. To see this, note that
\begin{equation}
\mathcal{E}(x) = \begin{cases}  \mathcal{E}'(x), & x \neq I \\ I, &x = I \end{cases}.
\end{equation}
is a FOP map on $S$. Frame-preservation is obvious, while the only non-trivial case for order-preservation is for $I < x$.  But by assumption, then $x$ is an output and both $I$ and $x$ are mapped to themselves by $\mathcal{E}$.

By \cref{mincrit}, if $S$ is a labelling of a minimal representative, then $S \backslash \{I\}$ is a labelling of a minimal representative of class \MINOtype{0}{3}. Thus, we can view $S$ as the union of a labelling of a minimal representative of class \MINOtype{0}{3} with $I$ up to adding order relations between $I$ and other elements. Equivalently, this is a necessary condition for $S$ to be a labelling of a minimal representative.

Let us now strengthen the above condition to find one which is both necessary and sufficient. Note that any FOP map $\mathcal{E}: S \rightarrow S$ which is non-surjective would have to map some internal element $x \in S$ to $I$ (as otherwise $\mathcal{E}|_{S \backslash \{I\}}$ would be a self-map on $S \backslash \{I\}$ and non-surjective). A sufficient condition for a $S$ to be a labelling of a minimal representative is thus that no such FOP maps exist. This is the case if there exists no no $x \in S$ internal such that the future of $x$ is contained in the future of $I$.

Combining both conditions we obtain the following: $S$ is a labelling of a minimal representative iff $S \backslash \{I\}$ is a labelling of a minimal representative of class \MINOtype{0}{3} and there exists no $x \in S$ internal such that the future of $x$ is contained in the future of $I$. Comparing to \cref{eq:0I3O}, we then see that the top two rows as well as the bottom row of \cref{eq:1I3O} are exactly the FPO types that fulfill these two conditions.

Let us now show that an FPO $S$ as in the third row of \cref{eq:1I3O} is a labelling of a minimal representative. Let $\mathcal{E}: S \rightarrow S$ be a FOP map. Denote the internal elements with $x_1, ..., x_N$ from left to right. Additionally, denote $x_0 = I$ and $x_{N+1} = O_3$. Note that 
\beq
\forall i, m: \mathcal{E}(x_i) = x_m \Rightarrow \mathcal{E}(x_{i\pm1}) = x_m, x_{m\pm1}.
\eeq 
Iterating this fact, we find that 
\beq
\forall i,j, m: \mathcal{E}(x_i) = x_m \Rightarrow \mathcal{E}(x_{i+j}) = x_{m + k}, \quad k \in \mathbb{Z}, |k| < j.
\eeq
Plugging in $i=0$ and using frame preservation, i.e., $\mathcal{E}(x_0) = \mathcal{E}(I) = I = x_0$, together with the fact that the index is a non-negative number we find that 
\beq
\mathcal{E}(x_j) = x_k, \quad k \leq j. 
\eeq
Conversely, plugging in $i = N+1$, we find that 
\beq
\mathcal{E}(x_j) = x_k, \quad k \geq j.
\eeq
Thus, $\mathcal{E}(x_j) = x_j$ for all $j$, meaning $\mathcal{E}$ is the identity and thus surjective.

Finally, we show that there are no more minimal representatives. We already showed this if  the input is not related to any internal element.  Thus, assume now there exists $x > I$ internal. By \cref{2parents} and terminality we then have w.l.o.g. $I < x < O_1, O_2$. However, $I \not < O_3$ as otherwise by \cref{collapse} the minimal representative would be the one in the bottom right corner (and $x$ could not even exist). 

Further, assume for now that $S$ has only a single internal connection component $\mathsf{ConInt}(x)$ (cf. \cref{intcon}). Since we can assume that we also have only a single connection component due to \cref{parmin}, we can even assume $S = \mathsf{ConInt}(x)$. Let now $I = x_0, x_1, ..., x_m, x_{m+1} = O_3$ be a shortest internal path from $I$ to $O_3$. For all $y \in S \backslash \{O_1, O_2\}$, define $p(y)$ to be the length of the shortest internal path from $I$ to $y$ (note that this is well-defined as we have assumed that all elements of $S$ are internally connected). Define

\begin{equation}
\mathcal{E}(y) = \begin{cases}  x_i, & y \neq O_1, O_2 \text{ and } p(y) = i \\ x_m, &y \neq O_1, O_2 \text{ and } p(y) > m+1 \\ y, &\text{otherwise} \end{cases}
\end{equation}

Since $p(I) = 0, p(O_3) = m+1$, this is frame-preserving. Let now $y < z$. We consider the cases of $z$ being an output or not an output.

\textit{\underline{Case 1:} $z = O_j$ is an output:} We distinguish further the cases $p(y) < m$ and $p(y) \geq m$.

\textit{\underline{Case 1a:} $p(y) < m$:} It cannot be that $z = O_3$ as $y < O_3$ is impossible. This is because $I$ to $y$ and then $O_3$ would define an internal path from $I$ to $O_3$ with length $p(y) + 1 < m+1$, which is not possible by assumption. Thus, we have that $z = O_1, O_2$ and by \cref{2parents} and terminality, we also have that $y < O_1, O_2$. Note that the latter holds in particular for $x_{p(y)}$ and so we have order-preservation in this case, i.e.,
\beq
\mathcal{E}(y) = x_{p(y)} < O_k = \mathcal{E}(O_k), \quad k = 1,2.
\eeq
\textit{\underline{Case 1b:} $p(y) \geq m$:} By definition, we have that $\mathcal{E}(y) = x_m$. Note that $x_m$ is related to both $x_{m-1}$ and $O_3$ because they are neigbouring elements of a path. Thus, we have $x_m < O_3$ and also $x_m < x_{m-1} < O_1, O_2$ (if $x_m > x_{m-1}$, then we could drop $x_m$ from the internal path to obtain a shorter one, i.e., $I,x_1,...,x_{m-1},O_3$ defines an internal path from $I$ to $O_3$ of length $m$. Hence, it must be that $x_m < x_{m-1}$). Thus, the order is preserved, $\mathcal{E}(y) = x_m < O_j = \mathcal{E}(O_j)$ for $j=1,2,3$. 

\textit{\underline{Case 2}: $z$ is internal:} On the other hand, if $y=I$, then $p(z) = 1$ and $\mathcal{E}(y) = x_1 > I$. Consider now the case where $y, z$ are internal. As $y, z$ are related and by definition of $p(y), p(z)$ as the length of the shortest path to $y, z$, we have that $|p(y) - p(z)| = 1$. We distinguish the cases $p(y) = p(z)$ and $p(z) = p(y) \pm 1$.

\textit{\underline{Case 2a}: $p(y) = p(z)$:} The map $\mathcal{E}$ maps $y$ and $z$ to the same element, $\mathcal{E}(y) = x_{p(y)} = \mathcal{E}(z)$ and the order is preserved. 

\textit{\underline{Case 2b}: $p(y) = p(z) \pm 1$:} Denote the shortest (internal) path from $I$ to $y$ with $I = y_0, y_1,..., y_{p(y) -1}, y_{p(y)} = y$. Then, it holds for all $j < p(y) - $, $y_{2j} < y_{2j \pm 1}$, i.e., the even-numbered elements are minimal elements of the path (not necessarily of $S$!), the odd-numbered elements are maximal elements of the path. We can show this by induction. It holds for $I = y_0$ as inputs have no elements in their past. Assume it holds up to some $j$. Then, $y_{2j + 1}$ is a maximal element in the path and $y_{2j+2}$ is related to it. As it cannot be larger than a maximal element it holds that $y_{2j+1} > y_{2j+2}$. On the other hand, if $y_{2j+2} > y_{2j+3}$, then by transitivity $y_{2j+1} > y_{2j+3}$ and we could drop $y_{2j+2}$ from the path to obtain an even shorter than the shortest one. By contradiction, we thus also have $y_{2j+2} < y_{2j+3}$.

Note now that $y < z$ are the last two elements in a shortest internal path from $I$ to $z$. By the above, it must then be that $p(y)$ is even and $p(z)$ odd. Then, 
\beq
\mathcal{E}(y) = x_{p(y)} < x_{p(y) \pm 1} = \mathcal{E}(z)
\eeq
by definition of $\mathcal{E}$ on the one hand, while on the other hand $x_{p(y)} < x_{p(y) \pm 1}$ is again the property of shortest internal paths we proved above. 

This exhausts all cases and, thus, $\mathcal{E}$ defines a FOP map. Since $S$ is, by assumption, a labelling of a minimal representative, $\mathcal{E}$ is surjective due to \cref{mincrit}. For that, we must have that the elements $x_1, ... x_{m}$ are the only internal elements. Thus, $S$ is a labelling of one of the minimal representatives in the third row of \cref{eq:1I3O}.

Assume now $S$ is a labelling of a minimal representative which has more than one internal connection component. By \cref{onecon}, every internal connected component (possibly supplemented with missing frame elements) is a labelling of a minimal representative. Additionally, since we assume that there exists an internal element $x < I$, $\mathsf{ConInt}(x)$ can only be a labelling of one of the minimal representatives from the third row of \cref{eq:1I3O}.

By assumption, there exists $y$ internal such that $\mathsf{ConInt}(y) \neq \mathsf{ConInt}(x)$. By \cref{onecon} this immediately implies that $\mathsf{ConInt}(y)$ is not a labelling of a minimal representative from the third row of \cref{eq:1I3O} as these are all related under $\succ$. We see by inspection of the remaining minimal representatives that then $y$ is the only internal element in $\mathsf{ConInt}(y)$ and additionally $I \not < y$\footnote{Note that we also need to consider trivial minimal representatives here}. Then,
\begin{equation}
\mathcal{E}(z) = \begin{cases}  x_{m}, & z=y \\ z, &\text{otherwise} \end{cases}
\end{equation}
defines a non-surjective FOP map on $S$ since $x_{m} < O_1, O_2, O_3$. Thus, $S$ is not actually a labelling of a minimal representative.
\end{proof}

\subsection{Example 3: two inputs}

\begin{restatable}{prop}{2I2O}\label{2I2O}
The non-trivial relevant minimal representatives for terminal theories with two inputs and two outputs are 
\begin{equation}\label{eq:2I2Ob}
\input{Diagrams/2I2O_v3.tikz}.
\end{equation}
\end{restatable}
\begin{proof}
First note that due \cref{2parents} and terminality any internal element in a minimal representative is in the past of both outputs. This means it also cannot be mapped into an output by a FOP map. We distinguish by number of internal elements.

\textit{\underline{Case 1:} No internal elements:} Any FPO type without any internal elements is a minimal representative and it is is obvious that the listed ones in \cref{eq:2I2O} are the only (non-trivial) ones. 

\textit{\underline{Case 2:} One internal elements:} The FPO types with exactly one internal element in \cref{eq:2I2O} are also minimal representatives as otherwise we would be able to map the internal element to an element of the frame but it is clear that this is not possible in an order-preserving manner. These are also the only non-trivial minimal representatives with only one element . To see this, let $S$ be a labelling of a  minimal representative such that $x$ is its unique internal element. By \cref{2parents} and terminality, $x < O_1, O_2$. Also by \cref{2parents}, either $x$ is minimal in $S$ or $x > I_1, I_2$. In the latter case, this is then necessarily the FPO type which is the least useful minimal representative. In the former case, the only one that is potentially missing is the case where exactly one of the inputs is less than both outputs, w.l.o.g., $I_1 < O_1, O_2$. But then mapping $x$ into $I_1$ would define a non-surjective FOP map and $S$ would not be a labelling of a minimal representative. 

\textit{\underline{Case 3:} More than one internal elements:} That the last group of FPO types are minimal representatives follows from a proof completely analogous to the one used to show that the FPO types in the third row of \cref{eq:1I3O} are minimal representatives. That these are the only minimal representatives if there is only a single internal connected component can also be shown in an analogous manner.

Note now that all the minimal representatives with a single internal connected component are related under $\succ$. This means that by \cref{onecon} no minimal representative can have more than one internal connected component. Thus, we have found all minimal representatives.

\end{proof}

\section{Proofs for relevance of minimal representatives in cd process theories and quantum theory}\label{relevancecq}

In this appendix, we prove \cref{markovreduction,classicalall} from \cref{sec:markov} as well as  \cref{cliffcasc,zigzagcnot} from \cref{sec:quantum}.

\markovreduction*
\begin{proof}
Let $R$ be a minimal representative of an arbitrary process $\mathtt{F} \in \mathbf{Proc}$ and let $S \in R$. If $S$ has no internal elements that are not minimal in $S$, we are done. Thus, assume that there exists $x \in S$ internal and not minimal in $S$. Since $R$ is a minimal representative, due to \cref{2parents}, $x$ must have at least two unrelated elements in its past and since $\mathbf{Proc}$ is terminal, we can assume w.l.o.g. that $x$ is not maximal and thus also has at least two unrelated elements in its future. For simplicity, we will consider the case where $x$ has precisely two such elements in its past and future respectively but the following readily generalises to arbitrary numbers.

Let $i_\mathtt{F}$ be an implementation such that $S \cong \mathcal{G}(i_\mathtt{F})$. Let $\mathtt{g}$ be the process associated to the element $x$. By assumption, there exists a dilation $\mathtt{G}$ which is deterministic. We can then rewrite $\mathtt{g}$ as follows

\begin{gather}\label{deterministicreduction}
\begin{aligned}
\input{Diagrams/process2in2out.tikz} \quad &= \quad \input{Diagrams/bigGcopydisc.tikz} \\
&= \quad \input{Diagrams/copythenbigG.tikz}
\end{aligned}
\end{gather}
Rewriting $\mathtt{g}$ in this way, yields a new implementation $j_\mathtt{F} \in \mathcal{I}_\mathtt{F}$ which then also yields a new FPO as follows (depicting only the part of the FPO corresponding to the boxes in \cref{deterministicreduction})
\begin{equation}
\input{Diagrams/detred1.tikz} \quad \substack{\cref{deterministicreduction} \\ \rightarrow} \quad \input{Diagrams/detred2.tikz} \quad \sim \quad \input{Diagrams/detred3.tikz} \quad \sim \quad \input{Diagrams/detred4.tikz} \quad \succ \quad \input{Diagrams/detred1.tikz}.
\end{equation}
The first two diagrams are the FPO versions of \cref{deterministicreduction}. The second FPO is then equivalent under $\succ$ to the FPO $S \backslash \{x\}$, that is $S$ with $x$ removed, but any order relations left intact, which is the fourth FPO depicted above\footnote{To see this, note that the third FPO is a subset of the second and the second can be mapped in a FOP way by mapping the rightmost two elements of the second into either of the bottom two elements of the third. The third FPO is then equivalent to the fourth. Again, the fourth is a subset of the third and the third can be mapped in a FOP way into the four by mapping the depicted elements into the elements just outside the depicted area.}. 

Let now $R'$ be the minimal representative of $S \backslash \{x\}$ (note that this is not necessarily the FPO type of $S \backslash \{x\}$). As we have $S \backslash \{x\} \succ S$, we then have $R' \succ R$. As $S \backslash \{x\}$ corresponds to an implementation of $\mathtt{F}$, we have that $R' \in \mathcal{R}_\mathtt{F}$ by definition and since $\mathtt{F}$ was arbitrary we find that $R' \succ_\mathbf{Proc} R$ by definition. Further, $|R'| \leq |S \backslash \{x\}| = |S| - 1 = |R| - 1$ and thus $R' \neq R$ by \cref{minrep}. Hence, $R$ is not relevant.

Note that the above can be repeated until the resulting minimal representative has no internal elements left that are not internal (each step reduces the number of internal elements, hence, the algorithm must halt at some point). 

We show now that if all internal elements of $R'$ are minimal, then $R'$ cannot be quasi-relevant. This shows both that $R$ has to be irrelevant, since $R$ being quasi-relevant implies that $R' \succ_\mathbf{Proc} R$ is quasi-relevant and that there are no quasi-relevant minimal representatives.

We claim that if $R'$ is quasi-relevant, then there exists a sequence $\{R_n\}_{n \in \mathbb{N}}$ of minimal representatives such that all internal elements of $R_n$ are minimal and $R_{n+1} \succ_\mathbf{Proc} R_n$, $R_{n+1} \neq R_n$ for all $n$. However, this is absurd as the number of such minimal representatives is finite (for a given class \MINO). 

We prove the claim by induction. Set $R' = R_1$. Assume now we can find such $R_n$ for all $n \leq N$ for some $N \in \mathbb{N}$. Since $R_n$ is quasi-relevant by assumption, there exists $R'_{n+1} \succ_\mathbf{Proc} R_n \succ_\mathbf{Proc} R$ with $R'_{n+1} \neq R_n$. Additionally, as we have shown previously there exists $R_{n+1} \succ_\mathbf{Proc} R'_{n+1}$ such that $R_{n+1}$ has only minimal internal elements. This completes the induction step.

Now let us consider what happens if the theory is deterministic. If every process is deterministic in the first place, then states are also deterministic. Let $i_\mathtt{F}$ be an implementation such that $\mathcal{G}(i_\mathtt{F}) =: S$ has only minimal internal elements. Let $\rho$ be a state in $i_\mathtt{F}$ and use that it is deterministic to rewrite 
\begin{equation}
\input{Diagrams/state.tikz} \quad = \quad \input{Diagrams/statecopydisc.tikz} \quad = \quad \input{Diagrams/copiedstate.tikz},
\end{equation} 
which has the following effect on the FPO type
\begin{equation}
\input{Diagrams/statefpo.tikz} \quad \rightarrow \quad \input{Diagrams/statecopydiscfpo.tikz} \quad \rightarrow \quad \input{Diagrams/copiedstatefpo.tikz} \succ \input{Diagrams/statefpo.tikz},
\end{equation} 
Doing this with every state and combining them with their outputs, yields an implementation whose associated framed partial order has no internal elements. Showing that the original minimal representative is indeed irrelevant and that there are no quasi-relevant minimal representatives is analogous to what we did for the first part of the proof. 
\end{proof}

\classicalall*

\begin{proof}
Assume $R$ were irrelevant. Then, by definition, there exists a relevant minimal representative $R'$ with $R' \succ R$ such that every process that has $R$ as a minimal representative also has $R'$. By \cref{markovreduction}, we can assume that all internal elements of $R, R'$ are minimal.

We consider first the case $\MINOtype{0}{N}$. We choose labellings $S, S'$ of $R, R'$ as follows: in both cases, we denote the frame elements (which are all outputs) with $O_1,..., O_N$ and we denote an internal element $x_{\mathcal{K}}$ where $\mathcal{K} \subseteq \{1,...,N\}$ iff $x_\mathcal{K} < O_i$ iff $i \in \mathcal{K}$. The requirement that $R, R'$ are minimal representatives implies that this is well-defined (i.e., there is at most one element with this property for each set $\mathcal{K}$). Since $R \not \succ R'$, there must exist $x_\mathcal{K} \in S$ such that $x_{\mathcal{K}'} \not \in S'$ for any superset $\mathcal{K} \subseteq \mathcal{K}'$. W.l.o.g., we assume $\mathcal{K} = \{1,...,k\}$. Consider now the probability distribution
\begin{equation}
P(a_1,...,a_k) = \begin{cases} 1/2, &\text{ if } a_1 = a_2 = ... = a_k \\ 0, &\text{ else.} \end{cases}
\end{equation}
This is compatible with $S$ (identifying the variable $a_i$ with the outputs $O_i$ and treating the outputs $O_j$ for $j>k$ as trivial) as 
\begin{equation}
P(a_1,...,a_k) = \sum_\lambda \prod_{i \leq k} P(a_i |\lambda) P(\lambda)
\end{equation}
with $P(a_i|\lambda) = \delta_{a_i, \lambda}$ and $P(\lambda) = \frac{1}{2}$. Let us now show that there exists no implementation of this probability distribution compatible with $S'$. Since the outputs $O_j$ for $j>k$ are trivial, they correspond to tracing out and using that probabilistic classical theory is terminal, we find that if there exists an implementation compatible with $S'$ it can be written as
\begin{equation}\label{eq:classicalall}
P(a_1,...,a_k) = \sum_{\vec{\lambda}}\prod_{i \leq k} \prod_{\mathcal{K}: x_{\mathcal{K}} \in S'} P_i(a_i|\vec{\lambda}_i) P_{\mathcal{K}}(\lambda_{\mathcal{K}}) 
\end{equation}
where $\vec{\lambda}$ refers to all variables $\lambda_{\mathcal{K}}$ and $\vec{\lambda}_i$ refers to all variables $\lambda_{\mathcal{K}}$ with $i \in \mathcal{K}$. We assume w.l.o.g., that $P_{\mathcal{K}}(\lambda_{\mathcal{K}}) \neq 0$ for all $\lambda_{\mathcal{K}}$. If any if the variables $a_i$ differs from another one, we have that
\begin{equation}
0= P(a_1,...,a_k) = \sum_{\vec{\lambda}}\prod_{i \leq k} \prod_{\mathcal{K}: x_{\mathcal{K}} \in S'} P_i(a_i|\vec{\lambda}_i) P_{\mathcal{K}}(\lambda_{\mathcal{K}}) 
\end{equation}
which since all terms on the RHS are non-negative implies that for all $\vec{\lambda}$, 
\begin{equation}\label{eq:classicalall2}
0 = \prod_{i \leq k} P_i(a_i|\vec{\lambda}_i) \text{ if } \exists a_i \neq a_j.
\end{equation}
 Define now $a_i(\vec{\lambda}_i)$ for $i \neq 1$ such that $P_i(a_i(\vec{\lambda}_i|\vec{\lambda_i}) \neq 0$. We must then have $P_1(a_1|\vec{\lambda_1}) = \delta_{a_1, f_1(\vec{\lambda_1})}$ for some binary function $f_1$ in order to satisfy \cref{eq:classicalall2}. Due to symmetry, we must then also have $P(a_i|\vec{\lambda_i}) =  \delta_{a_i, f_i(\vec{\lambda_i})}$. Note that for a given choice of $\vec{\lambda}$, we must then have $f_1(\lambda_1) = f_2(\lambda_2) = ... = f_k(\lambda_k)$ to satisfy \cref{eq:classicalall2}. However, this means that changing the value of $\lambda_\mathcal{K}$ for any $\mathcal{K}$ cannot change the value of any of the $f_i$ as there always exists $j \not \in \mathcal{K}$ and so $f_j$ does not depend on the value of $\lambda_\mathcal{K}$. This means that all $f_i$ must be independent of $\vec{\lambda}$ but this is impossible. Hence, we have found that $P(a_1,...,a_k)$ is not compatible with $S'$, i.e., $R'$, which shows that $R$ is relevant.



Consider now the case \MINO with $M \neq 0$. We choose labellings $S$ for $R$ and $S'$ for $R'$ where the $i$-th element of the input is denoted $I_i$ and the $j$-th element of the output is denoted $O_j$. We first show that $I_i < O_j$ in $S$ iff $I_i < O_j$ in $S'$. Due to $S' \succ S$, we have $I_i < O_j$ in $S'$ implies $I_i < O_j$ in $S$ due to order-preservation. If now it holds $I_i < O_j$ in $S$, but not in $S'$, then the process that discards all inputs except for $I_i$ which it outputs at $O_j$ has $R$ as a minimal representative but not $R'$. 

Thus, $R, R'$ have to differ in their internal elements. We can then consider processes which discard the inputs and we have reduced the problem to the case of \MINOtype{0}{N}.
\end{proof}

In order to prove \cref{cliffcasc,zigzagcnot}, we first prove a lemma giving a sufficient condition for when a process in an arbitrary theory has an implementation for \textit{every} zigzag.

\begin{restatable}{lemma}{zigzagcascade}\label{zigzagcascade}
Let $\mathtt{F} \in \mathbf{Proc}$ be a \MINOtype{2}{2} process such that 
\begin{equation}\label{zigzagcasccond}
\input{Diagrams/zigzagcascade.tikz} \quad \in \quad \mathcal{I}_F \quad \text{or} \quad \input{Diagrams/zigzagcascadealt.tikz} \quad \in \quad \mathcal{I}_F
\end{equation}

is an implementation of $\mathtt{F}$ for some processes $\mathtt{A}, \mathtt{B}, \mathtt{C}, \rho \in \mathbf{Proc}$ or respectively $\mathtt{A}', \mathtt{B}', \mathtt{C}', \rho' \in \mathbf{Proc}$.

Then, $\mathtt{F}$ has an implementation for every zigzag minimal representative.
\end{restatable}

\begin{proof}
We prove by induction over all $N$-zigzags.

The $N=1$ case is true by assumption.

Assume now $\mathtt{F}$ has an implementation for the $N$-zigzag. Plug in this decomposition for the box $\mathtt{F}$ in \cref{zigzagcasccond}. This yields a new implementation whose minimal representative is the $N+1$-zigzag.
\end{proof}

Using this condition, we can prove \cref{cliffcasc}.

\cliffcasc*
\begin{proof}
We show that $U$ satisfies the antecedent of \cref{zigzagcascade}. For readability, we only explicitly show the case where all wires of $U$ correspond to qubits. The generalisation is straightforward.

Since $U$ is an element of the Clifford group, by definition the equality
\begin{equation}\label{cliff}
\input{Diagrams/cliffordproof0a.tikz} \quad = \quad \input{Diagrams/cliffordproof0b.tikz}
\end{equation}
must hold, for some $f, g$ which are functions from the set $\{0,1,2,3\}$ to itself and $\theta$ is a function from $\{0,1,2,3\}$ to $\{0, i\pi/2, i\pi, 3i\pi/2\}$. Note that in order for $U$ to satisfy unitarity, it must hold that if $f(i) = f(j)$ and $g(i) = g(j)$, then $i = j$.

Consider the diagram (note that sums are to be treated coherently and traces are applied after the sum)
\begin{equation}
\input{Diagrams/cliffordproof1.tikz}.
\end{equation}
Composing the red, green and blue shaded regions into a single CPTP map each\footnote{The red shaded region is an isometry. It corresponds to coherently outputting the state $\ket{jkjk}$ depending on which Bell state is on the two input wires. The green and blue shaded regions are CPTP maps as they correspond to measuring the system on the two wires on the right hand side and applying unitaries conditioned on the outcome}, the minimal representative is the 1-zigzag. Using that the composition of the cup and the cap is the identity, we then get
\begin{equation}
\input{Diagrams/cliffordproof2.tikz} \quad = \quad \input{Diagrams/cliffordproof3.tikz},
\end{equation}
where we used \cref{cliff} twice to obtain the equality. Note that the traced out state simply gives the scalar 1. Hence, $U$ satisfies the antecedent of \cref{zigzagcascade} and we are done.
\end{proof}

Next, we need a necessary condition for when a quantum process can be implemented with the two-way communication minimal representation.

\begin{restatable}{lemma}{evcond}\label{evcond}
Let $U: A \otimes B \rightarrow C \otimes D$ be a \MINOtype{2}{2} isometry in quantum theory. Let $p_{\psi, \phi}(\lambda)$ be the characteristic polynomial of the reduced state $\text{tr}_D U(\ket{\psi}^A \otimes \ket{\phi}^B)$. If $U$ has an implementation corresponding to the 2-way communication minimal representative, then the following relation holds for all $\ket{\psi}^A, \ket{\psi'}^A, \ket{\phi}^B, \ket{\phi'}^B$
\begin{equation}\label{eq:evcond}
p_{\psi, \phi}(\lambda)p_{\psi', \phi'}(\lambda) = p_{\psi', \phi}(\lambda) p_{\psi, \phi'}(\lambda).
\end{equation} 
\end{restatable}

\begin{proof}
Assume $U$ can be decomposed according to the two-way communication minimal representative
\begin{equation}\label{2waycomm}
\input{Diagrams/2waycomm.tikz}.
\end{equation}
W.l.o.g., we can assume that such a decomposition exists where $\mathtt{A}, \mathtt{B}$ are isometries. If they are not, let $V_\mathtt{A}, V_\mathtt{B}$ be isometries purifying $\mathtt{A}, \mathtt{B}$ and note that 
\begin{equation}\label{2waycommwlog}
\input{Diagrams/unitary.tikz} \quad = \quad  \input{Diagrams/2waycomm.tikz} \quad = \quad \input{Diagrams/2waycommwlog.tikz}.
\end{equation}
Hence, the RHS of the above equation is a decomposition of $U$ whose minimal representative is the two-way communication one (after composing the processes in the blue respectively green shaded regions into a single box).

Let now $V_\mathtt{C}, V_\mathtt{D}$ be purifications of the CPTP maps $\mathtt{C}, \mathtt{D}$. Then, it holds that
\begin{equation}\label{2waycommpur}
\input{Diagrams/2waycommpur.tikz} \quad = \quad \input{Diagrams/unitary.tikz} \quad \input{Diagrams/purifyingstate.tikz} \quad =: \quad \input{Diagrams/evconduprime.tikz}
\end{equation} 
up to ordering of the output wires.

The LHS is a purification of $U$. To see this, note that the processes $\mathtt{A}, \mathtt{B}, V_\mathtt{C}, V_{\mathtt{D}}$ are all isometries, hence, the overall diagram is an isometry. Tracing out the wires $C_E, D_E$ yields the original decomposition of $U$ from \cref{2waycomm} and so the LHS is a purification of $U$. A purification of an isometry must be product (this is because the RHS is also \textit{a} purification of $U$ and the purification of a channel is unique up to an isometry on the purifying systems). From this, we can derive the following equality for $\mathtt{A}$
\begin{equation}
\input{Diagrams/evcond1.tikz} \quad = \quad \input{Diagrams/evcond2.tikz} \quad = \quad \input{Diagrams/evcond3.tikz} \quad = \quad \input{Diagrams/evcond4.tikz}
\end{equation}
where $\phi'$ is an arbitrary pure state. Here, we used that $\mathtt{B}$ is CPTP in the first equality, then we added an identity process on the four outgoing wires before the trace in the second equality (the correctness of this equality can be easily checked by starting from the RHS to the LHS) and finally we used \cref{2waycommpur} in the last equality. An analogous equality holds for $\mathtt{B}$. Plugging these into the LHS of \cref{2waycommpur}, combining boxes so as to obtain the form of the 2-way communication minimal representative and then purifying the top two boxes, we find up to ordering of the output wires
\begin{equation}
\input{Diagrams/evcond5.tikz} \quad = \quad \input{Diagrams/unitary.tikz} \quad \input{Diagrams/purifyingstate.tikz}\quad \input{Diagrams/purifyingstate2.tikz}.
\end{equation}
Tracing out the systems $D, D_E, D'_E$ and using terminality, we find for all states $\ket{\psi}^A, \ket{\phi}^ B$
\begin{equation}
\input{Diagrams/evcond6.tikz} \quad = \quad \input{Diagrams/unitarymarginal.tikz} \quad \input{Diagrams/purifyingstatemarginal.tikz} \quad \input{Diagrams/purifyingstate2marginal.tikz}.
\end{equation}
Since $V'_{\mathtt{C}}$ is an isometry it does not change the characteristic polynomial of the state it acts on and we have that (using additionally that $U' = U \otimes \ket{\chi}$)
\begin{equation}
\input{Diagrams/evcond7.tikz} \quad \input{Diagrams/purifyingstatemarginal.tikz} \quad \input{Diagrams/purifyingstatemarginal.tikz} \quad \text{ and } \quad \input{Diagrams/unitarymarginal.tikz} \quad \input{Diagrams/purifyingstatemarginal.tikz} \quad \input{Diagrams/purifyingstate2marginal.tikz}.
\end{equation}
have the same characteristic polynomial. Hence, we have for all states $\ket{\psi}^A, \ket{\phi}^B$
\begin{equation}
p_{\psi, \phi}(\lambda)p_{\chi}(\lambda) p_{\xi}(\lambda) = p_{\psi', \phi}(\lambda) p_{\psi, \phi'}(\lambda)p_{\chi}(\lambda)  p_{\chi}(\lambda).
\end{equation}
where $p_{\chi}(\lambda), p_{\xi}(\lambda)$ refer to the characteristic polynomials of the reduced states of $\ket{\chi}, \ket{\xi}$. We can cancel one $p_{\chi}(\lambda)$ on each side, using polynomial division (note that $p_{\chi}(\lambda)$ cannot be the zero polynomial as it is the characteristic polynomial of a quantum state), obtaining
\begin{equation}\label{charpol}
p_{\psi, \phi}(\lambda)p_{\xi}(\lambda) = p_{\psi', \phi}(\lambda) p_{\psi, \phi'}(\lambda) p_{\chi}(\lambda).
\end{equation}
Note that this equality holds for all $\lambda$, even the zeros of $p_{\chi}(\lambda)$ due to the properties of polynomial division (alternatively, one can think of it as being due to continuity. The above equation holds for all $\lambda$ except the zeros of $p_{\chi}(\lambda)$ but since these are isolated points and polynomials are continuous functions, the equality holds at those points to). Plugging in $\psi'$ for $\psi$ and $\phi'$ for $\phi$, we obtain
\begin{equation}
p_{\psi', \phi'}(\lambda)p_{\xi}(\lambda) = p_{\psi', \phi'}(\lambda) p_{\psi', \phi'}(\lambda) p_{\chi}(\lambda) \Rightarrow p_{\xi}(\lambda) = p_{\psi', \phi'}(\lambda) p_{\chi}(\lambda).
\end{equation}
Plugging this back into \cref{charpol}, we obtain \cref{eq:evcond}. Note that while we initially fixed $\psi', \phi'$, only the states $\ket{\chi}, \ket{\xi}$ depended on this choice. Since \cref{eq:evcond} is independent of $\ket{\chi}, \ket{\xi}$, it must hold for all $\ket{\psi}^A, \ket{\psi'}^A, \ket{\phi}^B, \ket{\phi'}^B$.
\end{proof}

Combining \cref{cliffcasc} and \cref{evcond} yields \cref{zigzagcnot}
\zigzagcnot*
\begin{proof}
Since $\text{CNOT}$ belongs to the Clifford group every zigzag minimal representative is an implementation of $\text{CNOT}$ by \cref{cliffcasc}. We need to show that $\text{CNOT}$ cannot be implemented by the 2-way communication minimal representative.

Note that 
\begin{gather}
\begin{aligned}
\text{CNOT}(\ket{00}) &= \ket{00} \\
\text{CNOT}(\ket{0+}) &= \ket{0+} \\
\text{CNOT}(\ket{+0}) &= \frac{1}{\sqrt{2}} (\ket{00} + \ket{11}) \\
\text{CNOT}(\ket{++}) &= \frac{1}{2} (\ket{00} + \ket{01} + \ket{10} + \ket{11}).
\end{aligned}
\end{gather}
Three of these states are separable states, which is equivalent to their reduced density matrices having a single eigenvalue $1$. Hence, in order to satisfy \cref{evcond}, the last one, $\text{CNOT}(\ket{+0}) = \frac{1}{\sqrt{2}} (\ket{00} + \ket{11})$, should also be separable, but this is not the case, as it is the maximally entangled state. Thus, $\text{CNOT}$ cannot be decomposed using the 2-way communication minimal representative.

\end{proof} 
\end{document}

%% file: preamble.tex
\usepackage{rotating}
\usepackage{floatpag}
\rotfloatpagestyle{empty}

\usepackage{hyperref}
\usepackage{dsfont}
\usepackage{amsmath,amsthm,amssymb}
\usepackage{xspace,enumerate,epsfig}
\usepackage{graphicx}
\graphicspath{{.}{./figures/}}
\usepackage{marvosym}

\usepackage{tikzfig}
\usepackage{stmaryrd}
\usepackage{keycommand}

\input{defs.tex}

\input{tikzstyles.tex}

\let\olddagger\dagger
\renewcommand{\dagger}{\ensuremath{\olddagger}\xspace}

\usepackage{color}
\def\bR{\begin{color}{red}}
\def\bB{\begin{color}{blue}}
\def\bM{\begin{color}{magenta}}
\def\bC{\begin{color}{cyan}}
\def\bW{\begin{color}{white}}
\def\bBl{\begin{color}{black}}
\def\bG{\begin{color}{green}}
\def\bY{\begin{color}{yellow}}
\def\e{\end{color}\xspace}
\newcommand{\bit}{\begin{itemize}}
\newcommand{\eit}{\end{itemize}\par\noindent}
\newcommand{\ben}{\begin{enumerate}}
\newcommand{\een}{\end{enumerate}\par\noindent}
\newcommand{\beq}{\begin{equation}}
\newcommand{\eeq}{\end{equation}\par\noindent}
\newcommand{\beqa}{\begin{eqnarray*}}
\newcommand{\eeqa}{\end{eqnarray*}\par\noindent}
\newcommand{\beqn}{\begin{eqnarray}}
\newcommand{\eeqn}{\end{eqnarray}\par\noindent}

\def\jR{\begin{color}{DarkRed}}
\def\jB{\begin{color}{MidnightBlue}}
\def\jM{\begin{color}{magenta}}
\def\jC{\begin{color}{cyan}}
\def\jW{\begin{color}{white}}
\def\jBl{\begin{color}{black}}
\def\jG{\begin{color}{green}}
\def\jY{\begin{color}{yellow}}

\def\cR{\begin{color}{Crimson}}
\def\cB{\begin{color}{cyan}}
\def\cM{\begin{color}{magenta}}
\def\cC{\begin{color}{cyan}}
\def\cW{\begin{color}{white}}
\def\cBl{\begin{color}{black}}
\def\cG{\begin{color}{green}}
\def\cY{\begin{color}{yellow}}


%% file: defs.tex
\newkeycommand{\pointmap}[style=point][1]{\,\tikz{\node[style=\commandkey{style}] (x) at (0,-0.3) {$#1$};\node [style=none] (3) at (0, 1.0) {};\node [style=none] (3) at (0, -1.0) {};\draw (x)--(0,0.7);}\,}

\newkeycommand{\typedkpoint}[style=kpoint,edgestyle=][2]{%
\,\begin{tikzpicture}
  \begin{pgfonlayer}{nodelayer}
    \node [style=none] (0) at (0, 1) {};
    \node [style=\commandkey{style}] (1) at (0, -0.75) {$#2$};
    \node [style=right label] (2) at (0.25, 0.5) {$#1$};
  \end{pgfonlayer}
  \begin{pgfonlayer}{edgelayer}
    \draw [\commandkey{edgestyle}] (1) to (0.center);
  \end{pgfonlayer}
\end{tikzpicture}\,}

\newkeycommand{\typedkpointdag}[style=kpoint adjoint,edgestyle=][2]{%
\,\begin{tikzpicture}
  \begin{pgfonlayer}{nodelayer}
    \node [style=none] (0) at (0, -1) {};
    \node [style=\commandkey{style}] (1) at (0, 0.75) {$#2$};
    \node [style=right label] (2) at (0.25, -0.5) {$#1$};
  \end{pgfonlayer}
  \begin{pgfonlayer}{edgelayer}
    \draw [\commandkey{edgestyle}] (0.center) to (1);
  \end{pgfonlayer}
\end{tikzpicture}\,}

\newcommand{\bistateadj}[1]{\,\begin{tikzpicture}[yshift=-3mm]
  \begin{pgfonlayer}{nodelayer}
    \node [style=kpointadj, minimum width=1 cm, inner sep=2pt] (0) at (0, 0.5) {$\psi$};
    \node [style=none] (1) at (-0.75, 0.25) {};
    \node [style=none] (2) at (0.75, 0.25) {};
    \node [style=none] (3) at (-0.75, -0.5) {};
    \node [style=none] (4) at (0.75, -0.5) {};
  \end{pgfonlayer}
  \begin{pgfonlayer}{edgelayer}
    \draw (1.center) to (3.center);
    \draw (2.center) to (4.center);
  \end{pgfonlayer}
\end{tikzpicture}\,}

\newkeycommand{\pointketbra}[style1=copoint,style2=point][2]{\,%
\begin{tikzpicture}
  \begin{pgfonlayer}{nodelayer}
    \node [style=none] (0) at (0, -1.5) {};
    \node [style=\commandkey{style1}] (1) at (0, -0.75) {$#1$};
    \node [style=\commandkey{style2}] (2) at (0, 0.75) {$#2$};
    \node [style=none] (3) at (0, 1.5) {};
  \end{pgfonlayer}
  \begin{pgfonlayer}{edgelayer}
    \draw (0.center) to (1);
    \draw (2) to (3.center);
  \end{pgfonlayer}
\end{tikzpicture}\,}

\newkeycommand{\twocopointketbra}[style1=copoint,style2=copoint,style3=point][3]{\,%
\begin{tikzpicture}
  \begin{pgfonlayer}{nodelayer}
    \node [style=none] (0) at (-0.75, -1.5) {};
    \node [style=none] (1) at (0.75, -1.5) {};
    \node [style=\commandkey{style1}] (2) at (-0.75, -0.75) {$#1$};
    \node [style=\commandkey{style2}] (3) at (0.75, -0.75) {$#2$};
    \node [style=\commandkey{style3}] (4) at (0, 0.75) {$#3$};
    \node [style=none] (5) at (0, 1.5) {};
  \end{pgfonlayer}
  \begin{pgfonlayer}{edgelayer}
    \draw (0.center) to (2);
    \draw (1.center) to (3);
    \draw (4) to (5.center);
  \end{pgfonlayer}
\end{tikzpicture}\,}

\newkeycommand{\twopointketbra}[style1=copoint,style2=point,style3=point][3]{\,%
\begin{tikzpicture}
  \begin{pgfonlayer}{nodelayer}
    \node [style=none] (0) at (0, -1.5) {};
    \node [style=none] (1) at (-0.75, 1.5) {};
    \node [style=\commandkey{style1}] (2) at (0, -0.75) {$#1$};
    \node [style=\commandkey{style2}] (3) at (-0.75, 0.75) {$#2$};
    \node [style=\commandkey{style3}] (4) at (0.75, 0.75) {$#3$};
    \node [style=none] (5) at (0.75, 1.5) {};
  \end{pgfonlayer}
  \begin{pgfonlayer}{edgelayer}
    \draw (0.center) to (2);
    \draw (1.center) to (3);
    \draw (4) to (5.center);
  \end{pgfonlayer}
\end{tikzpicture}\,}

\newkeycommand{\pointbraket}[style1=point,style2=copoint][2]{\,%
\begin{tikzpicture}
  \begin{pgfonlayer}{nodelayer}
    \node [style=\commandkey{style1}] (0) at (0, -0.5) {$#1$};
    \node [style=\commandkey{style2}] (1) at (0, 0.5) {$#2$};
  \end{pgfonlayer}
  \begin{pgfonlayer}{edgelayer}
    \draw (0) to (1);
  \end{pgfonlayer}
\end{tikzpicture}\,}

\newkeycommand{\kpointbraket}[style1=kpoint,style2=kpoint adjoint][2]{\,%
\begin{tikzpicture}
  \begin{pgfonlayer}{nodelayer}
    \node [style=\commandkey{style1}] (0) at (0, -0.75) {$#1$};
    \node [style=\commandkey{style2}] (1) at (0, 0.75) {$#2$};
  \end{pgfonlayer}
  \begin{pgfonlayer}{edgelayer}
    \draw (0) to (1);
  \end{pgfonlayer}
\end{tikzpicture}\,}

\newkeycommand{\kpointmap}[style1=kpoint,style2=map][2]{\,%
\begin{tikzpicture}
  \begin{pgfonlayer}{nodelayer}
    \node [style=\commandkey{style1}] (0) at (0, -1.1) {$#1$};
    \node [style=\commandkey{style2}] (1) at (0, 0.2) {$#2$};
    \node [style=none] (2) at (0, 1.5) {};
  \end{pgfonlayer}
  \begin{pgfonlayer}{edgelayer}
    \draw (0) to (1);
    \draw (1) to (2.center);
  \end{pgfonlayer}
\end{tikzpicture}%
\,}

\newkeycommand{\sandwichtwo}[style1=point,style2=map,style3=map,style4=copoint][4]{\,%
\begin{tikzpicture}
  \begin{pgfonlayer}{nodelayer}
    \node [style=\commandkey{style2}] (0) at (0, -0.75) {$#2$};
    \node [style=\commandkey{style3}] (1) at (0, 0.75) {$#3$};
    \node [style=\commandkey{style1}] (2) at (0, -2) {$#1$};
    \node [style=\commandkey{style4}] (3) at (0, 2) {$#4$};
  \end{pgfonlayer}
  \begin{pgfonlayer}{edgelayer}
    \draw (2) to (0);
    \draw (0) to (1);
    \draw (1) to (3);
  \end{pgfonlayer}
\end{tikzpicture}\,}

\newkeycommand{\longbraket}[style1=point,style2=copoint][2]{\,%
\begin{tikzpicture}
  \begin{pgfonlayer}{nodelayer}
    \node [style=\commandkey{style1}] (0) at (0, -2) {$#1$};
    \node [style=\commandkey{style2}] (1) at (0, 2) {$#2$};
  \end{pgfonlayer}
  \begin{pgfonlayer}{edgelayer}
    \draw (0) to (1);
  \end{pgfonlayer}
\end{tikzpicture}\,}

\newkeycommand{\onb}[style=point]{\ensuremath{\left\{\tikz{\node[style=\commandkey{style}] (x) at (0,-0.3) {$j$};\draw (x)--(0,0.7);}\right\}}\xspace}

\newkeycommand{\copointmap}[style=copoint][1]{\,\tikz{\node[style=\commandkey{style}] (x) at (0,0.3) {$#1$};\draw (x)--(0,-0.7);}\,}

\tikzstyle{cdiag}=[matrix of math nodes, row sep=3em, column sep=3em, text height=1.5ex, text depth=0.25ex,inner sep=0.5em]
\tikzstyle{arrow above}=[transform canvas={yshift=0.5ex}]
\tikzstyle{arrow below}=[transform canvas={yshift=-0.5ex}]



%% file: tikzstyles.tex
\usetikzlibrary{shapes.multipart}

\tikzstyle{bwSpider}=[
       rectangle split,
       rectangle split parts=2,
       rectangle split part fill={black,white},
 minimum size=3.6 mm, inner sep=-2mm, draw=black,scale=0.5,rounded corners=0.8 mm
       ]
 \tikzstyle{wbSpider}=[
       rectangle split,
       rectangle split parts=2,
       rectangle split part fill={white,black},
 minimum size=3.6 mm, inner sep=-2mm, draw=black,scale=0.5,rounded corners=0.8 mm
       ]
\tikzstyle{qWire}=[line width = 1pt, color=black]
\tikzstyle{cWire}=[color=gray,thin]
\tikzstyle{env}=[copoint,regular polygon rotate=0,minimum width=0.2cm, fill=black]

\tikzstyle{probs}=[shape=semicircle,fill=white,draw=black,shape border rotate=180,minimum width=1.2cm]

\tikzstyle{every picture}=[baseline=-0.25em,scale=0.5]
\tikzstyle{dotpic}=[] 
\tikzstyle{diredges}=[every to/.style={diredge}]
\tikzstyle{math matrix}=[matrix of math nodes,left delimiter=(,right delimiter=),inner sep=2pt,column sep=1em,row sep=0.5em,nodes={inner sep=0pt},text height=1.5ex, text depth=0.25ex]

\tikzstyle{inline text}=[text height=1.5ex, text depth=0.25ex,yshift=0.5mm]
\tikzstyle{label}=[font=\footnotesize,text height=1.5ex, text depth=0.25ex,yshift=0.5mm]
\tikzstyle{left label}=[label,anchor=east,xshift=1.5mm]
\tikzstyle{right label}=[label,anchor=west,xshift=-1mm]

\tikzstyle{braceedge}=[decorate,decoration={brace,amplitude=2mm,raise=-1mm}]
\tikzstyle{small braceedge}=[decorate,decoration={brace,amplitude=1mm,raise=-1mm}]

\tikzstyle{doubled}=[line width=1.6pt] 
\tikzstyle{boldedge}=[doubled,shorten <=-0.17mm,shorten >=-0.17mm]
\tikzstyle{boldedgegray}=[doubled,gray,shorten <=-0.17mm,shorten >=-0.17mm]
\tikzstyle{singleedgegray}=[gray]

\tikzstyle{semidoubled}=[line width=1.4pt] 
\tikzstyle{semiboldedgegray}=[semidoubled,gray,shorten <=-0.17mm,shorten >=-0.17mm]

\tikzstyle{boxedge}=[semiboldedgegray]

\tikzstyle{boldedgedashed}=[very thick,dashed,shorten <=-0.17mm,shorten >=-0.17mm]
\tikzstyle{vboldedgedashed}=[doubled,dashed,shorten <=-0.17mm,shorten >=-0.17mm]
\tikzstyle{left hook arrow}=[left hook-latex]
\tikzstyle{right hook arrow}=[right hook-latex]
\tikzstyle{sembracket}=[line width=0.5pt,shorten <=-0.07mm,shorten >=-0.07mm]

\tikzstyle{causal edge}=[->,thick,gray]
\tikzstyle{causal nondir}=[thick,gray]
\tikzstyle{timeline}=[thick,gray, dashed]

\tikzstyle{cedge}=[<->,thick,gray!70!white]

\tikzstyle{empty diagram}=[draw=gray!40!white,dashed,shape=rectangle,minimum width=1cm,minimum height=1cm]
\tikzstyle{empty diagram small}=[draw=gray!50!white,dashed,shape=rectangle,minimum width=0.6cm,minimum height=0.5cm]

\tikzstyle{dot}=[inner sep=0mm,minimum width=2mm,minimum height=2mm,draw,shape=circle]
\tikzstyle{leak}=[white dot, shape=regular polygon, minimum size=3.3 mm, regular polygon sides=3, outer sep=-0.2mm, regular polygon rotate=270]
\tikzstyle{proj}=[draw,fill=white,chamfered rectangle,chamfered rectangle angle=30, minimum width=2mm, minimum height= 1mm,scale=0.5, outer sep=-0.2mm]
\tikzstyle{wide proj}=[draw,fill=white,chamfered rectangle,chamfered rectangle angle=30, minimum width=15mm,minimum height=1mm,scale=0.5, outer sep=-0.2mm]
\tikzstyle{very wide proj}=[draw,fill=white,chamfered rectangle,chamfered rectangle angle=30, minimum width=25mm,minimum height=1mm,scale=0.5, outer sep=-0.2mm]
\tikzstyle{very very wide proj}=[draw,fill=white,chamfered rectangle,chamfered rectangle angle=30, minimum width=35mm,minimum height=1mm,scale=0.5, outer sep=-0.2mm]
\tikzstyle{preleak}=[proj]

\tikzstyle{split proj out}=[regular polygon,regular polygon sides=3,draw,scale=0.75,inner sep=-0.5pt,minimum width=3.3mm,fill=white,regular polygon rotate=180]
\tikzstyle{split proj in}=[regular polygon,regular polygon sides=3,draw,scale=0.75,inner sep=-0.5pt,minimum width=3.3mm,fill=white]
\tikzstyle{Vleak}=[white dot, shape=regular polygon, minimum size=3.3 mm, regular polygon sides=3, outer sep=-0.2mm, regular polygon rotate=90]
\tikzstyle{dleak}=[white dot, line width=1.6pt, shape=regular polygon, minimum size=3.3 mm, regular polygon sides=3, outer sep=-0.2mm, regular polygon rotate=270]

\tikzstyle{Wsquare}=[white dot, shape=regular polygon, rounded corners=0.8 mm, minimum size=3.3 mm, regular polygon sides=3, outer sep=-0.2mm]
\tikzstyle{Wsquareadj}=[white dot, shape=regular polygon, rounded corners=0.8 mm, minimum size=3.3 mm, regular polygon sides=3, outer sep=-0.2mm, regular polygon rotate=180]
\tikzstyle{ddot}=[inner sep=0mm, doubled, minimum width=2.5mm,minimum height=2.5mm,draw,shape=circle]

\tikzstyle{black dot}=[dot,fill=black]
\tikzstyle{white dot}=[dot,fill=white,,text depth=-0.2mm]
\tikzstyle{white Wsquare}=[Wsquare,fill=gray,,text depth=-0.2mm]
\tikzstyle{white Wsquareadj}=[Wsquareadj,fill=white,,text depth=-0.2mm]
\tikzstyle{green dot}=[white dot] 
\tikzstyle{gray dot}=[dot,fill=gray!40!white,,text depth=-0.2mm]
\tikzstyle{red dot}=[gray dot] 

\tikzstyle{black ddot}=[ddot,fill=black]
\tikzstyle{white ddot}=[ddot,fill=white]
\tikzstyle{gray ddot}=[ddot,fill=gray!40!white]

\tikzstyle{gray edge}=[gray!60!white]

\tikzstyle{small dot}=[inner sep=0.5mm,minimum width=0pt,minimum height=0pt,draw,shape=circle]

\tikzstyle{small black dot}=[small dot,fill=black]
\tikzstyle{small white dot}=[small dot,fill=white]
\tikzstyle{small gray dot}=[small dot,fill=gray!40!white]

\tikzstyle{causal dot}=[inner sep=0.4mm,minimum width=0pt,minimum height=0pt,draw=white,shape=circle,fill=gray!40!white]

\tikzstyle{phase dimensions}=[minimum size=5mm,font=\footnotesize,rectangle,rounded corners=2.5mm,inner sep=0.2mm,outer sep=-2mm]
\tikzstyle{dphase dimensions}=[minimum size=5mm,font=\footnotesize,rectangle,rounded corners=2.5mm,inner sep=0.2mm,outer sep=-2mm]

\tikzstyle{white phase dot}=[dot,fill=white,phase dimensions]
\tikzstyle{white phase ddot}=[ddot,fill=white,dphase dimensions]

\tikzstyle{white rect ddot}=[draw=black,fill=white,doubled,minimum size=5mm,font=\footnotesize,rectangle,rounded corners=2.5mm,inner sep=0.2mm]
\tikzstyle{gray rect ddot}=[draw=black,fill=gray!40!white,doubled,minimum size=6mm,font=\footnotesize,rectangle,rounded corners=3mm]

\tikzstyle{gray phase dot}=[dot,fill=gray!40!white,phase dimensions]
\tikzstyle{gray phase ddot}=[ddot,fill=gray!40!white,dphase dimensions]
\tikzstyle{grey phase dot}=[gray phase dot]
\tikzstyle{grey phase ddot}=[gray phase ddot]

\tikzstyle{small phase dimensions}=[minimum size=4mm,font=\tiny,rectangle,rounded corners=2mm,inner sep=0.2mm,outer sep=-2mm]
\tikzstyle{small dphase dimensions}=[minimum size=4mm,font=\tiny,rectangle,rounded corners=2mm,inner sep=0.2mm,outer sep=-2mm]

\tikzstyle{small gray phase dot}=[dot,fill=gray!40!white,small phase dimensions]
\tikzstyle{small gray phase ddot}=[ddot,fill=gray!40!white,small dphase dimensions]

\tikzstyle{small map}=[draw,shape=rectangle,minimum height=4mm,minimum width=4mm,fill=white]

\tikzstyle{cnot}=[fill=white,shape=circle,inner sep=-1.4pt]

\tikzstyle{asym hadamard}=[fill=white,draw,shape=NEbox,inner sep=0.6mm,font=\footnotesize,minimum height=4mm]
\tikzstyle{asym hadamard conj}=[fill=white,draw,shape=NWbox,inner sep=0.6mm,font=\footnotesize,minimum height=4mm]
\tikzstyle{asym hadamard dag}=[fill=white,draw,shape=SEbox,inner sep=0.6mm,font=\footnotesize,minimum height=4mm]

\tikzstyle{hadamard}=[fill=white,draw,inner sep=0.6mm,font=\footnotesize,minimum height=4mm,minimum width=4mm]
\tikzstyle{small hadamard}=[fill=white,draw,inner sep=0.6mm,minimum height=1.5mm,minimum width=1.5mm]
\tikzstyle{small hadamard rotate}=[small hadamard,rotate=45]
\tikzstyle{dhadamard}=[hadamard,doubled]
\tikzstyle{small dhadamard}=[small hadamard,doubled]
\tikzstyle{small dhadamard rotate}=[small hadamard rotate,doubled]
\tikzstyle{antipode}=[white dot,inner sep=0.3mm,font=\footnotesize]

\tikzstyle{scalar}=[diamond,draw,inner sep=0.5pt,font=\small]
\tikzstyle{dscalar}=[diamond,doubled, draw,inner sep=0.5pt,font=\small]

\tikzstyle{small box}=[rectangle,inline text,fill=white,draw,minimum height=5mm,yshift=-0.5mm,minimum width=5mm,font=\small]
\tikzstyle{small gray box}=[small box,fill=gray!30]
\tikzstyle{medium box}=[rectangle,inline text,fill=white,draw,minimum height=5mm,yshift=-0.5mm,minimum width=10mm,font=\small]
\tikzstyle{square box}=[small box] 
\tikzstyle{medium gray box}=[small box,fill=gray!30]
\tikzstyle{semilarge box}=[rectangle,inline text,fill=white,draw,minimum height=5mm,yshift=-0.5mm,minimum width=12.5mm,font=\small]
\tikzstyle{large box}=[rectangle,inline text,fill=white,draw,minimum height=5mm,yshift=-0.5mm,minimum width=15mm,font=\small]
\tikzstyle{large gray box}=[small box,fill=gray!30]

\tikzstyle{Bayes box}=[rectangle,fill=black,draw, minimum height=3mm, minimum width=3mm]

\tikzstyle{gray square point}=[small box,fill=gray!50]

\tikzstyle{dphase box white}=[dhadamard]
\tikzstyle{dphase box gray}=[dhadamard,fill=gray!50!white]
\tikzstyle{phase box white}=[hadamard]
\tikzstyle{phase box gray}=[hadamard,fill=gray!50!white]

\tikzstyle{point}=[regular polygon,regular polygon sides=3,draw,scale=0.75,inner sep=-0.5pt,minimum width=9mm,fill=white,regular polygon rotate=180]
\tikzstyle{point nosep}=[regular polygon,regular polygon sides=3,draw,scale=0.75,inner sep=-2pt,minimum width=9mm,fill=white,regular polygon rotate=180]
\tikzstyle{copoint}=[regular polygon,regular polygon sides=3,draw,scale=0.75,inner sep=-0.5pt,minimum width=9mm,fill=white]
\tikzstyle{dpoint}=[point,doubled]
\tikzstyle{dcopoint}=[copoint,doubled]

\tikzstyle{pointgrow}=[shape=cornerpoint,kpoint common,scale=0.75,inner sep=3pt]
\tikzstyle{pointgrow dag}=[shape=cornercopoint,kpoint common,scale=0.75,inner sep=3pt]

\tikzstyle{wide copoint}=[fill=white,draw,shape=isosceles triangle,shape border rotate=90,isosceles triangle stretches=true,inner sep=0pt,minimum width=1.5cm,minimum height=6.12mm]
\tikzstyle{wide point}=[fill=white,draw,shape=isosceles triangle,shape border rotate=-90,isosceles triangle stretches=true,inner sep=0pt,minimum width=1.5cm,minimum height=6.12mm,yshift=-0.0mm]
\tikzstyle{wide point plus}=[fill=white,draw,shape=isosceles triangle,shape border rotate=-90,isosceles triangle stretches=true,inner sep=0pt,minimum width=1.74cm,minimum height=7mm,yshift=-0.0mm]

\tikzstyle{wide dpoint}=[fill=white,doubled,draw,shape=isosceles triangle,shape border rotate=-90,isosceles triangle stretches=true,inner sep=0pt,minimum width=1.5cm,minimum height=6.12mm,yshift=-0.0mm]

\tikzstyle{tinypoint}=[regular polygon,regular polygon sides=3,draw,scale=0.55,inner sep=-0.15pt,minimum width=6mm,fill=white,regular polygon rotate=180]

\tikzstyle{white point}=[point]
\tikzstyle{white dpoint}=[dpoint]
\tikzstyle{green point}=[white point] 
\tikzstyle{white copoint}=[copoint]
\tikzstyle{gray point}=[point,fill=gray!40!white]
\tikzstyle{gray dpoint}=[gray point,doubled]
\tikzstyle{red point}=[gray point] 
\tikzstyle{gray copoint}=[copoint,fill=gray!40!white]
\tikzstyle{gray dcopoint}=[gray copoint,doubled]

\tikzstyle{white point guide}=[regular polygon,regular polygon sides=3,font=\scriptsize,draw,scale=0.65,inner sep=-0.5pt,minimum width=9mm,fill=white,regular polygon rotate=180]

\tikzstyle{black point}=[point,fill=black,font=\color{white}]
\tikzstyle{black copoint}=[copoint,fill=black,font=\color{white}]

\tikzstyle{tiny gray point}=[tinypoint,fill=gray!40!white]

\tikzstyle{diredge}=[->]
\tikzstyle{ddiredge}=[<->]
\tikzstyle{rdiredge}=[<-]
\tikzstyle{thickdiredge}=[->, very thick]
\tikzstyle{pointer edge}=[->,very thick,gray]
\tikzstyle{pointer edge part}=[very thick,gray]
\tikzstyle{dashed edge}=[dashed]
\tikzstyle{thick dashed edge}=[very thick,dashed]
\tikzstyle{thick gray dashed edge}=[thick dashed edge,gray!40]
\tikzstyle{thick map edge}=[very thick,|->]

\makeatletter
\newcommand{\boxshape}[3]{%
\pgfdeclareshape{#1}{
\inheritsavedanchors[from=rectangle] 
\inheritanchorborder[from=rectangle]
\inheritanchor[from=rectangle]{center}
\inheritanchor[from=rectangle]{north}
\inheritanchor[from=rectangle]{south}
\inheritanchor[from=rectangle]{west}
\inheritanchor[from=rectangle]{east}
\backgroundpath{
\southwest \pgf@xa=\pgf@x \pgf@ya=\pgf@y
\northeast \pgf@xb=\pgf@x \pgf@yb=\pgf@y

\@tempdima=#2
\@tempdimb=#3

\pgfpathmoveto{\pgfpoint{\pgf@xa - 5pt + \@tempdima}{\pgf@ya}}
\pgfpathlineto{\pgfpoint{\pgf@xa - 5pt - \@tempdima}{\pgf@yb}}
\pgfpathlineto{\pgfpoint{\pgf@xb + 5pt + \@tempdimb}{\pgf@yb}}
\pgfpathlineto{\pgfpoint{\pgf@xb + 5pt - \@tempdimb}{\pgf@ya}}
\pgfpathlineto{\pgfpoint{\pgf@xa - 5pt + \@tempdima}{\pgf@ya}}
\pgfpathclose
}
}}

\boxshape{NEbox}{0pt}{3pt}
\boxshape{SEbox}{0pt}{-3pt}
\boxshape{NWbox}{5pt}{0pt}
\boxshape{SWbox}{-5pt}{0pt}
\boxshape{EBox}{-3pt}{3pt}
\boxshape{WBox}{3pt}{-3pt}
\makeatother

\tikzstyle{cloud}=[shape=cloud,draw,minimum width=1.5cm,minimum height=1.5cm]

\tikzstyle{map}=[draw,shape=NEbox,inner sep=1pt,minimum height=4mm,fill=white]
\tikzstyle{dashedmap}=[draw,dashed,shape=NEbox,inner sep=2pt,minimum height=6mm,fill=white]
\tikzstyle{mapdag}=[draw,shape=SEbox,inner sep=1pt,minimum height=4mm,fill=white]
\tikzstyle{mapadj}=[draw,shape=SEbox,inner sep=2pt,minimum height=6mm,fill=white]
\tikzstyle{maptrans}=[draw,shape=SWbox,inner sep=2pt,minimum height=6mm,fill=white]
\tikzstyle{mapconj}=[draw,shape=NWbox,inner sep=2pt,minimum height=6mm,fill=white]

\tikzstyle{medium map}=[draw,shape=NEbox,inner sep=2pt,minimum height=6mm,fill=white,minimum width=7mm]
\tikzstyle{medium map dag}=[draw,shape=SEbox,inner sep=2pt,minimum height=6mm,fill=white,minimum width=7mm]
\tikzstyle{medium map adj}=[draw,shape=SEbox,inner sep=2pt,minimum height=6mm,fill=white,minimum width=7mm]
\tikzstyle{medium map trans}=[draw,shape=SWbox,inner sep=2pt,minimum height=6mm,fill=white,minimum width=7mm]
\tikzstyle{medium map conj}=[draw,shape=NWbox,inner sep=2pt,minimum height=6mm,fill=white,minimum width=7mm]
\tikzstyle{semilarge map}=[draw,shape=NEbox,inner sep=2pt,minimum height=6mm,fill=white,minimum width=9.5mm]
\tikzstyle{semilarge map trans}=[draw,shape=SWbox,inner sep=2pt,minimum height=6mm,fill=white,minimum width=9.5mm]
\tikzstyle{semilarge map adj}=[draw,shape=SEbox,inner sep=2pt,minimum height=6mm,fill=white,minimum width=9.5mm]
\tikzstyle{semilarge map dag}=[draw,shape=SEbox,inner sep=2pt,minimum height=6mm,fill=white,minimum width=9.5mm]
\tikzstyle{semilarge map conj}=[draw,shape=NWbox,inner sep=2pt,minimum height=6mm,fill=white,minimum width=9.5mm]
\tikzstyle{large map}=[draw,shape=NEbox,inner sep=2pt,minimum height=6mm,fill=white,minimum width=12mm]
\tikzstyle{large map conj}=[draw,shape=NWbox,inner sep=2pt,minimum height=6mm,fill=white,minimum width=12mm]
\tikzstyle{very large map}=[draw,shape=NEbox,inner sep=2pt,minimum height=6mm,fill=white,minimum width=17mm]

\tikzstyle{medium dmap}=[draw,doubled,shape=NEbox,inner sep=2pt,minimum height=6mm,fill=white,minimum width=7mm]
\tikzstyle{medium dmap dag}=[draw,doubled,shape=SEbox,inner sep=2pt,minimum height=6mm,fill=white,minimum width=7mm]
\tikzstyle{medium dmap adj}=[draw,doubled,shape=SEbox,inner sep=2pt,minimum height=6mm,fill=white,minimum width=7mm]
\tikzstyle{medium dmap trans}=[draw,doubled,shape=SWbox,inner sep=2pt,minimum height=6mm,fill=white,minimum width=7mm]
\tikzstyle{medium dmap conj}=[draw,doubled,shape=NWbox,inner sep=2pt,minimum height=6mm,fill=white,minimum width=7mm]
\tikzstyle{semilarge dmap}=[draw,doubled,shape=NEbox,inner sep=2pt,minimum height=6mm,fill=white,minimum width=9.5mm]
\tikzstyle{semilarge dmap trans}=[draw,doubled,shape=SWbox,inner sep=2pt,minimum height=6mm,fill=white,minimum width=9.5mm]
\tikzstyle{semilarge dmap adj}=[draw,doubled,shape=SEbox,inner sep=2pt,minimum height=6mm,fill=white,minimum width=9.5mm]
\tikzstyle{semilarge dmap dag}=[draw,doubled,shape=SEbox,inner sep=2pt,minimum height=6mm,fill=white,minimum width=9.5mm]
\tikzstyle{semilarge dmap conj}=[draw,doubled,shape=NWbox,inner sep=2pt,minimum height=6mm,fill=white,minimum width=9.5mm]
\tikzstyle{large dmap}=[draw,doubled,shape=NEbox,inner sep=2pt,minimum height=6mm,fill=white,minimum width=12mm]
\tikzstyle{large dmap conj}=[draw,doubled,shape=NWbox,inner sep=2pt,minimum height=6mm,fill=white,minimum width=12mm]
\tikzstyle{large dmap trans}=[draw,doubled,shape=SWbox,inner sep=2pt,minimum height=6mm,fill=white,minimum width=12mm]
\tikzstyle{large dmap adj}=[draw,doubled,shape=SEbox,inner sep=2pt,minimum height=6mm,fill=white,minimum width=12mm]
\tikzstyle{large dmap dag}=[draw,doubled,shape=SEbox,inner sep=2pt,minimum height=6mm,fill=white,minimum width=12mm]
\tikzstyle{very large dmap}=[draw,doubled,shape=NEbox,inner sep=2pt,minimum height=6mm,fill=white,minimum width=19.5mm]

\tikzstyle{muxbox}=[draw,shape=rectangle,minimum height=3mm,minimum width=3mm,fill=white]
\tikzstyle{dmuxbox}=[muxbox,doubled]

\tikzstyle{box}=[draw,shape=rectangle,inner sep=2pt,minimum height=6mm,minimum width=6mm,fill=white]
\tikzstyle{dbox}=[draw,doubled,shape=rectangle,inner sep=2pt,minimum height=6mm,minimum width=6mm,fill=white]
\tikzstyle{dmap}=[draw,doubled,shape=NEbox,inner sep=2pt,minimum height=6mm,fill=white]
\tikzstyle{dmapdag}=[draw,doubled,shape=SEbox,inner sep=2pt,minimum height=6mm,fill=white]
\tikzstyle{dmapadj}=[draw,doubled,shape=SEbox,inner sep=2pt,minimum height=6mm,fill=white]
\tikzstyle{dmaptrans}=[draw,doubled,shape=SWbox,inner sep=2pt,minimum height=6mm,fill=white]
\tikzstyle{dmapconj}=[draw,doubled,shape=NWbox,inner sep=2pt,minimum height=6mm,fill=white]

\tikzstyle{ddmap}=[draw,doubled,dashed,shape=NEbox,inner sep=2pt,minimum height=6mm,fill=white]
\tikzstyle{ddmapdag}=[draw,doubled,dashed,shape=SEbox,inner sep=2pt,minimum height=6mm,fill=white]
\tikzstyle{ddmapadj}=[draw,doubled,dashed,shape=SEbox,inner sep=2pt,minimum height=6mm,fill=white]
\tikzstyle{ddmaptrans}=[draw,doubled,dashed,shape=SWbox,inner sep=2pt,minimum height=6mm,fill=white]
\tikzstyle{ddmapconj}=[draw,doubled,dashed,shape=NWbox,inner sep=2pt,minimum height=6mm,fill=white]

\boxshape{sNEbox}{0pt}{3pt}
\boxshape{sSEbox}{0pt}{-3pt}
\boxshape{sNWbox}{3pt}{0pt}
\boxshape{sSWbox}{-3pt}{0pt}
\tikzstyle{smap}=[draw,shape=sNEbox,fill=white]
\tikzstyle{smapdag}=[draw,shape=sSEbox,fill=white]
\tikzstyle{smapadj}=[draw,shape=sSEbox,fill=white]
\tikzstyle{smaptrans}=[draw,shape=sSWbox,fill=white]
\tikzstyle{smapconj}=[draw,shape=sNWbox,fill=white]

\tikzstyle{dsmap}=[draw,dashed,shape=sNEbox,fill=white]
\tikzstyle{dsmapdag}=[draw,dashed,shape=sSEbox,fill=white]
\tikzstyle{dsmaptrans}=[draw,dashed,shape=sSWbox,fill=white]
\tikzstyle{dsmapconj}=[draw,dashed,shape=sNWbox,fill=white]

\boxshape{mNEbox}{0pt}{10pt}
\boxshape{mSEbox}{0pt}{-10pt}
\boxshape{mNWbox}{10pt}{0pt}
\boxshape{mSWbox}{-10pt}{0pt}
\tikzstyle{mmap}=[draw,shape=mNEbox]
\tikzstyle{mmapdag}=[draw,shape=mSEbox]
\tikzstyle{mmaptrans}=[draw,shape=mSWbox]
\tikzstyle{mmapconj}=[draw,shape=mNWbox]

\tikzstyle{mmapgray}=[draw,fill=gray!40!white,shape=mNEbox]
\tikzstyle{smapgray}=[draw,fill=gray!40!white,shape=sNEbox]

\makeatletter

\pgfdeclareshape{cornerpoint}{
\inheritsavedanchors[from=rectangle] 
\inheritanchorborder[from=rectangle]
\inheritanchor[from=rectangle]{center}
\inheritanchor[from=rectangle]{north}
\inheritanchor[from=rectangle]{south}
\inheritanchor[from=rectangle]{west}
\inheritanchor[from=rectangle]{east}
\backgroundpath{
\southwest \pgf@xa=\pgf@x \pgf@ya=\pgf@y
\northeast \pgf@xb=\pgf@x \pgf@yb=\pgf@y

\pgfmathsetmacro{\pgf@shorten@left}{\pgfkeysvalueof{/tikz/shorten left}}
\pgfmathsetmacro{\pgf@shorten@right}{\pgfkeysvalueof{/tikz/shorten right}}

\pgfpathmoveto{\pgfpoint{0.5 * (\pgf@xa + \pgf@xb)}{\pgf@ya - 5pt}}
\pgfpathlineto{\pgfpoint{\pgf@xa - 8pt + \pgf@shorten@left}{\pgf@yb - 1.5 * \pgf@shorten@left}}
\pgfpathlineto{\pgfpoint{\pgf@xa - 8pt + \pgf@shorten@left}{\pgf@yb}}
\pgfpathlineto{\pgfpoint{\pgf@xb + 8pt - \pgf@shorten@right}{\pgf@yb}}
\pgfpathlineto{\pgfpoint{\pgf@xb + 8pt - \pgf@shorten@right}{\pgf@yb - 1.5 * \pgf@shorten@right}}
\pgfpathclose
}
}

\pgfdeclareshape{cornercopoint}{
\inheritsavedanchors[from=rectangle] 
\inheritanchorborder[from=rectangle]
\inheritanchor[from=rectangle]{center}
\inheritanchor[from=rectangle]{north}
\inheritanchor[from=rectangle]{south}
\inheritanchor[from=rectangle]{west}
\inheritanchor[from=rectangle]{east}
\backgroundpath{
\southwest \pgf@xa=\pgf@x \pgf@ya=\pgf@y
\northeast \pgf@xb=\pgf@x \pgf@yb=\pgf@y

\pgfmathsetmacro{\pgf@shorten@left}{\pgfkeysvalueof{/tikz/shorten left}}
\pgfmathsetmacro{\pgf@shorten@right}{\pgfkeysvalueof{/tikz/shorten right}}

\pgfpathmoveto{\pgfpoint{0.5 * (\pgf@xa + \pgf@xb)}{\pgf@yb + 5pt}}
\pgfpathlineto{\pgfpoint{\pgf@xa - 8pt + \pgf@shorten@left}{\pgf@ya + 1.5 * \pgf@shorten@left}}
\pgfpathlineto{\pgfpoint{\pgf@xa - 8pt + \pgf@shorten@left}{\pgf@ya}}
\pgfpathlineto{\pgfpoint{\pgf@xb + 8pt - \pgf@shorten@right}{\pgf@ya}}
\pgfpathlineto{\pgfpoint{\pgf@xb + 8pt - \pgf@shorten@right}{\pgf@ya + 1.5 * \pgf@shorten@right}}
\pgfpathclose
}
}

\makeatother

\pgfkeyssetvalue{/tikz/shorten left}{0pt}
\pgfkeyssetvalue{/tikz/shorten right}{0pt}

\tikzstyle{kpoint common}=[draw,fill=white,inner sep=1pt,minimum height=4mm]
\tikzstyle{kpoint sc}=[shape=cornerpoint,kpoint common]
\tikzstyle{kpoint adjoint sc}=[shape=cornercopoint,kpoint common]
\tikzstyle{kpoint}=[shape=cornerpoint,shorten left=5pt,kpoint common]
\tikzstyle{kpoint adjoint}=[shape=cornercopoint,shorten left=5pt,kpoint common]
\tikzstyle{kpoint conjugate}=[shape=cornerpoint,shorten right=5pt,kpoint common]
\tikzstyle{kpoint transpose}=[shape=cornercopoint,shorten right=5pt,kpoint common]
\tikzstyle{kpoint symm}=[shape=cornerpoint,shorten left=5pt,shorten right=5pt,kpoint common]

\tikzstyle{wide kpoint sc}=[shape=cornerpoint,kpoint common, minimum width=1 cm]
\tikzstyle{wide kpointdag sc}=[shape=cornercopoint,kpoint common, minimum width=1 cm]

\tikzstyle{black kpoint}=[shape=cornerpoint,shorten left=5pt,kpoint common,fill=black,font=\color{white}]

\tikzstyle{black kpoint sm}=[shape=cornerpoint,shorten left=5pt,kpoint common,fill=black,font=\color{white},scale=0.75]

\tikzstyle{black kpoint adjoint}=[shape=cornercopoint,shorten left=5pt,kpoint common,fill=black,font=\color{white}]
\tikzstyle{black kpointadj}=[shape=cornercopoint,shorten left=5pt,kpoint common,fill=black,font=\color{white}]

\tikzstyle{black kpointadj sm}=[shape=cornercopoint,shorten left=5pt,kpoint common,fill=black,font=\color{white},scale=0.75]

\tikzstyle{black dkpoint}=[shape=cornerpoint,shorten left=5pt,kpoint common,fill=black, doubled,font=\color{white}]
\tikzstyle{black dkpoint adjoint}=[shape=cornercopoint,shorten left=5pt,kpoint common,fill=black, doubled,font=\color{white}]
\tikzstyle{black dkpointadj}=[shape=cornercopoint,shorten left=5pt,kpoint common,fill=black, doubled,font=\color{white}]

\tikzstyle{black dkpoint sm}=[shape=cornerpoint,shorten left=5pt,kpoint common,fill=black, doubled,font=\color{white},scale=0.75]
\tikzstyle{black dkpointadj sm}=[shape=cornercopoint,shorten left=5pt,kpoint common,fill=black, doubled,font=\color{white},scale=0.75]

\tikzstyle{kpointdag}=[kpoint adjoint]
\tikzstyle{kpointadj}=[kpoint adjoint]
\tikzstyle{kpointconj}=[kpoint conjugate]
\tikzstyle{kpointtrans}=[kpoint transpose]

\tikzstyle{big kpoint}=[kpoint, minimum width=1.2 cm, minimum height=8mm, inner sep=4pt, text depth=3mm]

\tikzstyle{wide kpoint}=[kpoint, minimum width=1 cm, inner sep=2pt]
\tikzstyle{wide kpointdag}=[kpointdag, minimum width=1 cm, inner sep=2pt]
\tikzstyle{wide kpointconj}=[kpointconj, minimum width=1 cm, inner sep=2pt]
\tikzstyle{wide kpointtrans}=[kpointtrans, minimum width=1 cm, inner sep=2pt]

\tikzstyle{wider kpoint}=[kpoint, minimum width=1.25 cm, inner sep=2pt]
\tikzstyle{wider kpointdag}=[kpointdag, minimum width=1.25 cm, inner sep=2pt]
\tikzstyle{wider kpointconj}=[kpointconj, minimum width=1.25 cm, inner sep=2pt]
\tikzstyle{wider kpointtrans}=[kpointtrans, minimum width=1.25 cm, inner sep=2pt]

\tikzstyle{gray kpoint}=[kpoint,fill=gray!50!white]
\tikzstyle{gray kpointdag}=[kpointdag,fill=gray!50!white]
\tikzstyle{gray kpointadj}=[kpointadj,fill=gray!50!white]
\tikzstyle{gray kpointconj}=[kpointconj,fill=gray!50!white]
\tikzstyle{gray kpointtrans}=[kpointtrans,fill=gray!50!white]

\tikzstyle{gray dkpoint}=[kpoint,fill=gray!50!white,doubled]
\tikzstyle{gray dkpointdag}=[kpointdag,fill=gray!50!white,doubled]
\tikzstyle{gray dkpointadj}=[kpointadj,fill=gray!50!white,doubled]
\tikzstyle{gray dkpointconj}=[kpointconj,fill=gray!50!white,doubled]
\tikzstyle{gray dkpointtrans}=[kpointtrans,fill=gray!50!white,doubled]

\tikzstyle{white label}=[draw,fill=white,rectangle,inner sep=0.7 mm]
\tikzstyle{gray label}=[draw,fill=gray!50!white,rectangle,inner sep=0.7 mm]
\tikzstyle{black label}=[draw,fill=black,rectangle,inner sep=0.7 mm]

\tikzstyle{dkpoint}=[kpoint,doubled]
\tikzstyle{wide dkpoint}=[wide kpoint,doubled]
\tikzstyle{dkpointdag}=[kpoint adjoint,doubled]
\tikzstyle{wide dkpointdag}=[wide kpointdag,doubled]
\tikzstyle{dkcopoint}=[kpoint adjoint,doubled]
\tikzstyle{dkpointadj}=[kpoint adjoint,doubled]
\tikzstyle{dkpointconj}=[kpoint conjugate,doubled]
\tikzstyle{dkpointtrans}=[kpoint transpose,doubled]

\tikzstyle{kscalar}=[kpoint common, shape=EBox, inner xsep=-1pt, inner ysep=3pt,font=\small]
\tikzstyle{kscalarconj}=[kpoint common, shape=WBox, inner xsep=-1pt, inner ysep=3pt,font=\small]

\tikzstyle{spekpoint}=[kpoint sc,minimum height=5mm,inner sep=3pt]
\tikzstyle{spekcopoint}=[kpoint adjoint sc,minimum height=5mm,inner sep=3pt]

\tikzstyle{dspekpoint}=[spekpoint,doubled]
\tikzstyle{dspekcopoint}=[spekcopoint,doubled]

 \tikzstyle{upground}=[circuit ee IEC,thick,ground,rotate=90,scale=2.5]
 \tikzstyle{downground}=[circuit ee IEC,thick,ground,rotate=-90,scale=2.5]
 \tikzstyle{bigground}=[circuit ee IEC,thick,ground,rotate=90,xscale=2.5,yscale=4.5]

\tikzstyle{arrs}=[-latex,font=\small,auto]
\tikzstyle{arrow plain}=[arrs]
\tikzstyle{arrow dashed}=[dashed,arrs]
\tikzstyle{arrow bold}=[very thick,arrs]
\tikzstyle{arrow hide}=[draw=white!0,-]
\tikzstyle{arrow reverse}=[latex-]
\tikzstyle{cdnode}=[]

%% file: minreps.tikzstyles
\tikzstyle{frame}=[fill=red, draw=black, shape=circle]
\tikzstyle{internal}=[fill=black, draw=black, shape=circle]
\tikzstyle{smallblk}=[circle, fill, inner sep=1pt]


%% file: Diagrams/discard.tikz
\begin{tikzpicture}
	\begin{pgfonlayer}{nodelayer}
		\node [style=none] (0) at (0, -1) {};
		\node [style=upground] (1) at (0, 0.75) {};
	\end{pgfonlayer}
	\begin{pgfonlayer}{edgelayer}
		\draw [qWire] (0.center) to (1);
	\end{pgfonlayer}
\end{tikzpicture}

%% file: Diagrams/tensorsystemstriv2.tikz
\begin{tikzpicture}
	\begin{pgfonlayer}{nodelayer}
		\node [style=none] (9) at (0, -1.25) {};
		\node [style=right label] (10) at (0, -0.75) {$A$};
		\node [style=none] (11) at (0, 1) {};
	\end{pgfonlayer}
	\begin{pgfonlayer}{edgelayer}
		\draw [qWire] (9.center) to (11.center);
	\end{pgfonlayer}
\end{tikzpicture}

%% file: Diagrams/11minrepsnake.tikz
\begin{tikzpicture}
	\begin{pgfonlayer}{nodelayer}
		\node [style=frame] (2) at (-0.75, 0.75) {};
		\node [style=frame] (3) at (0.75, -0.75) {};
		\node [style=internal] (4) at (-0.75, -0.75) {};
		\node [style=internal] (5) at (0.75, 0.75) {};
	\end{pgfonlayer}
	\begin{pgfonlayer}{edgelayer}
		\draw (2) to (4);
		\draw (4) to (5);
		\draw (5) to (3);
	\end{pgfonlayer}
\end{tikzpicture}

%% file: Diagrams/11minrep.tikz
\begin{tikzpicture}
	\begin{pgfonlayer}{nodelayer}
		\node [style=none] (1) at (0, 0.75) {};
		\node [style=frame] (2) at (0, 0.75) {};
		\node [style=frame] (3) at (0, -0.75) {};
	\end{pgfonlayer}
	\begin{pgfonlayer}{edgelayer}
		\draw (3) to (2);
	\end{pgfonlayer}
\end{tikzpicture}

%% file: Diagrams/identityLong.tikz
\begin{tikzpicture}
	\begin{pgfonlayer}{nodelayer}
		\node [style=none] (1) at (0, -1) {};
		\node [style=none] (2) at (0, 2) {};
		\node [style=none] (3) at (1, 1) {};
	\end{pgfonlayer}
	\begin{pgfonlayer}{edgelayer}
		\draw [qWire](1.center) to (2.center);
	\end{pgfonlayer}
\end{tikzpicture}

%% file: Diagrams/splitminrep22.tikz
\begin{tikzpicture}
	\begin{pgfonlayer}{nodelayer}
		\node [style=frame] (0) at (-1, 1) {};
		\node [style=frame] (1) at (1, 1) {};
		\node [style=frame] (2) at (-1, -1) {};
		\node [style=frame] (3) at (1, -1) {};
	\end{pgfonlayer}
	\begin{pgfonlayer}{edgelayer}
		\draw (1) to (3);
		\draw (0) to (2);
	\end{pgfonlayer}
\end{tikzpicture}

%% file: Diagrams/statefpo.tikz
\begin{tikzpicture}
	\begin{pgfonlayer}{nodelayer}
		\node [style=internal] (0) at (0, -0.75) {};
		\node [style=none] (1) at (1, 0.75) {};
		\node [style=none] (2) at (-1, 0.75) {};
	\end{pgfonlayer}
	\begin{pgfonlayer}{edgelayer}
		\draw (2.center) to (0);
		\draw (0) to (1.center);
	\end{pgfonlayer}
\end{tikzpicture}

%% file: Diagrams/copiedstatefpo.tikz
\begin{tikzpicture}
	\begin{pgfonlayer}{nodelayer}
		\node [style=internal] (0) at (-1, -0.75) {};
		\node [style=none] (2) at (-1, 0.75) {};
		\node [style=internal] (3) at (1, -0.75) {};
		\node [style=none] (4) at (1, 0.75) {};
	\end{pgfonlayer}
	\begin{pgfonlayer}{edgelayer}
		\draw (2.center) to (0);
		\draw (4.center) to (3);
	\end{pgfonlayer}
\end{tikzpicture}